\documentclass[]{lmcs}
\pdfoutput=1

\usepackage{lastpage}
\lmcsdoi{19}{3}{12}
\lmcsheading{}{\pageref{LastPage}}{}{}%
{Oct.~29,~2021}{Aug.~09,~2023}{}

\newenvironment{theorem}{\begin{thm}}{\end{thm}}
\newenvironment{lemma}{\begin{lem}}{\end{lem}}
\newenvironment{corollary}{\begin{cor}}{\end{cor}}
\newenvironment{proposition}{\begin{prop}}{\end{prop}}

\newenvironment{definition}{\begin{defi}}{\end{defi}}

\newenvironment{observation}{\begin{obs}}{\end{obs}}

\newenvironment{example}{\begin{exa}}{\end{exa}}

\usepackage{amsfonts}
\usepackage{amsmath}
\usepackage{amssymb}
\usepackage{stmaryrd} \usepackage{mathtools} 

\usepackage[ruled,linesnumbered,vlined]{algorithm2e} 
\SetKwComment{Comment}{$\vartriangleright$  }{} 

\usepackage[english]{babel}
\usepackage[T1]{fontenc}
\usepackage[utf8]{inputenc}

\usepackage{microtype} 

\usepackage{etoolbox} \usepackage{needspace} 

\usepackage{booktabs}
\usepackage{multicol}
\usepackage{multirow}

\usepackage{xspace}

\usepackage{tikz}
\usetikzlibrary{automata}
\usetikzlibrary{fit}
\usetikzlibrary{backgrounds}
\usetikzlibrary{calc}
\usetikzlibrary{decorations.pathreplacing,decorations.pathmorphing}
\usetikzlibrary{positioning}

 \newcommand{\eqdef}{\mathrel{{:}{=}}}
\newcommand{\df}{\eqdef}
\newcommand{\eqnumber}[1]{\overset{\mathclap{(#1)}}{=}}
\newcommand{\set}[1]{\mathord{\{#1\}}}

\makeatletter
\@ifpackageloaded{stix}{}{\DeclareFontEncoding{LS2}{}{\noaccents@}
  \DeclareFontSubstitution{LS2}{stix}{m}{n}
  \DeclareSymbolFont{stix@largesymbols}{LS2}{stixex}{m}{n}
  \SetSymbolFont{stix@largesymbols}{bold}{LS2}{stixex}{b}{n}
  \DeclareMathDelimiter{\lBrace}{\mathopen} {stix@largesymbols}{"E8}{stix@largesymbols}{"0E}
  \DeclareMathDelimiter{\rBrace}{\mathclose}{stix@largesymbols}{"E9}{stix@largesymbols}{"0F}
}
\makeatother
\newcommand{\multiset}[1]{\ensuremath{\lBrace #1 \rBrace}\xspace}

\newcommand{\join}{\bowtie}

\newcommand{\K}{\srd}
\newcommand{\srd}{\mathbb{K}}
\newcommand{\srplus}{\oplus}
\newcommand{\srtimes}{\otimes}
\newcommand{\srzero}{\overline{0}}
\newcommand{\srone}{\overline{1}}
\newcommand{\boplus}{\bigoplus}
\newcommand{\sr}{(\srd,\srplus,\srtimes,\srzero,\srone)}

\newcommand{\B}{\mathbb{B}}
\newcommand{\true}{\mathsf{true}}
\newcommand{\false}{\mathsf{false}}

\newcommand{\nat}{\ensuremath{\mathbb{N}}\xspace}

\newcommand{\rationals}{\ensuremath{\mathbb{Q}}\xspace}
\newcommand{\Q}{\rationals}

\newcommand{\tropical}{\ensuremath{\mathbb{T}}\xspace}
\newcommand{\numerical}{\ensuremath{\mathbb{Q}}\xspace}
\newcommand{\tropicalSR}{\ensuremath{\tropical = (\rationals \cup \{\infty\},\min, +, \infty,
    0)}\xspace}
\newcommand{\numericalSR}{\ensuremath{\numerical = (\rationals, +, \times, 0, 1)}\xspace}

\newcommand{\doc}{\ensuremath{d}\xspace}
\newcommand{\alphabet}{\Sigma}
\newcommand{\docs}{\ensuremath{\alphabet^*}\xspace}
\newcommand{\spanner}{\ensuremath{S}\xspace}

\newcommand{\toSpanner}[1]{\ensuremath{{\llbracket #1 \rrbracket}}}
\newcommand{\weight}[1]{{\mathsf{w}}_{#1}}
\newcommand{\repspnrk}[2]{\toSpanner{#1}_{#2}}
\newcommand{\repspnrw}[1]{\toSpanner{#1}_{\srd}}
\newcommand{\repspnr}[1]{\toSpanner{#1}}

\newcommand{\blank}{\ensuremath{\sqcup}}

\newcommand{\vals}{\ensuremath{\mathsf{D}}\xspace}
\newcommand{\vars}{\ensuremath{\mathsf{Vars}}\xspace}
\newcommand{\svars}{\vars}
\newcommand{\tup}{\ensuremath{\mathrm{t}}\xspace}
\newcommand{\tupu}{\mathrm{t^\prime}}
\newcommand{\projectTup}[2]{\ensuremath{\pi_{#2}(#1)}\xspace}
\newcommand{\emptytup}{\ensuremath{()}\xspace}

\newcommand{\spanFromTo}[2]{\ensuremath{[#1,#2\rangle}}
\newcommand{\spanij}{\spanFromTo{i}{j}\xspace}
\newcommand{\mspan}[2]{\spanFromTo{#1}{#2}}
\newcommand{\spans}{\mathsf{Spans}}

\newcommand{\lang}{\ensuremath{\mathcal{L}}\xspace}
\newcommand{\refWord}{\ensuremath{\mathrm{r}}\xspace}
\newcommand{\refWordPrime}{\ensuremath{\mathrm{r'}}\xspace}
\newcommand{\reflang}{\ensuremath{\mathcal{R}}\xspace}
\newcommand{\pre}{\mathrm{pre}}

\newcommand{\post}{\mathrm{post}}

\newcommand{\vset}{\text{VSet}}

\newcommand{\spannerClass}{\ensuremath{\mathcal S}\xspace}

\newcommand{\vsa}{\ensuremath{\mathrm{VSA}}\xspace}
\newcommand{\fvsa}{\ensuremath{\mathrm{VSA}}\xspace}
\newcommand{\ufvsa}{\ensuremath{\mathrm{uVSA}}\xspace}

\newcommand{\varop}[1]{\Gamma_{#1}}
\newcommand{\vop}[1]{\ensuremath{\triangleright_{#1}}}
\newcommand{\vcl}[1]{\ensuremath{\triangleleft_{#1}}}

\newcommand{\clr}{\ensuremath{\operatorname{doc}}\xspace}
\newcommand{\refWordFrom}[2]{\ensuremath{\operatorname{ref}(#1,#2)}\xspace}
\newcommand{\toTuple}[1]{\ensuremath{\operatorname{tup}(#1)}\xspace}

\newcommand{\fromRun}[1]{\ensuremath{\operatorname{ref}(#1)}\xspace}

\newcommand{\rn}{\rho}
\newcommand{\Rn}[2]{P({#1},{#2})}

\newcommand{\NL}{\ensuremath{\mathrm{NL}}\xspace}

\newcommand{\np}{\ensuremath{\mathrm{NP}}\xspace}

\newcommand{\conp}{\ensuremath{\mathrm{coNP}}\xspace}

\newcommand{\spanl}{\ensuremath{\mathrm{spanL}}\xspace}

\newcommand{\fp}{\ensuremath{\mathrm{FP}}\xspace}
\newcommand{\bpp}{\ensuremath{\mathrm{BPP}}\xspace}
\newcommand{\rp}{\ensuremath{\mathrm{RP}}\xspace}

\newcommand{\pitwop}{\ensuremath{\mathrm{\Pi_2^P}}\xspace}
\newcommand{\sigmatwop}{\ensuremath{\mathrm{\Sigma_2^P}}\xspace}
\newcommand{\ph}{\ensuremath{\mathrm{PH}}\xspace}
\newcommand{\fptosp}{\ensuremath{\mathrm{FP^{\#P}}}\xspace}
\newcommand{\sharpp}{\ensuremath{\mathrm{\#P}}\xspace}
\newcommand{\optp}{\ensuremath{\mathrm{OptP}}\xspace}

\newcommand{\fpras}{\ensuremath{\mathrm{FPRAS}}\xspace}
\newcommand{\fexptime}{\ensuremath{\mathrm{FEXPTIME}}\xspace}

\newcommand{\vtup}[1]{{#1}\text{-}\mathsf{Tup}}
\newcommand{\Tup}{\mathsf{Tup}}

\newcommand{\len}{\text{len}\xspace}

\newcommand{\supp}{\mathsf{Supp}}

\newcommand{\UREG}{\textsc{UReg}\xspace}
\newcommand{\REG}{\textsc{Reg}\xspace}

\newcommand{\reg}[1]{\ensuremath{\REG_{#1}}\xspace}
\newcommand{\ureg}[1]{\ensuremath{\UREG_{#1}}\xspace}

\newcommand{\regk}{\reg{\K}\xspace}
\newcommand{\uregk}{\ureg{\K}\xspace}
\newcommand{\regtrop}{\reg{\tropical}}
\newcommand{\uregtrop}{\ureg{\tropical}}
\newcommand{\regnum}{\reg{\numerical}}
\newcommand{\uregnum}{\ureg{\numerical}}

\newcommand{\w}{\ensuremath{w}\xspace}
\newcommand{\wCodom}{\text{Img}}
\newcommand{\wC}{\ensuremath{\mathcal{W}}\xspace}
\newcommand{\cW}{\wC}

\newcommand{\xvars}{\ensuremath{X}\xspace}

\newcommand{\SW}{\textsc{CWidth}\xspace}
\newcommand{\PW}{\textsc{Poly}\xspace}

\newcommand{\supts}{\ensuremath{\mathrm{S}}\xspace}
\newcommand{\suptm}{\ensuremath{\mathcal{S}}\xspace}

\newcommand{\suptsp}{\ensuremath{\pi_\xvars(\toSpanner{A})(\doc)\xspace}}
\newcommand{\suptmp}{\ensuremath{\suptm_{A,\doc}}\xspace}

\newcommand{\spcount}{\ensuremath{\operatorname{Count}}\xspace}
\newcommand{\spsum}{\ensuremath{\operatorname{Sum}}\xspace}
\newcommand{\spavg}{\ensuremath{\operatorname{Avg}}\xspace}
\newcommand{\spmin}{\ensuremath{\operatorname{Min}}\xspace}
\newcommand{\spmax}{\ensuremath{\operatorname{Max}}\xspace}
\newcommand{\quantl}{\ensuremath{\operatorname{Quantile}}\xspace}

\newcommand{\qquantl}[1]{\ensuremath{#1\text{-}\!\quantl}}
\newcommand{\spquant}{\qquantl{q}}

\newcommand{\spcountagg}{\ensuremath{\spcount(\spanner,\doc)}\xspace}
\newcommand{\spsumagg}{\ensuremath{\spsum(\spanner,\doc,\w)}\xspace}
\newcommand{\spavgagg}{\ensuremath{\spavg(\spanner,\doc,\w)}\xspace}
\newcommand{\spmaxagg}{\ensuremath{\spmax(\spanner,\doc,\w)}\xspace}
\newcommand{\spminagg}{\ensuremath{\spmin(\spanner,\doc,\w)}\xspace}
\newcommand{\spqquantagg}[1]{\ensuremath{\qquantl{#1}(\spanner,\doc,\w)}\xspace}
\newcommand{\spquantagg}{\spqquantagg{q}}

\newcommand{\spAcountagg}{\ensuremath{\spcount(\toSpanner{A},\doc)}\xspace}
\newcommand{\spAsumagg}{\ensuremath{\spsum(\toSpanner{A},\doc,\w)}\xspace}
\newcommand{\spAavgagg}{\ensuremath{\spavg(\toSpanner{A},\doc,\w)}\xspace}
\newcommand{\spAmaxagg}{\ensuremath{\spmax(\toSpanner{A},\doc,\w)}\xspace}
\newcommand{\spAminagg}{\ensuremath{\spmin(\toSpanner{A},\doc,\w)}\xspace}
\newcommand{\spAqquantagg}[1]{\ensuremath{\qquantl{#1}(\toSpanner{A},\doc,\w)}\xspace}
\newcommand{\spAquantagg}{\spAqquantagg{q}}

\newcommand{\spAcountaggapx}{\ensuremath{\spcount(\toSpanner{A},\doc,\delta)}\xspace}
\newcommand{\spAsumaggapx}{\ensuremath{\spsum(\toSpanner{A},\doc,\w,\delta)}\xspace}
\newcommand{\spAavgaggapx}{\ensuremath{\spavg(\toSpanner{A},\doc,\w,\delta)}\xspace}

\newcommand{\countp}{\textsc{Count}\xspace}
\newcommand{\aggp}{\textsc{Agg}\xspace}
\newcommand{\sump}{\textsc{Sum}\xspace}
\newcommand{\avgp}{\textsc{Average}\xspace}
\newcommand{\qquantp}[1]{#1\text{-}\textsc{Quantile}\xspace}
\newcommand{\quantp}{\ensuremath{q}\text{-}\textsc{Quantile}\xspace}
\newcommand{\minp}{\textsc{Min}\xspace}
\newcommand{\maxp}{\textsc{Max}\xspace}

\newcommand{\tsum}{\mathrm{sum}}

\newcommand{\propab}[1]{\text{Pr}(#1)}
\newcommand{\bigprobab}[1]{\ensuremath{\text{Pr}\Big( #1 \Big)}}

\newcommand{\dags}{\ensuremath{src}\xspace}
\newcommand{\dagt}{\ensuremath{snk}\xspace}
\newcommand{\dagpaths}{\ensuremath{\text{Paths}(\dags,\dagt)}\xspace}

\newlength\boxwidth
\newlength\questionwidth

\newcommand{\computeproblemWidth}[4]{
  \setlength\boxwidth{#1\linewidth}{
    \setlength\questionwidth{\boxwidth}\addtolength\questionwidth{-2cm}{
    \begin{center}
      \fbox{\parbox[t]{\boxwidth}{\centerline{#2}
          \vspace{1ex}
          \begin{tabular}{lp{\questionwidth}}
            Input: & #3\\[1pt]
            Task: & #4
          \end{tabular}}}
    \end{center}}}}

\newcommand{\enc}[1]{\|#1\|}


\begin{document}

\title[Aggregate Queries on Extractions by Regular Expressions]{The Complexity of Aggregates over Extractions by Regular Expressions{\rsuper*}}
\titlecomment{{\lsuper*}A short version of this article has been published in a conference proceedings~\cite{DoleschalBKM21}.}

\author[J. Doleschal]{Johannes Doleschal\lmcsorcid{0000-0002-7045-7298}}[a,c]

\author[B. Kimelfeld]{Benny Kimelfeld\lmcsorcid{0000-0002-7156-1572}}[b]

\author[W. Martens]{Wim Martens\lmcsorcid{0000-0001-9480-3522}}[a]

\address{University of Bayreuth, Germany}
\email{johannes.doleschal@uni-bayreuth.de, wim.martens@uni-bayreuth.de}
\address{Technion, Haifa, Israel}
\email{bennyk@cs.technion.ac.il}
\address{Hasselt University, Belgium}

\begin{abstract}
  Regular expressions with capture variables, also known as \textquotedblleft
  regex-formulas,\textquotedblright{} extract relations of spans (intervals
  identified by their start and end indices) from text. In turn, the class of
  regular document spanners is the closure of the regex formulas under the
  Relational Algebra. We investigate the computational complexity of querying
  text by aggregate functions, such as sum, average, and quantile, on top of
  regular document spanners. To this end, we formally define aggregate functions
  over regular document spanners and analyze the computational complexity of
  exact and approximate computation. More precisely, we show that in a
  restricted case, all studied aggregate functions can be computed in polynomial
  time. In general, however, even though exact computation is intractable, some
  aggregates can still be approximated with fully polynomial-time randomized
  approximation schemes (\fpras).
\end{abstract}

\maketitle

\section{Introduction}
\label{intro}
Information extraction commonly refers to the task of extracting structured
information from text. A document spanner (or just spanner for short) is an
abstraction of an information extraction program: it states how to transform a
document into a relation over its spans. More formally, a \emph{document} is a
string $\doc$ over a finite alphabet, a \emph{span} of $\doc$ represents a
substring of $\doc$ by its start and end positions, and a \emph{spanner} is a
function that maps every document $\doc$ into a relation over the spans of
$\doc$~\cite{FaginKRV15}. The spanner framework has originally been introduced
as the theoretical basis underlying IBM's SQL-like rule system for information
extraction, namely SystemT~\cite{KrishnamurthyLRRVZ08, LiRC11}. The most studied
spanner instantiation is the class of \emph{regular spanners}, which is the closure of
regex formulas (regular expressions with capture variables) under the standard
operations of the relational algebra (projection, natural join, union, and
difference). Equivalently, the regular spanners are the ones expressible as
\emph{variable-set automata} (VSet-automata for short), which are nondeterministic
finite-state automata that can open and close capture variables. These spanners
extract from the text relations wherein the capture variables are the
attributes.

While regular spanners and natural generalizations thereof are the basis of
rule-based systems for text analytics, they are also used implicitly in other
types of systems, and particularly ones based on statistical models and machine
learning. Rules similar to regular spanners are used for \emph{feature
  generators} of graphical models (e.g., Conditional Random
Fields)~\cite{LiBC04,SuttonM12}, \emph{weak constraints} of Markov Logic
Networks~\cite{PoonD07} and extensions such as DeepDive~\cite{ShinWWSZR15}, and
the generators of \emph{noisy training data} (``labeling functions'') in the
state-of-the-art Snorkel system~\cite{RatnerBEFWR17}. Further connections to
regular spanners can potentially arise from efforts to express artificial neural
networks for natural language processing as finite-state automata~\cite{MayrY18,
  MichalenkoSVBCP19, WeissGY18}. The computational complexity of evaluating
regular spanners has been well studied from various angles, including the data
and combined complexity of answer enumeration~\cite{AmarilliBMN19,
  FlorenzanoRUVV18, FreydenbergerKP18, MaturanaRV18}, the cost of combining
spanners via relational algebra operators~\cite{PeterfreundFKK19} and recursive
programs~\cite{PeterfreundCFK19}, their dynamic
complexity~\cite{FreydenbergerT20}, evaluation in the presence of weighted
transitions~\cite{DoleschalKMP-lmcs22}, and the ability to distribute their
evaluation over fragments of the document~\cite{DoleschalKMNN19}.

In this article, we study the computational complexity of evaluating
\emph{aggregate functions} over regular spanners. These are queries that map a
document $\doc$ and a spanner $\spanner$ into a number $\alpha(\spanner(\doc))$,
where $\spanner(\doc)$ is the relation obtained by applying $\spanner$ to $\doc$
and $\alpha$ is a standard aggregate function: Count, Sum, Average, Min, Max, or
Quantile. There are various scenarios where queries that involve aggregate
functions over spanners can be useful. For example, such queries arise in the
extraction of statistics from textual resources like medical
publications~\cite{NordonKSKSR19} and news reports~\cite{SchumakerC09}. As
another example, when applying advanced text search or protein/DNA motif
matching using regular expressions~\cite{CockwellG89, NeuwaldG94}, the search
engine typically provides the (exact or approximate) number of answers, and we
would like to be able to compute this number without actually computing the
answers, especially when the number of answers is prohibitively large. Finally,
when programming feature generators or labeling functions in extractor
development, the programmer is likely to be interested in aggregate statistics
and summaries for the extractions (e.g., to get a holistic view of what is being
extracted from the dataset, such as quantiles over extracted ages and so on),
and again, we would like to be able to estimate these statistics faster than it
takes to materialize the entire set of answers.

Our main objective in this work is to understand when it is tractable to compute
$\alpha(\spanner(\doc))$. This question raises closely related questions that we
will discuss, such as when the materialization of the set of intermediate
results $\spanner(\doc)$ (which can be exponentially large) can be avoided. Furthermore, when the
exact computation of $\alpha(\spanner(\doc))$ is intractable, we study whether
it can be approximated.

At the technical level, each aggregate function (with the exception of Count)
requires a specification of how an extracted tuple of spans represents a number.
For example, the number $21$ can be represented by the span of the string
``21'', ``21.0'', ``twenty one'', ``twenty first'', ``three packs of seven'' and
so on. To abstract away from specific textual representations of numbers, we
consider several means of assigning weights to tuples. To this end, we assume
that a (representation of a) \emph{weight function} $\w$, which maps every tuple
of $\spanner(\doc)$ into a number, is part of the input of the aggregate
functions. Hence, the general form of the aggregate query we study is
$\alpha(\spanner,\doc,\w)$. The direct approach to evaluating $\alpha(\spanner,\doc,\w)$ is to
compute $\spanner(\doc)$, apply $\w$ to each tuple, and apply $\alpha$ to the resulting
sequence of numbers. This approach works well if the number of tuples in $\spanner(\doc)$
is manageable (e.g., bounded by some polynomial). However, the number of tuples
in $\spanner(\doc)$ can be exponential in the number of variables of $\spanner$, and so, the
direct approach takes exponential time in the worst case. We will identify
several cases in which $\spanner(\doc)$ is exponential, yet $\alpha(\spanner(\doc))$
can be computed in polynomial time.

It is not very surprising that, at the level of generality we adopt,
each of the aggregate functions is intractable (\#P-hard) in general.
Hence, we focus on several assumptions 
that can potentially reduce the inherent hardness of evaluation:
\begin{itemize}
\item Restricting the range of weight functions to positive numbers;
\item Restricting to weight functions that are determined by a single span
  or defined by (unambiguous) weighted VSet-automata;
\item Restricting to spanners that are represented by an unambiguous
  variant of VSet-automata;
\item Allowing for a randomized approximation (FPRAS, i.e., fully
  polynomial randomized approximation schemes).
\end{itemize}
Our analysis shows which of these assumptions brings the complexity down to
polynomial time, and which is insufficient for tractability. Importantly, we
derive an interesting and general tractable case for each of the aggregate
functions we study.

To the best of our knowledge, counting the number of tuples extracted by a
\vset-automaton (i.e., the Count aggregate function) is the only aggregation
function for document spanners which has been studied in
literature, except for Doleschal et al.~\cite{DoleschalKMP-lmcs22} who consider a variation of maximum aggregation. (Given a weighted
  \vset-automaton and a document, they study the computational complexity of
  returning a tuple with maximal weight.) Concerning counting, Florenzano et
al.~\cite{FlorenzanoRUVV18} study the problem of counting the number of
extractions of a \vset-automaton and approximation thereof is studied by Arenas
et al.~\cite{ArenasCJR19}. To be specific, Arenas et al.~\cite{ArenasCJR19} give
a polynomial-time uniform sampling algorithm from the space of words which are
accepted by an NFA and have a given length. Using that sampling, they establish
an \fpras for the Count aggregate function. Our \fpras results are based on
their results. We explain the connection between the
known results and our work in more detail throughout the article. Yet, to the best of our knowledge, this work is the first to consider aggregate functions over numerical values
extracted by document spanners.

\subsection*{Comparison to the Conference Version.} Compared to the conference version of this article~\cite{DoleschalBKM21}, the following aspects are new. We now consider constant-width weight functions, which generalize the single-variable weight functions from \cite{DoleschalBKM21};
Section \ref{sec:weight-function-expressiveness} is new; and we provide a more detailed complexity overview for regular weight functions over different semirings. Furthermore, proofs that were missing in \cite{DoleschalBKM21} are now included. On a technical level, we now use parsimonious reductions and metric reductions instead of Turing reductions for some of the results, which strengthens them.

\subsection*{Organization}
This article is organized as follows. In Section~\ref{preliminaries}, we give
preliminary definitions and notation. We summarize the main results in Section~\ref{main-results} and expand on these results in the later
sections. In Section~\ref{technical} we give some preliminary results. We
describe our investigation for constant-width weight functions, polynomial-time
weight functions and regular weight functions in Sections~\ref{cwidth},
\ref{pvar} and~\ref{reg}, respectively. Finally, we study approximation in Section~\ref{approx} and conclude in Section~\ref{conclusions}.

 \section{Preliminaries}
\label{preliminaries}
We define here the main concepts and notation that we will use throughout the article.

\subsection{Sets, Multisets, and Semirings}

The cardinality of a set $A$ is denoted by $|A|$. A \emph{multiset} over $A$ is
a function $M : A \to \nat$. We call $M(a)$ the \emph{multiplicity} of $a$ in
$M$ and say that $a \in M$ if $M(a) > 0$. The \emph{size} of $M$ denoted $|M|$,
is the sum of the multiplicities of all elements in $A$, that is $\sum_{a \in A}
M(a)$. Note that $|M|$ may be infinite. We denote multisets in brackets
$\multiset{\cdot}$ in the usual way. For example, in the multiset $M =
\multiset{1, 1, 3}$ we have $M(1) = 2$, $M(3) = 1$, and $|M| = 3$.
Furthermore, given a set $X$, we denote by $2^X$ the powerset of $X$, i.e., the set of all subsets of $X$.

A \emph{commutative monoid} $(\mathbb{M}, \ast, \text{id})$ is an algebraic
structure consisting of a set $\mathbb{M}$, a binary operation $\ast$ and an
element $\text{id} \in \mathbb{M}$, such that:
\begin{enumerate}
\item $\ast$ is associative, i.e., $(a \ast b) \ast c = a \ast (b \ast c)$ for
  all $a,b,c \in \mathbb{M}$, 
\item $\ast$ is commutative, i.e. $a \ast b = b \ast a$ for all $a,b \in
  \mathbb{M}$, and
\item $\text{id}$ is an identity, i.e., $\text{id} \ast a = a$ for all $a \in \mathbb{M}$.
\end{enumerate}
A \emph{commutative semiring} $\sr$ is an algebraic structure consisting of a
set $\srd$, containing two elements: the \emph{zero} element $\srzero$ and the
\emph{one} element $\srone$. Furthermore, it is equipped with two binary
operations, namely \emph{addition} $\srplus$ and \emph{multiplication}
$\srtimes$ such that:
\begin{enumerate}
\item $(\srd, \srplus, \srzero)$ and $(\srd, \srtimes, \srone)$ are commutative monoids,
\item multiplication distributes over addition, that is, $(a\srplus b) \srtimes
  c = (a\srtimes c) \srplus (b\srtimes c) $, for all $a,b,c \in \srd$, and
\item $\srzero$ is absorbing for $\srtimes$, that is, $\srzero \srtimes a =
  \srzero$ for all $a \in \srd$.
\end{enumerate}

\begin{example}\label{ex:semirings}The following are commutative semirings.
  \begin{enumerate}
  \item The \emph{numeric semiring} $(\Q,+,\cdot,0,1)$ over the rationals, with
    the usual addition and multiplication operators.
\item The \emph{Boolean semiring} $(\B, \vee,\wedge, \false,\true )$ where $\B
    \df \{ \false, \true\}$.
  \item The \emph{tropical semiring} $(\tropical,\min, +, \infty, 0)$ where
    $\tropical \df \Q \cup \{\infty\}$ and $\min$ stands for the binary minimum
    function.\qed
  \end{enumerate}
\end{example}

\subsection{Document Spanners and $\K$-Annotators}
This article is within the formalism of \emph{document spanners} by Fagin et
al.~\cite{FaginKRV15,FaginKRV15-sigrecord}. More precisely, we use the notion
of $\K$-annotators, as introduced by Doleschal et
al.~\cite{DoleschalKMP-lmcs22}, which enables document spanners to annotate
provenance information. 
Next, we revisit the definitions of these concepts. We
assume countably infinite and disjoint sets $\vals$ and $\svars$, containing
\emph{data values} (or simply \emph{values}) and \emph{variables}, respectively.

\subsubsection*{Documents and Spans}
Let $\Gamma$  be a finite set, disjoint from $\vals$ and $\vars$, of symbols. We refer to $\Gamma$ as an \emph{alphabet}. 
A sequence $s = \sigma_1 \cdots \sigma_n$ of
symbols where every $\sigma_i \in \Gamma$ is a \emph{string} over the set $\Gamma$. 
If $n=0$ we denote $s$ by $\varepsilon$ and call $s$
\emph{empty}. By $\Gamma^*$ we denote the set of all strings over $\Gamma$. We
denote by $|s|$ the length $n$ of a string $s \in \Gamma^*$.
In the context of Information Extraction, we will restrict ourselves to a subset of symbols of $\Gamma$, which we will always denote as $\alphabet$. We typically use the letter $\doc$ (and indexed variations thereof) to denote strings over $\alphabet$ and refer to such strings as \emph{documents}.

A \emph{span} of $\doc$ is an expression of the form \spanij with $1 \leq i \leq
j \leq n+1$. For a span \spanij of \doc, we denote by $\doc_{\spanij}$ the
string $\sigma_i \cdots \sigma_{j-1}$. A span $\spanij$ is empty if $i=j$ which
implies that $\doc_{\spanij}=\varepsilon$. Two spans $\spanFromTo{i_1}{j_1}$ and $\spanFromTo{i_2}{j_2}$ are \emph{equal}
if $i_1 = i_2$ and $j_1 = j_2$. In particular, we observe that two spans do not
have to be equal if they select the same string. That is,
$\doc_{\spanFromTo{i_1}{j_1}} = \doc_{\spanFromTo{i_2}{j_2}}$ does not imply
that $\spanFromTo{i_1}{j_1} = \spanFromTo{i_2}{j_2}$.
For a document \doc, we denote by
$\spans(\doc)$ the set of all possible spans of \doc and by $\spans$ the set of
all possible spans of all possible documents. 

\subsubsection*{$\srd$-relations and $\srd$-annotators}
Let
$V\subseteq \svars$ be a finite set of variables.
A $V$-\emph{tuple} is a function $\tup : V \to \vals$ that assigns values to
variables in $V$. We sometimes leave $V$ implicit when the precise set is not
important.
For such a tuple $\tup$, we denote the set $V$ by $\vars(\tup)$.
We denote the set of all $V$-tuples by $\vtup{V}$.  For a subset $X
\subseteq \svars$, we denote the restriction of $\tup$ to the variables in $X$
by $\pi_X(\tup)$ or simply $\pi_X\tup$.  We say that a tuple $\tup$ is \emph{empty}, denoted by $\tup =
\emptytup$, if $\vars(\tup) = \emptyset$. 

 A \emph{$\srd$-relation $R$ over $V$}
is a function $R: \vtup{V} \rightarrow \srd$ such that its \emph{support},
defined by $\supp{(R)} \df \{\tup \mid R(\tup) \neq \srzero \}$, is finite. We
will also write $\tup \in R$ to abbreviate $\tup \in \supp{(R)}$. Furthermore,
we say that two $\srd$-relations $R_1$ and $R_2$ are \emph{disjoint} if
$\supp(R_1) \cap \supp(R_2) = \emptyset$. The \emph{size} of a $\srd$-relation $R$ is
the size of its support, that is, $|R| \df |\supp(R)|$. The \emph{arity}
of a $V$-tuple $\tup$ is the cardinality $|V|$ of $V$ and, similarly, the
\emph{arity} of a $\srd$-relation over $V$ is $|V|$.

The framework focuses on functions
that extract spans from documents and assigns them to variables. Since we will
be working with relations over spans, also called \emph{span relations}, we
assume that $\vals$ is such that $\spans \subseteq \vals$. A \emph{$\doc$-tuple}
$\tup$ is a $V$-tuple which only assigns values from $\spans(\doc)$, that is,
$\tup(x) \subseteq \spans(\doc)$ for every $x \in \vars(\tup)$. If the document
$\doc$ is clear from the context, we sometimes say simply \emph{tuple} instead
of $\doc$-tuple. We denote by $\doc_{\tup}$ the tuple
$(\doc_{\tup(x_1)}, \ldots , \doc_{\tup(x_n)})$, where $\vars(\tup) =
\{x_1,\ldots,x_n\}$.

A \emph{$\K$-weighted span relation over document $\doc$ and variables $V$} is a
$\srd$-relation $R$ wherein every tuple is a $\doc$-tuple $\tup$ with $\vars(\tup) =
V$. We also denote $V$ by $\vars(R)$.
A \emph{$\K$-weighted string relation} is a
$\K$-relation $R$ wherein every tuple $\tup \in R$ assigns strings, that is,
$\tup(x) \in \docs$ for every variable $x \in \vars(\tup)$.
Note that we can associate a string relation to every span relation over a
document $\doc$ by replacing every span $\spanij$ with the string $\doc_{\spanij}$.

\def\n#1{\textsf{\tiny{#1}}}
\def\s#1{\texttt{#1}}
\begin{figure}[t]
  \centering\small
  {
    \setlength{\tabcolsep}{.3ex}
    \begin{tabular}{
          cccccccccc
          cccccccccc
          cccccccccc
          cccccccccc
          c
          }
          $\s{T}$ & $\s{h}$ & $\s{e}$  & $\s{r}$ & $\s{e}$ & $\s{\blank}$
          & $\s{a}$ & $\s{r}$ & $\s{e}$ & $\s{\blank}$                                                  
          & $\s{7}$ & $\s{\blank}$ & $\s{e}$ & $\s{v}$ & $\s{e}$ & $\s{n}$ & $\s{t}$ & $\s{s}$ & $\s{\blank}$
          & $\s{i}$ & $\s{n}$ & $\s{\blank}$  & $\s{B}$ & $\s{e}$ & $\s{l}$ & $\s{g}$ & $\s{i}$ & $\s{u}$ & $\s{m}$  & $\s{,}$  & $\s{\blank}$
          & $\s{1}$ & $\s{0}$ & $\s{-}$ & $\s{1}$ & $\s{5}$ & $\s{\blank}$ & $\s{i}$ & $\s{n}$ & $\s{\blank}$
          & 
          \\\cmidrule{1-40}
          $\n{1}$ & $\n{2}$ & $\n{3}$  & $\n{4}$ & $\n{5}$ & $\n{6}$ & $\n{7}$ & $\n{8}$ & $\n{9}$ & $\n{10}$
          & $\n{11}$ & $\n{12}$ & $\n{13}$ & $\n{14}$ & $\n{15}$ & $\n{16}$ & $\n{17}$ & $\n{18}$ & $\n{19}$ & $\n{20}$
          & $\n{21}$ & $\n{22}$ & $\n{23}$ & $\n{24}$ & $\n{25}$  & $\n{26}$ & $\n{27}$ & $\n{28}$ & $\n{29}$ & $\n{30}$
          & $\n{31}$ & $\n{32}$ & $\n{33}$ & $\n{34}$ & $\n{35}$ & $\n{36}$ & $\n{37}$ & $\n{38}$ & $\n{39}$ & $\n{40}$
          & \\\addlinespace[1.5\defaultaddspace]
          $\s{F}$ & $\s{r}$ & $\s{a}$ & $\s{n}$ & $\s{c}$ & $\s{e}$ & $\s{,}$ & $\s{\blank}$
          & $\s{4}$ & $\s{\blank}$ & $\s{i}$ & $\s{n}$ & $\s{\blank}$
          & $\s{L}$ & $\s{u}$ & $\s{x}$ & $\s{e}$ & $\s{m}$ & $\s{b}$ & $\s{o}$ & $\s{u}$ & $\s{r}$ & $\s{g}$ & $\s{,}$ & $\s{\blank}$ 
          & $\s{t}$ & $\s{h}$ & $\s{r}$ & $\s{e}$ & $\s{e}$ & $\s{\blank}$ & $\s{i}$ & $\s{n}$ & $\s{\blank}$
          & $\s{B}$ & $\s{e}$ & $\s{r}$ & $\s{l}$ & $\s{i}$ & $\s{n}$ & $\s{.}$ 
          \\\midrule
          $\n{41}$ & $\n{42}$ & $\n{43}$ & $\n{44}$ & $\n{45}$ & $\n{46}$ & $\n{47}$ & $\n{48}$ & $\n{49}$ & $\n{50}$
          & $\n{51}$ & $\n{52}$ & $\n{53}$ & $\n{54}$ & $\n{55}$ & $\n{56}$ & $\n{57}$ & $\n{58}$ & $\n{59}$ & $\n{60}$
          & $\n{61}$ & $\n{62}$ & $\n{63}$ & $\n{64}$ & $\n{65}$ & $\n{66}$ & $\n{67}$ & $\n{68}$ & $\n{69}$ & $\n{70}$
          & $\n{71}$ & $\n{72}$ & $\n{73}$ & $\n{74}$ & $\n{75}$ & $\n{76}$ & $\n{77}$ & $\n{78}$ & $\n{79}$ & $\n{80}$
          & $\n{81}$                                        
    \end{tabular}  
  }
  
  \medskip
  
   \begin{tabular}[t]{ll}
    \toprule
    $x_{\textsf{loc}}$ & $x_{\textsf{events}}$\\
    \midrule
    $\mspan{23}{30}$ & $\mspan{11}{12}$  \\
    $\mspan{41}{47}$ & $\mspan{32}{37}$  \\
    $\mspan{54}{64}$ & $\mspan{49}{50}$  \\
    $\mspan{75}{81}$ & $\mspan{66}{71}$  \\
    \bottomrule
   \end{tabular}
   \qquad
  \begin{tabular}[t]{ll}
    \toprule
    $\doc_{x_{\textsf{loc}}}$ & $\doc_{x_{\textsf{events}}}$\\
    \midrule
    Belgium & 7 \\
    France & 10-15 \\
    Luxembourg & 4 \\
    Berlin & three \\
    \bottomrule
  \end{tabular}
  \qquad
   \caption{A document $\doc$ (top), a span relation $R$ (bottom left) and the
     corresponding string relation (bottom right).
     \label{fig:MainExample}}
\end{figure}

\begin{example}\label{ex:MainExample}Consider the document in Figure~\ref{fig:MainExample}. The table on the bottom
  left depicts a ($\B$-weighted) span relation $R$, encoding a possible extraction of locations
  with the corresponding number of events. The string relation at
  the bottom right is the corresponding string relation.\qed
\end{example}

\begin{definition}A \emph{$\srd$-annotator} (or \emph{annotator} for short) is a function $\spanner$
  that is associated to a finite set $V\subseteq \svars$ of variables and maps
  each document $\doc$ into a $\srd$-weighted span relation over $V$. We denote $V$ by
  $\vars(\spanner)$. We sometimes also refer to a $\srd$-annotator as an
  \emph{annotator over $\srd$} when we want to emphasize the semiring.
\end{definition}
\begin{example}
  As an example of a $\srd$-weighted annotator, consider again the setting in Example~\ref{ex:MainExample}. A $\Q$-weighted annotator in this setting is the function $S$ that maps each document $d$ to the span relation $R$ in which the tuples are pairs, consisting of a name of a country and a number (or numeric range), and in which the weight associated to each tuple is the smallest value in the numeric range. An example of such a tuple for the document in Figure~\ref{fig:MainExample} would be $\tup_1$ with $\tup_1(x_{\textsf{loc}}) = \mspan{23}{30}$ (the span of ``Belgium'') and $\tup_1(x_\textsf{events}) = \mspan{11}{12}$ (the span of ``7'').  Another example would be $\tup_2$ with $\tup_2(x_{\textsf{loc}}) = \mspan{41}{47}$ (the span of ``France'') and $\tup_1(x_\textsf{events}) = \mspan{32}{37}$ (the span of ``10-15''). The relation $R$ would assign $R(\tup_1) = 7$ and $R(\tup_2) = 10$.
\qed\end{example}

We say that two $\srd$-annotators $S_1$ and $S_2$ are \emph{disjoint} if, for
every document $\doc \in \docs$, the $\srd$-relations $S_1(\doc)$ and
$S_2(\doc)$ are disjoint. Furthermore, we denote by $\spanner=\spanner'$ the
fact that $\spanner$ and $\spanner'$ define the same function.

Notice that $\B$-annotators, i.e., annotators over the Boolean semiring
are simply the \emph{functional document spanners} as defined by Fagin et
al.~\cite{FaginKRV15, FaginKRV15-sigrecord}. Throughout this article, we refer to
$\B$-annotators as \emph{document spanners} (also \emph{spanner} for short).

\subsection{Algebraic Operators on \texorpdfstring{$\srd$-Relations}{K-Relations} and \texorpdfstring{$\srd$-Annotators}{K-Annotators}}\label{sec:projection}
Green et al.~\cite{GreenKT07} defined a set of operators on $\srd$-relations
that naturally correspond to relational algebra operators and map
$\srd$-relations to $\srd$-relations. As in much of the work on semirings in provenance, they do not consider the \emph{difference}
  operator (which would require additive inverses). More precisely, they define the algebraic
operators \emph{union},
\emph{projection}, and \emph{natural join} for all finite sets $V_1, V_2
\subseteq \svars$ and for all $\srd$-relations $R_1$ over $V_1$ and $R_2$ over
$V_2$, as follows.
\begin{itemize}
\item \textbf{Union}: If $V_1 = V_2$ then the union $R \df R_1 \cup R_2$ is a
  function $R : \vtup{V_1} \rightarrow \srd$ defined by $R(\tup) \df R_1(\tup)
  \srplus R_2(\tup)$. (Otherwise, the union is not defined.)
\item \textbf{Projection}: For $X\subseteq V_1$, the projection $R\df \pi_X(R_1)$
  is a function $R: \vtup{X} \rightarrow \srd$ defined by
 	\[
    R(\tup) \df \boplus_{\tup = \pi_X(\tup^\prime) \text{ and
      }R_1(\tup^\prime)\ne \srzero} R_1(\tup^\prime)\;.
  \]
\item \textbf{Natural Join:} The natural join $R \df R_1 \join R_2 $ is a
  function $R : \vtup{(V_1 \cup V_2)} \rightarrow \srd$ defined by
  \[ 
    R(\tup)  \df R_1(\pi_{V_1}(\tup)) \srtimes R_2(\pi_{V_2}(\tup))\;.
  \]
\end{itemize}

\begin{proposition}[{Green et al.~\cite{GreenKT07}}]\label{prop:Kalgebra}The above operators preserve the finiteness of the supports. Therefore, they
  map $\srd$-relations into $\srd$-relations.
\end{proposition}
Hence, we obtain an algebra on $\srd$-relations.

We now lift the relational algebra operators on $\K$-relations to the level of
$\K$-annotators. For all documents $\doc$ and for all annotators $\spanner_1$ and $\spanner_2$
associated with $V_1$ and $V_2$, respectively, we define the following:
\begin{itemize}
	\item \textbf{Union:} If $V_1= V_2$ then the union $\spanner \df \spanner_1 \cup \spanner_2$ is
    defined by $\spanner(\doc) \df \spanner_1 (\doc)\cup \spanner_2(\doc)$.\footnote{Here, $\cup$
      stands for the union of two $K$-relations as was defined previously. The
      same is valid also for the other operators.} 
	\item \textbf{Projection:} For $X \subseteq V_1$, the projection $\spanner \df \pi_X
    \spanner_1$ is defined by $\spanner(\doc) \df \pi_X \spanner_1(\doc)$. 
	\item \textbf{Natural Join:} The natural join $\spanner \df \spanner_1 \join \spanner_2 $ is
    defined by $\spanner(\doc) \df  \spanner_1(\doc) \join \spanner_2(\doc)$.
\end{itemize}
Due to Proposition~\ref{prop:Kalgebra}, it follows that the above operators form
an algebra on $\srd$-annotators.

\subsection{Ref-Words}
We use \emph{weighted \vset-automata} (or simply \emph{\vset-automata} for the
Boolean semiring) in order to represent $\srd$-annotators. Following
Freydenberger~\cite{Freydenberger19}, we introduce so-called \emph{ref-words},
which connect spanner representations with regular languages. We also introduce
unambiguous and functional \vset-automata, which have properties essential to
the tractability of some problems we study.

For a finite set $V \subseteq \svars$ of variables, ref-words are defined over
the extended alphabet $\alphabet \cup \varop{V}$, where $\varop{V} \eqdef
\{\vop{x} \mid x \in V \} \cup \{\vcl{x} \mid x \in V \}$. We assume that
$\varop{V}$ is disjoint with $\alphabet$ and $\svars$. Ref-words extend strings
over $\alphabet$ by encoding opening ($\vop{x}$) and closing ($\vcl{x}$) of
variables.

A ref-word $\refWord \in (\alphabet \cup \varop{V})^*$ is \emph{valid} if every
occurring variable is opened and closed exactly once. More formally, for each $x
\in V$, the string \refWord has precisely one occurrence of $\vop{x}$ and
precisely one occurrence of $\vcl{x}$, which is after the occurrence of
$\vop{x}$. For every valid ref-word $\refWord$ over $(\alphabet \cup \varop{V})$, we
define $\vars(\refWord)$ as the set of variables $x \in V$ which occur in the
ref-word. More formally,
\[
  \vars(\refWord) \;\eqdef\; \{ x \in V \mid \exists \refWord_x^{\pre},
  \refWord_x, \refWord_x^{\post} \in (\alphabet \cup \varop{V})^* \text{ such that }
  \refWord = \refWord_x^{\pre} \cdot \vop{x} \cdot \refWord_x \cdot \vcl{x}
  \cdot \refWord_x^{\post}\}.
\]
Intuitively, each valid ref-word $\refWord$ encodes a $\doc$-tuple for some
document $\doc$, where the document is given by symbols from $\sigma$ in
$\refWord$ and the variable markers encode where the spans begin and end.
Formally, we define functions $\clr$ and $\operatorname{tup}$ that, given a
valid ref-word, output the corresponding document and
tuple.\footnote{The function \clr is sometimes also called
  $\operatorname{clr}$ in literature (cf. Freydenberger et
  al.~\cite{FreydenbergerKP18}).} The morphism $\clr \colon (\alphabet \cup
\varop{V})^* \to \alphabet^*$ is defined on single symbols as:
\[
  \clr(\sigma) \;\eqdef\; \begin{cases}
    \sigma & \text{if }\sigma \in \alphabet \\
    \varepsilon & \text{if } \sigma \in \varop{V}
  \end{cases}
\]
and we define $\clr(\sigma_1 \cdots \sigma_n) \eqdef \clr(\sigma_1)\cdots\clr(\sigma_n)$.

We now define the function tup. 
By definition, every valid ref-word $\refWord$ over
$(\alphabet \cup \varop{V})$ has a unique factorization
\[
  \refWord \;=\; \refWord_x^{\pre} \cdot \vop{x} \cdot \refWord_x \cdot \vcl{x}
  \cdot \refWord_x^{\post}
\]
for each $x \in \vars(\refWord)$. We then define the function
$\operatorname{tup}$ as
\[
  \operatorname{tup}(\refWord) \;\eqdef\; \{ x \mapsto
  \spanFromTo{i_x}{j_x} \mid x \in \vars(\refWord), i_x =
  |\clr(\refWord_x^{\pre})|, j_x=i_x + |\clr(\refWord_x)| \}\;.
\]
The usage of the \clr morphism in the above definition ensures that the indices $i_x$ and $j_x$ refer to
positions in the document and do not consider other variable operations.

A \emph{ref-word language} $\reflang$ is a language of ref-words. 
We say that $\reflang$ is \emph{functional} if every ref-word $\refWord \in
\reflang$ is valid and there is a set $V$ of variables such that
$\vars(\refWord) = V$ for each $\refWord \in \reflang$.

Given a functional ref-word language $\reflang$, the spanner
$\toSpanner{\reflang}$ represented by $\reflang$ is given by
\[
  \toSpanner{\reflang}(\doc) \;\eqdef\; \big\{\toTuple{\refWord} \mid \refWord
  \in \reflang \text{ and } \clr(\refWord) = \doc \big\}\;.
\]

\subsubsection*{The Variable Order Condition and the $\operatorname{ref}$ Function}
Let $\refWord = \vop{x_1}\vop{x_2}a \vcl{x_1}\vcl{x_2}$ and $\refWordPrime =
\vop{x_1}\vop{x_2}a \vcl{x_2}\vcl{x_1}$ be ref-words. We observe that both
ref-words encode the tuple that selects the span $\spanFromTo{1}{2}$ in both
variables $x_1,x_2$ on document $a$. Thus, the same spanner can be represented
by multiple ref-word languages. We now introduce the \emph{variable order
  condition}, in order to achieve a one-to-one mapping between ref-words (resp.,
ref-word languages) and tuples (resp., spanners). To this end, we fix a total
linear order $\prec$ on the set $\varop{\svars}$ of variable operations, such
that $\vop{v} \prec \vcl{v}$ for every variable $v \in \svars$. We say that a
ref-word $\refWord$ \emph{satisfies the variable order condition} if all
adjacent variable operations in $\refWord$ are ordered according to the fixed
linear order $\prec$. That is, the ref-word $\refWord = \sigma_1 \cdots
\sigma_n$ satisfies the variable order condition if $\sigma_i \prec
\sigma_{i+1}$ for every $1 \leq i < n$ with $\sigma_i,\sigma_{i+1} \in
\varop{\svars}$. We observe that, for every document $\doc$ and every tuple
and $\tup = \toTuple{\refWord}$ that satisfies the variable order condition. 

We define $\operatorname{ref}$ as the function that, given a document $\doc$ and a
\doc-tuple $\tup$, returns the unique ref-word that satisfies the variable
order condition.
The following observation shows the connections between the functions $\clr$,
$\operatorname{ref},$ and $\operatorname{tup}$.
\begin{observation}\label{obs:canonicalRefWord}Let $\refWord$ be a valid ref-word and let $\refWordPrime \eqdef
  \refWordFrom{\clr(\refWord)}{\toTuple{\refWord}}$. Then $\toTuple{\refWord} =
  \toTuple{\refWordPrime}$. Furthermore, $\refWord = \refWordPrime$ if and only
  if $\refWord$ satisfies the variable order condition.
\end{observation}

Analogously to functionality, we say that a ref-word language $\reflang$
satisfies the variable order condition if every ref-word $\refWord \in \reflang$
satisfies the variable order condition.

\subsection{(Weighted) Variable Set-Automata}\label{sec:vset-automata}

In this section, we revisit the definition of \emph{weighted \vset-automaton} as a
formalism to represent $\K$-annotators~\cite{DoleschalKMP-lmcs22}. This formalism is a natural
generalization of both \vset-automata and weighted automata \cite{DrosteKV09}. Throughout the article, we will use weighted \vset-automata for two purposes: we use the \vset-automata over the Boolean semiring $\mathbb{B}$ for extracting spans from documents (as in the usual document spanner framework \cite{FaginKRV13}) and the more general $\srd$-weighted \vset-automata as one formalism for weight functions. (We discuss all considered variants for weight functions in Section~\ref{sec:weightfunctions}.)

Let $V \subseteq \vars$ be a finite set of variables. A \emph{weighted
  variable-set automaton over semiring $\K$} (alternatively, a \emph{weighted
  \vset-automaton} or a \emph{$\K$-weighted \vset-automaton}) is a tuple $A \df
(\alphabet, V,Q, I, F, \delta)$ where $\alphabet$ is a finite alphabet; $V\subseteq
\svars$ is a finite set of variables; $Q$ is a finite set of \emph{states}; $I :
Q \to \srd$ is the \emph{initial weight function}; $F: Q \to \srd$ is the
\emph{final weight function}; and $\delta : Q \times (\alphabet \cup
\{\varepsilon\} \cup \varop{V} ) \times Q \rightarrow \srd$ is a
\emph{($\K$-weighted) transition} function. 
We define the \emph{transitions} of $A$ as the set of triples $(p,o,q)$ with
$\delta(p,o,q) \ne \srzero$. Likewise, the \emph{initial} (resp.,
\emph{accepting}) states are those states $q$ with $I(q) \neq \srzero$ (resp.,
$F(q)\neq \srzero$).  For every semiring element $a \in \srd$, we denote the length of
the encoding of $a$ by $\enc{a}$. The \emph{size} of a weighted \vset-automaton
$A$ is defined by
\[
  |A| \eqdef |Q|+\sum_{q\in Q}\enc{I(q)}+\sum_{q\in Q}\enc{F(q)}+\sum_{p,q\in
    Q,\; a\in(\Sigma\cup\{\varepsilon\}\cup \varop{V})} \enc{\delta(p,a,q)}\;.
\]

Runs of $A$ are defined over ref-words. More precisely, a \emph{run} $\rn$ of $A$ on ref-word $\refWord = \sigma_1
\ldots \sigma_m$ is a sequence
$q_0 \overset{\sigma_1}{\rightarrow} \cdots \overset{\sigma_{m-1}}{\rightarrow}
  q_{m-1} \overset{\sigma_{m}}{\rightarrow} q_{m}$
where: 
\begin{itemize}
\item $I(q_0) \neq \srzero$ and $F(q_{m}) \neq \srzero$;
\item $\delta(q_i,\sigma_{i+1},q_{i+1}) \neq \srzero$ for all $0 \leq i < m$.
\end{itemize}
We say that a run $\rn$ is \emph{on a document $\doc$} if $\rn$ is a run on $\refWord$
and $\clr(\refWord) = \doc$. Furthermore, overloading notation, given a
run $\rn$ of $A$ on $\refWord$, we denote $\refWord$ by $\fromRun{\rn}$. We
define the \emph{ref-word language $\reflang(A)$} as the set of all ref-words
$\refWord$ such that $A$ has a run on $\refWord$.

The \emph{weight} of a run is obtained by $\srtimes$-multiplying the weights of
its constituent transitions. Formally, the weight $\weight{\rn}$ of $\rn$ is an
element in $\srd$ given by the expression
\[
  I(q_0) \srtimes  \delta(q_0,\sigma_1,q_1) \srtimes \cdots \srtimes
  \delta(q_{m},\sigma_{m},q_{m+1}) \srtimes F(q_{m+1})\;.
\]
We call $\rn$ \emph{nonzero} if $\weight{\rn} \neq \srzero$. Furthermore, $\rn$
is called \emph{valid} if $\fromRun{\rn}$ is valid and
\[\vars(\toTuple{\fromRun{\rn}}) = V\;.\footnote{Note that the second condition
  ensures that all valid runs are over the same set of variables. This is required, as
  $\K$-annotators map documents to annotated relations.}
\] If $\rn$ is valid we
denote the tuple $\toTuple{\fromRun{\rn}}$ by $\toTuple{\rn}$.

We say that a weighted \vset-automaton $A$ is \emph{functional} if 
every run of $A$ is valid.
We
denote the set of all valid and nonzero runs of $A$ on $\doc$ by
\[
  \Rn{A}{\doc} \eqdef \{\rn \mid \fromRun{\rn} \in \reflang(A) \text{ and }
  \doc = \clr(\fromRun{\rn})\}\;.
\]

Notice that there may be infinitely many valid and nonzero runs of a weighted
\vset-automaton on a given document, due to \emph{$\varepsilon$-cycles}, which
are states $q_1,\ldots,q_k$ such that $(q_i,\varepsilon,q_{i+1})$ is a
transition for every $i \in \{1,\ldots,k-1\}$ and $q_1 = q_k$. Following
Doleschal et al.~\cite{DoleschalKMP-lmcs22} we assume that weighted
\vset-automata do not have $\varepsilon$-cycles, unless mentioned otherwise.

As such, if $A$ does not have $\varepsilon$-cycles, then the result of applying
$A$ on a document $\doc$, denoted $\repspnrw{A}(\doc)$, is the
$\srd$-relation $R$ for which
\[
  R(\tup) \df \boplus_{\rn\in \Rn{A}{\doc} \text{ and } \tup = \toTuple{\rn}}
 \weight{\rn}\;.
\]
Note that $\Rn{A}{\doc}$ only contains runs $\rn$ that are valid and nonzero. If
$\tup$ is a $V^\prime$-tuple with $V^\prime \ne V$ then $R(\tup) = \srzero$,
because we only consider valid runs. In addition, $\repspnrw{A}$ is a well defined
$\srd$-annotator since every $V$-tuple in the support of $\repspnrw{A}(\doc)$ is a $V$-tuple over
$\spans(\doc)$. To simplify notation, we sometimes denote
$\repspnrw{A}(\doc)(\tup)$ --- the weight assigned to the \doc-tuple $\tup$ by
$A$ --- by $\repspnrw{A}(\doc,\tup)$. We say that two $\srd$-weighted
\vset-automata $A_1$ and $A_2$ are \emph{disjoint} if $\reflang(A_1) \cap
\reflang(A_2) = \emptyset$. This implies that also the corresponding
$\srd$-annotators $\repspnrw{A_1}$ and $\repspnrw{A_2}$ are disjoint.

We say that a $\K$-annotator (resp., document spanner) $S$ is \emph{regular} if
there exists a weighted \vset-automaton (resp., $\B$-weighted \vset-automaton) $A$
such that $S = \repspnrw{A}$. Note that this is an equality between functions. If $\srd$ is clear from the context, we may just write $\toSpanner{A}$ instead of $\repspnrw{A}$.

We say that two weighted \vset-automata $A$ and $A^\prime$ are
\emph{equivalent} if they define the same $\K$-annotator, that is, $\repspnrw{A} =
\repspnrw{A^\prime}$, which is the case if $\repspnrw{A}(\doc) =
\repspnrw{A^\prime}(\doc)$ for every $\doc \in \docs$. 

Similar to our
terminology on $\B$-annotators, we refer to (functional) $\B$-weighted
\vset-automata as \emph{(functional) \vset-automata}. Since \vset-automata can always be
translated into equivalent functional \vset-automata~\cite[Proposition
3.9]{Freydenberger19}, we assume in this article that \vset-automata are
functional. This is a common assumption for document spanners involving regular
languages~\cite{FaginKRV15, Freydenberger19, PeterfreundFKK19}. Furthermore, we 
assume that all weighted \vset-automata are functional as well. In the
following, we denote by \regk the class of all functional $\K$-weighted
\vset-automata and by $\fvsa$ the class of all functional \vset-automata. 

Due to the close relationship between regular expressions and $\B$-weighted automata, and since regular expressions are easy to read, we sometimes define $\B$-weighted \vset-automata using regular expressions over $\Sigma \cup \Gamma_V$. Here, we use $\cdot$ to denote concatenation, $\lor$ to denote disjunction, and $^*$ to denote Kleene star. As usual, we often omit $\cdot$ and use priority rules ($^*$ before $\cdot$ before $\lor$) for improving the readability of expressions.

\subsubsection*{Unambiguous (weighted) \vset-Automata}

We now discuss unambiguity for (weighted) \vset-automata. A
(weighted) \vset-automaton $A$ is \emph{unambiguous} if it satisfies the
following two conditions.
\begin{enumerate}[label=(C\arabic*)]
\item $\reflang(A)$ satisfies the variable order condition; \label{cond:voc}
\item for every $\refWord \in
  \reflang(A)$, there is exactly one run of $A$ on $\refWord$.\label{cond:unambig}
\end{enumerate}

We note that for Boolean spanners, i.e. spanners with no variables, the
definitions coincide with the classical unambiguity definition of finite state
automata. That is, a \vset-automaton with $\vars(A) = \emptyset$ is unambiguous
if it is a unambiguous finite state automaton. Furthermore, we note that every
\vset-automaton $A$ can 
be transformed to an equivalent unambiguous \vset-automaton $A'$. (e.g. Doleschal et al.~\cite[Lemma
4.5]{DoleschalKMNN19-journal}). However, \vset-automata
can be exponentially more succinct than equivalent unambiguous \vset-automata.\footnote{Note that, for functional
  \vset-automata, the exponential factor in the relative succinctness is caused
  by condition~\ref{cond:unambig}. That is, for every functional \vset-automaton
  $A$, there is an equivalent functional \vset-automaton $|A'|$ which satisfies
  condition~\ref{cond:voc} and is of size at most polynomial in $|A|$.}

\begin{example}The span relation on the bottom right of Figure~\ref{fig:MainExample} can be
  extracted from \doc by a spanner that matches textual representations of
  numbers (or ranges) in the variable $x_{\textsf{events}}$, followed by a city
  or country name, matched in $x_{\textsf{loc}}$.
  Figure~\ref{fig:exampleVsetAutomata} shows how two such \vset-automata may
  look like. Note that some strings, like $\text{Luxembourg}$ are the name of a
  city as well as a country. Thus, the upper automaton is ambiguous, because the
  tuple with Luxembourg is captured twice (thus, violating \ref{cond:unambig}).
  The lower automaton is unambiguous, because the sub-automaton for $\text{Loc}$
  only matches such names once.\qed
\end{example}
  
\begin{figure}
  \resizebox{\linewidth}{!}{
    \begin{tikzpicture}[>=latex]
      \node[state,initial] (q0) at (0,0){$q_0$};
      \node[state] (q1) at ($(q0)+(2,0)$) {$q_1$};
      \node[state] (q2) at ($(q1)+(2,0)$) {$q_2$};
      \node[state] (q3) at ($(q2)+(2,0)$) {$q_3$};
      \node[state] (q4) at ($(q3)+(2,0)$) {$q_4$};
      \node[state] (q5) at ($(q4)+(2,0)$) {$q_5$};
      \node[state] (q6) at ($(q5)+(2,0)$) {$q_6$};
      \node[state,accepting] (q7) at ($(q6)+(2,0)$) {$q_7$};
      
      \draw (q0) edge[->, loop above] node [above] {$\alphabet$}(q0);
      \draw (q0) edge[->] node[below] {$\vop{x_{\textsf{events}}}$}(q1);
      \draw (q1) edge[->, decorate, decoration={snake, amplitude=.5mm, segment length=2mm, post length=1mm}] node[above] {\textsf{Num}}(q2);
      \draw (q2) edge[->] node[below] {$\vcl{x_{\textsf{events}}}$}(q3);
      \draw (q3) edge[->, decorate, decoration={snake, amplitude=.5mm, segment length=2mm, post length=1mm}] node[above] {\textsf{Gap}}(q4);
      \draw (q4) edge[->] node[below] {$\vop{x_{\textsf{loc}}}$}(q5);
      \draw (q5) edge[->, bend left, decorate, decoration={snake, amplitude=.5mm, segment length=2mm, post length=1mm}] node[above] {\textsf{City}}(q6);
      \draw (q5) edge[->, bend right, decorate, decoration={snake, amplitude=.5mm, segment length=2mm, post length=1mm}] node[below] {\textsf{Country}}(q6);
      \draw (q6) edge[->]  node[below] {$\vcl{x_{\textsf{loc}}}$}(q7);
      \draw (q7) edge[->, loop above] node [above] {$\alphabet$}(q7);
    \end{tikzpicture}
  }
    
  \resizebox{\linewidth}{!}{
      \begin{tikzpicture}[>=latex]
      \node[state,initial] (q0) at (0,0){$q_0$};
      \node[state] (q1) at ($(q0)+(2,0)$) {$q_1$};
      \node[state] (q2) at ($(q1)+(2,0)$) {$q_2$};
      \node[state] (q3) at ($(q2)+(2,0)$) {$q_3$};
      \node[state] (q4) at ($(q3)+(2,0)$) {$q_4$};
      \node[state] (q5) at ($(q4)+(2,0)$) {$q_5$};
      \node[state] (q6) at ($(q5)+(2,0)$) {$q_6$};
      \node[state,accepting] (q7) at ($(q6)+(2,0)$) {$q_7$};
      
      \draw (q0) edge[->, loop above] node [above] {$\alphabet$}(q0);
      \draw (q0) edge[->] node[below] {$\vop{x_{\textsf{events}}}$}(q1);
      \draw (q1) edge[->, decorate, decoration={snake, amplitude=.5mm, segment length=2mm, post length=1mm}] node[above] {\textsf{Num}}(q2);
      \draw (q2) edge[->] node[below] {$\vcl{x_{\textsf{events}}}$}(q3);
      \draw (q3) edge[->, decorate, decoration={snake, amplitude=.5mm, segment length=2mm, post length=1mm}] node[above] {\textsf{Gap}}(q4);
      \draw (q4) edge[->] node[below] {$\vop{x_{\textsf{loc}}}$}(q5);
      \draw (q5) edge[->, decorate, decoration={snake, amplitude=.5mm, segment length=2mm, post length=1mm}] node[above] {\textsf{Loc}}(q6);
      \draw (q6) edge[->] node[below] {$\vcl{x_{\textsf{loc}}}$}(q7);
      \draw (q7) edge[->, loop above] node [above] {$\alphabet$}(q7);
    \end{tikzpicture}
  }
  \caption{Two example \vset-automata that extract the span relation $R$ on
    input $\doc$ as defined in Figure~\ref{fig:MainExample}. For the sake of
    presentation, the automata are simplified as follows: \textsf{Num} is a
    sub-automaton matching anything representing a number (of events) or range,
    \textsf{Gap} is a sub-automaton matching sequences of at most three words,
    \textsf{City} and \textsf{Country} are sub-automata matching city and
    country names respectively. \textsf{Loc} is a sub-automaton for the union of
    \textsf{City} and \textsf{Country}. All these sub-automata are assumed to be
    unambiguous.\label{fig:exampleVsetAutomata}}
\end{figure}

In the following, we denote by \uregk the class of $\K$-weighted unambiguous
functional \vset-automata and by \ufvsa the class of unambiguous functional
\vset-automata.

\subsection{Aggregate Queries}
Aggregation functions, such as $\min$, $\max$, and $\tsum$ operate on numerical
values from database tuples, whereas all the values of \doc-tuples are spans.
Yet, these spans may represent numerical values, from the document $\doc$,
encoded by the captured words (e.g., ``3,'' ``three,'' ``March'' and so on). To
connect spans to numerical values, we will use \emph{weight functions}
\begin{definition}[Weight function]
  Denote by $\Tup$ the set of all $V$-tuples for sets $V$, i.e., the union of all sets $\vtup{V}$.
  A \emph{weight function} is a function $\w: \docs \times \Tup \to \{\Q \cup \infty\}$. It maps pairs of documents $\doc$ and $\doc$-tuples $\tup$ to values in $\Q$ or to $\infty$.
\end{definition}
In the definition of weight functions, we allow the range to include $\infty$, since we will use subsets of $\Q$ and the tropical semiring $\tropical$, the latter of which contains $\infty$. 
We discuss weight
functions in more detail in Section~\ref{sec:weightfunctions}.

\def\n#1{\textsf{\tiny{#1}}}
\def\s#1{\texttt{#1}}
\begin{figure}[t]
  \centering\small
  {
    \setlength{\tabcolsep}{.3ex}
    \begin{tabular}{
          cccccccccc
          cccccccccc
          cccccccccc
          cccccccccc
          c
          }
          $\s{T}$ & $\s{h}$ & $\s{e}$  & $\s{r}$ & $\s{e}$ & $\s{\blank}$
          & $\s{a}$ & $\s{r}$ & $\s{e}$ & $\s{\blank}$                                                  
          & $\s{7}$ & $\s{\blank}$ & $\s{e}$ & $\s{v}$ & $\s{e}$ & $\s{n}$ & $\s{t}$ & $\s{s}$ & $\s{\blank}$
          & $\s{i}$ & $\s{n}$ & $\s{\blank}$  & $\s{B}$ & $\s{e}$ & $\s{l}$ & $\s{g}$ & $\s{i}$ & $\s{u}$ & $\s{m}$  & $\s{,}$  & $\s{\blank}$
          & $\s{1}$ & $\s{0}$ & $\s{-}$ & $\s{1}$ & $\s{5}$ & $\s{\blank}$ & $\s{i}$ & $\s{n}$ & $\s{\blank}$
          & 
          \\\cmidrule{1-40}
          $\n{1}$ & $\n{2}$ & $\n{3}$  & $\n{4}$ & $\n{5}$ & $\n{6}$ & $\n{7}$ & $\n{8}$ & $\n{9}$ & $\n{10}$
          & $\n{11}$ & $\n{12}$ & $\n{13}$ & $\n{14}$ & $\n{15}$ & $\n{16}$ & $\n{17}$ & $\n{18}$ & $\n{19}$ & $\n{20}$
          & $\n{21}$ & $\n{22}$ & $\n{23}$ & $\n{24}$ & $\n{25}$  & $\n{26}$ & $\n{27}$ & $\n{28}$ & $\n{29}$ & $\n{30}$
          & $\n{31}$ & $\n{32}$ & $\n{33}$ & $\n{34}$ & $\n{35}$ & $\n{36}$ & $\n{37}$ & $\n{38}$ & $\n{39}$ & $\n{40}$
          & \\\addlinespace[1.5\defaultaddspace]
          $\s{F}$ & $\s{r}$ & $\s{a}$ & $\s{n}$ & $\s{c}$ & $\s{e}$ & $\s{,}$ & $\s{\blank}$
          & $\s{4}$ & $\s{\blank}$ & $\s{i}$ & $\s{n}$ & $\s{\blank}$
          & $\s{L}$ & $\s{u}$ & $\s{x}$ & $\s{e}$ & $\s{m}$ & $\s{b}$ & $\s{o}$ & $\s{u}$ & $\s{r}$ & $\s{g}$ & $\s{,}$ & $\s{\blank}$ 
          & $\s{t}$ & $\s{h}$ & $\s{r}$ & $\s{e}$ & $\s{e}$ & $\s{\blank}$ & $\s{i}$ & $\s{n}$ & $\s{\blank}$
          & $\s{B}$ & $\s{e}$ & $\s{r}$ & $\s{l}$ & $\s{i}$ & $\s{n}$ & $\s{.}$ 
          \\\midrule
          $\n{41}$ & $\n{42}$ & $\n{43}$ & $\n{44}$ & $\n{45}$ & $\n{46}$ & $\n{47}$ & $\n{48}$ & $\n{49}$ & $\n{50}$
          & $\n{51}$ & $\n{52}$ & $\n{53}$ & $\n{54}$ & $\n{55}$ & $\n{56}$ & $\n{57}$ & $\n{58}$ & $\n{59}$ & $\n{60}$
          & $\n{61}$ & $\n{62}$ & $\n{63}$ & $\n{64}$ & $\n{65}$ & $\n{66}$ & $\n{67}$ & $\n{68}$ & $\n{69}$ & $\n{70}$
          & $\n{71}$ & $\n{72}$ & $\n{73}$ & $\n{74}$ & $\n{75}$ & $\n{76}$ & $\n{77}$ & $\n{78}$ & $\n{79}$ & $\n{80}$
          & $\n{81}$                                        
    \end{tabular}  
  }
  
  \medskip
  
   \begin{tabular}[t]{ll}
    \toprule
    $x_{\textsf{loc}}$ & $x_{\textsf{events}}$\\
    \midrule
    $\mspan{23}{30}$ & $\mspan{11}{12}$ \\
    $\mspan{41}{47}$ & $\mspan{32}{37}$ \\
    $\mspan{54}{64}$ & $\mspan{49}{50}$ \\
    $\mspan{75}{81}$ & $\mspan{66}{71}$ \\
    \bottomrule
   \end{tabular}
   \qquad
   \begin{tabular}[t]{ll}
     \toprule
     ${x_{\textsf{events}}}$ & $W({x_{\textsf{events}}})$ \\
     \midrule
     7 & 7 \\
     10-15 & 10 \\
     4 & 4 \\
     three & 3 \\
     \bottomrule
   \end{tabular}
   \qquad
   \begin{tabular}[t]{lll}
     \toprule
     $\doc_{x_{\textsf{loc}}}$ & $\doc_{x_{\textsf{events}}}$ & $W_R(\tup)$ \\
     \midrule
     Belgium & 7 & 7 \\
     France & 10-15 & 10 \\
     Luxembourg & 4 & 4 \\
     Berlin & three & 3 \\
    \bottomrule
  \end{tabular}
  \qquad
   \caption{A document $\doc$ (top), a span relation $R$ (bottom left), a
     $\Q$-weighted string relation $W$ (bottom middle) and the $\Q$-weighted string
     relation $W_R$ resulting from $W,\doc,$ and $R$ (bottom right).
     \label{fig:MainExampleWithWeights}}
\end{figure}

\begin{example}\label{ex:MainExampleAggregates}Consider the document in Figure~\ref{fig:MainExampleWithWeights} and assume
  that we want to calculate the total number of mentioned events. The relation
  $R$ at the bottom left depicts a possible extraction of locations with their
  number of events. The table in the bottom middle depicts a weighted string relation $W$ (where the weight of each string is in the rightmost column). The relation on the bottom right depicts the string relation where
  each tuple is annotated with a weight corresponding to $W,R,$ and $\doc$. To
  get an understanding of the total number of events, we may want to take the
  sum over the weights of the extracted tuples, namely $7 + 10 + 4 + 3 =
  24$.\qed
\end{example}

For a spanner $\spanner$, a document $\doc$, and weight function $\w$, we denote
by $\wCodom(\spanner,\doc,\w)$ the set of weights of output tuples of $\spanner$
on $\doc$, that is, $\wCodom(\spanner, \doc,\w) = \{\w(\doc,\tup) \mid \tup \in
\spanner(\doc)\}$. Furthermore, let $\wCodom(\w) \subseteq \Q$ be the set of
weights assigned by $\w$, that is, $k \in \wCodom(\w)$ if and only if there is a
document $\doc$ and a $\doc$-tuple $\tup$ with $\w(\doc,\tup) = k$.

\begin{definition}Let \doc be a document and $A$ be a \vset-automaton such that
  $\toSpanner{A}(\doc) \neq \emptyset$. Let $\spanner = \toSpanner{A}$, let \w be a
  weight function, and $q \in \rationals$ with $0 \leq q \leq 1$. We define the
  following spanner aggregation functions:
  \begin{align*}
    \spcountagg \eqdef& \;\;|\spanner(\doc)| \\ 
    \spminagg \eqdef& \min_{\tup\in \spanner(\doc)}\w(\doc,\tup) \\
    \spmaxagg \eqdef& \max_{\tup \in \spanner(\doc)}\w(\doc,\tup) \\
    \spsumagg \eqdef&  \sum_{\tup \in \spanner(\doc)}\w(\doc,\tup) \\
    \spavgagg \eqdef& \;\; \frac{\spsumagg}{\spcountagg} \\
    \spquantagg \eqdef& \;\;\min\bigg\{r \in \wCodom(\spanner,\doc,\w) \ \bigg| \ 
                         \frac{|\set{\tup\in\spanner(\doc)\mid \w(\doc,\tup)\leq r}|}
                        {|\spanner(\doc)|} \geq q\bigg\}
  \end{align*}
We observe that $\spminagg = \spqquantagg{0}$ and $\spmaxagg = \spqquantagg{1}$.
\end{definition}

\subsection{Main Problems}
Let $\spannerClass$ be a class of regular document spanners and $\wC$ be a class
of weight functions. We define the following problems.

\computeproblemWidth{.56}{$\countp[\spannerClass]$}{Spanner $\spanner \in \spannerClass$ and
  document $\doc \in \docs$.}{Compute $\spcount(\spanner,\doc)$.}

\computeproblemWidth{.81}{$\sump[\spannerClass,\wC]$}{Spanner $\spanner \in \spannerClass$,
  document $\doc \in \docs$, a weight function $\w \in \wC$.}{Compute
  $\spsumagg$.}

The problems $\avgp[\spannerClass,\wC], \quantp[\spannerClass,\wC],
\minp[\spannerClass,\wC],$ and $\maxp[\spannerClass,\wC]$ are defined
analogously to $\sump[\spannerClass,\wC]$. Notice that all these problems study
\emph{combined complexity}. Since the number of tuples in $\spanner(\doc)$ is
always in $O(|\doc|^{2k})$, where $k$ is the number of variables of the spanner
$\spanner$ (cf. Corlollary~\ref{cor:tuplesPerSpanner}), the \emph{data
  complexity} of all the problems is in \fp: One can just materialize
$\spanner(\doc)$ and apply the necessary aggregate. Under combined complexity,
we will therefore need to find ways to avoid materializing $\spanner(\doc)$ to
achieve tractability.

\subsection{Algorithms and Complexity Classes}
Before we discuss our main results in Section~\ref{main-results}, we provide a few definitions on computational complexity.

We first define fully polynomial-time randomized
approximation schemes (\fpras).

\begin{definition}
  Let $f$ be a function that maps inputs $x$ to rational numbers and let
  $\mathcal{A}$ be a probabilistic algorithm, which takes an input instance $x$
  and a parameter $\delta > 0$. Then $\mathcal{A}$ is called a \emph{fully
    polynomial-time randomized approximation scheme} (\fpras), if
  \begin{itemize}
  \item $\bigprobab{\big|\mathcal{A}(x,\delta) - f(x)\big| \;\leq\; \delta \cdot
      \big|f(x)\big|} \;\;\geq\;\;\frac{3}{4}\;$;
  \item the runtime of $\mathcal{A}$ is polynomial in $|x|$ and
    $\frac{1}{\delta}\;$.
  \end{itemize}
\end{definition}

The following definitions closely follow the Handbook of Theoretical
Computer Science~\cite{HandbookTCS91}. The class \emph{\fp} (respectively,
\fexptime) is the set of all functions that are computable in polynomial time
(resp., in exponential time). A \emph{counting Turing Machine} is an
non-deterministic Turing Machine whose output for a given input is the number of
accepting computations for that input. Given functions $f,g: \Sigma^* \to \nat$,
$f$ is said to be \emph{parsimoniously reducible} to $g$ in polynomial time if there is
a function $h: \Sigma^* \to \Sigma^*$, which is computable in polynomial time,
such that for every $x \in \Sigma^*$ it holds that $f(x) = g(h(x))$.
Furthermore, we say that $f$ is \emph{Turing reducible} to $g$ in polynomial time, if
$f$ can be computed by a polynomial time Turing Machine $M$, which has access to
an oracle for $g$.

The class \emph{\sharpp} is the set of all functions that are computable by
polynomial-time counting Turing Machines. A problem $X$ is \emph{\sharpp-hard}
under parsimonious reductions (resp., Turing reductions) if there are polynomial
time parsimonious reductions (resp., Turing reductions) to it from all problems
in \sharpp. If in addition $X \in \sharpp$, we say that $X$ is \sharpp-complete
under parsimonious reductions (resp., Turing reductions).

The class \emph{\fptosp} is the set of all functions that are computable in
polynomial time by an oracle Turing Machine with a \sharpp oracle. It is easy
to see that, under Turing reductions, a problem is hard for the class \sharpp if
and only if it is hard for \fptosp. We note that every problem which is
$\sharpp$-hard under parsimonious reductions is also \sharpp-hard under Turing
reductions. Therefore, unless mentioned otherwise, we always use parsimonious
reductions.

The class $\spanl$ is the class of all functions $f:\Sigma^* \to \nat$ for which
there is a nondeterministic logarithmic space Turing Machine $M$ with input
alphabet $\Sigma$ such that $f(x) = |M(x)|$.

The class \emph{\optp} is the set of all functions computable by taking the
maximum output value over all accepting computations of a polynomial-time
non-deterministic Turing Machine that outputs natural numbers. Assume that
$\Gamma$ is the Turing Machine alphabet. Let $f,g: \Gamma^* \to \nat$ be
functions. A \emph{metric reduction}, as introduced by Krentel~\cite{Krentel88},
from $f$ to $g$ is a pair of polynomial-time computable functions $T_1, T_2$,
where $T_1: \Gamma^* \to \Gamma^*$ and $T_2: \Gamma^* \times \nat \to \nat$,
such that $f(x) = T_2(x,g(T_1(x)))$ for all $x \in \Gamma^*$.

The class \emph{\bpp} is the set of all decision problems solvable in polynomial
time by a probabilistic Turing Machine in which the answer always has
probability at least $\frac{1}{2} + \delta$ of being correct for some fixed
$\delta > 0$.

 \section{Main Results}
\label{main-results}
In this section we present the main results of this article.

\subsection{Known Results}
We begin by giving an overview of the results on \countp, which are known from
the literature. 

\begin{thmC}[Arenas et al.~{\cite{ArenasCJR19}, Florenzano et
    al.~\cite{FlorenzanoRUVV18}}]\label{thm:countResults}$\countp[\ufvsa]$ is in \fp and $\countp[\fvsa]$ is \spanl-complete.
  Furthermore, $\countp[\fvsa]$ can be approximated by an \fpras.
\end{thmC}
\begin{proof} 
  Follows from Arenas et al.~\cite[Corollaries~4.1 and~4.2]{ArenasCJR19}, and
  Florenzano et al.~\cite[Theorem 5.2]{FlorenzanoRUVV18}.
\end{proof}

The $\spanl$ lower bound by Florenzano et al.~\cite[Theorem
5.2]{FlorenzanoRUVV18} is due to a parsimonious reduction from the
\#NFA$(n)$-problem\footnote{Given an NFA $A$ and a natural number $n$, encoded
  in binary, the \#NFA$(n)$ problem asks for the number of words $w \in
  \lang(A)$ of length $n$. The \#NFA$(n)$ problem is sometimes also called
  Census Problem.} which is known to be \sharpp-complete under Turing reductions (cf.
Kannan et al.~\cite{KannanSM95}). As every parsimonious reduction is also a
Turing reduction, the following corollary follows immediately.

\begin{corollary}\label{cor:countSharpp}$\countp[\fvsa]$ is \sharpp-hard under Turing reductions.
\end{corollary}

Two observations can be made from these results. First, \countp requires the
input spanner to be \emph{unambiguous} for tractability. This tractability
implies that \countp can be computed without materializing the possibly
exponentially large set $\spanner(\doc)$ if the spanner is unambiguous.
Furthermore, if the spanner is not unambiguous then, due to \spanl-completeness
of \countp, we do not know an efficient algorithm for its exact computation (and
therefore may have to materialize $\spanner(\doc)$), but \countp can be
\emph{approximated} by an \fpras. We will explore to which extent this picture
generalizes to other aggregates.

\subsection{Overview of New Results}

\begin{table}[t]
 \resizebox{\textwidth}{!}{
    \centering
    \begin{tabular}{cllll} 
      \toprule
      Aggregate & Spanner & Weights & Complexity & Approximation\\
      \midrule
      \multirow{2}{*}{$\countp$} & $\ufvsa$ & - & in \fp & - \\
      \cmidrule(lr){2-5}
                & $\fvsa$ & - & \sharpp-hard$^\dagger$ & \fpras \\
      \cmidrule(lr){1-5}
      \multirow{2}{*}{$\minp$} 
                & \multirow{2}{*}{$\ufvsa, \fvsa$} & $\SW, \UREG, \regtrop$ & in \fp (\ref{thm:minmaxfp},\ref{thm:minmaxRegTractable}) & - \\
      \cmidrule(lr){3-5}
                &  & $\regnum, \PW$  & \optp-hard (\ref{thm:pwintractableKarp},\ref{thm:minmaxRegintractable}) & no \fpras (\ref{thm:minapx})\\
      \cmidrule(lr){1-5}
      \multirow{2}{*}{$\maxp$}
                & \multirow{2}{*}{$\ufvsa, \fvsa$} & $\SW,\UREG$ & in \fp (\ref{thm:minmaxfp},\ref{thm:minmaxRegTractable}) & - \\
      \cmidrule(lr){3-5}
                & & $\regtrop, \regnum, \PW$  & \optp-hard (\ref{thm:pwintractableKarp},\ref{thm:minmaxRegintractable}) & no \fpras (\ref{thm:minapx},\ref{thm:maxapx})\\
      \cmidrule(lr){1-5}
      \multirow{4}{*}{$\sump$} & \multirow{2}{*}{$\ufvsa$} & $\SW,\UREG, \regnum$ & in \fp (\ref{thm:aggregatefp},\ref{thm:sumREGnumerical},\ref{thm:sumUregFP}) & - \\
      \cmidrule(lr){3-5}
                & & $\regtrop, \PW$ & $\sharpp$-hard (\ref{thm:pwintractableKarp},\ref{thm:sumTropicalHard}) & no \fpras (\ref{cor:sumavgNoFpras}) \\
      \cmidrule(lr){2-5} 
                & \multirow{2}{*}{$\fvsa$} & $\SW_{\nat}$ & \spanl-complete (\ref{thm:sw-sum-spanl}) & \fpras (\ref{cor:sumSpanl})\\
      \cmidrule(lr){3-5}
                & & $\SW,\UREG,\REG,\PW$ & $\sharpp$-hard (\ref{thm:sw-sum-cnf}) & no \fpras (\ref{cor:sumavgNoFpras}) \\
      \cmidrule(lr){1-5}
      \multirow{4}{*}{$\avgp$} & \multirow{2}{*}{$\ufvsa$} & $\SW,\UREG, \regnum$ & in \fp (\ref{thm:aggregatefp},\ref{cor:avgregfp}) & - \\
      \cmidrule(lr){3-5}
                & & $\regtrop, \PW$ & $\sharpp$-hard (\ref{thm:pwintractableKarp}) & no \fpras (\ref{cor:sumavgNoFpras}) \\
      \cmidrule(lr){2-5} 
                & \multirow{2}{*}{$\fvsa$} & $\SW_{\Q^+}$ & \sharpp-hard$^\dagger$ (\ref{thm:sw-quant-avg-hard}) & \fpras (\ref{thm:fprasSumAvg}) \\
      \cmidrule(lr){3-5}
                & & $\SW,\UREG,\REG,\PW$ & $\sharpp$-hard$^\dagger$ (\ref{thm:sw-quant-avg-hard},\ref{thm:sumTropicalHard}) & no \fpras (\ref{thm:negSumAvg})\\
      \cmidrule(lr){1-5}
      \multirow{3}{*}{$\quantp$} & \multirow{2}{*}{$\ufvsa$} & $\SW$ & in \fp (\ref{thm:aggregatefp})& - \\
      \cmidrule(lr){3-5}
                & & $\ureg,\REG,\PW$ & $\sharpp$-hard$^\dagger$ (\ref{thm:pwintractableCook},\ref{thm:quantUregHard}) & no \fpras (\ref{thm:noUregFprasQuant}) \\
      \cmidrule(lr){2-5} 
                & $\fvsa$ & $\SW,\UREG,\REG,\PW$ & \sharpp-hard$^\dagger$ (\ref{thm:sw-quant-avg-hard}) & no \fpras (\ref{thm:noFprasQuant})\\
      \cmidrule(lr){1-5}
      $\quantp$  & \multirow{2}{*}{$\fvsa$} & \multirow{2}{*}{$\PW$} & - & \fpras-like \\[-1mm]
      (positional)&&&&approx.\ (\ref{thm:epsQuantApproximation})\\
      \bottomrule
    \end{tabular}
 }
  \caption{Detailed overview of complexities of aggregate problems for document
    spanners. All problems are in \fexptime. The ``no FPRAS'' claims either
    assume that $\rp \neq \np$ or assume that the polynomial hierarchy does not
    collapse. The $\sharpp$-hardness results, marked with $^\dagger$ rely on
    Turing reductions. The numbers refer to the numbers of new results.}\label{tab:overview}
\end{table}
The complexity results are summarized in Table~\ref{tab:overview}. By now the
reader is familiar with the aggregate problems and the types of spanners we
study.
We obtain different results for different representations of weight functions, which we denote here as 
\SW, \PW, and \REG (resp., \UREG) and define formally in Section~\ref{sec:weightfunctions}.
Intuitively, \SW are \emph{constant-width} weight functions that assign values based on strings selected by a constant number of variables; \PW are polynomial-time computable weight
functions, and \REG (resp., \UREG) are weight functions represented by weighted
(resp., unambiguous weighted) \vset-automata. Furthermore, we sometimes restrict these classes based on their range. For instance, $\SW_{\nat}$ and $\SW_{{\Q^+}}$ are the constant-width weight functions that map to natural numbers and positive rational numbers, respectively.

Entries in the table should be read from left to right. For instance, the third
row states that the \minp problem, for both spanner classes \ufvsa and \fvsa,
and for all three classes \SW, \uregtrop, and \regtrop of weight functions is in
\fp. Likewise, the fourth row states that the same problems with $\regnum$ or
$\PW$ weight functions become \optp-hard and that the existence of an \fpras
would contradict commonly believed conjectures.

In general, the table gives a detailed overview of the impact of (1) unambiguity
of spanners and (2) different weight function representations on the complexity
of computing aggregates.

\subsection{Results for Different Weight Functions}\label{sec:weightfunctions}

We formalize how we represent the weight functions for our new results. Recall
that weight functions $\w$ map pairs consisting of a document $\doc$ and
$\doc$-tuple $\tup$ to values in $\Q \cup \{\infty\}$.

\subsubsection{Constant-Width Weight Functions} 
The simplest type of weight functions we consider are the \emph{constant-width
  weight functions}.\footnote{These generalize the
  single-variable weight functions of Doleschal et
  al.~\cite{DoleschalBKM21}.} Let $1 \leq c \in \nat$ be a constant. A
\emph{constant-width weight function (\SW)} $\w$ assigns values based on the
strings selected by at most $c$ variables. 
A constant-width weight function \SW
is given in the input as a $\Q$-weighted string relation, i.e., a string relation $R$ over the numerical semiring \numericalSR
and the variables $\xvars$, where $\xvars \subseteq \svars$, is a set of at most
$c$ variables. 
Recall that $\doc_{\tup}$ denotes the tuple $(\doc_{\tup(x_1)},
\ldots , \doc_{\tup(x_n)})$, where $\vars(\tup) = \{x_1,\ldots,x_n\}$. To
facilitate presentation, we assume that the variables in $X$ are always present
in $\tup$, that is, $X \subseteq \vars(\tup)$. The weight function
$\w(\doc,\tup)$ is defined as
\[
  \w(\doc,\tup) = R(\doc_{\pi_{\xvars}\tup})\;.
\]

As we will see in Section~\ref{cwidth}, the problems $\maxp[\fvsa,\SW]$ and
$\minp[\fvsa,\SW]$ are in \fp (Theorem~\ref{thm:minmaxfp}). Furthermore, we show
that the problems $\sump[\spannerClass,\SW]$, $\avgp[\spannerClass,\SW]$, and
$\quantp[\spannerClass,\SW]$ behave similarly to $\countp[\spannerClass]$, that
is, they are in \fp if $\spannerClass = \ufvsa$ (Theorem~\ref{thm:aggregatefp})
and intractable if $\spannerClass = \fvsa$
(Theorems~\ref{thm:sw-sum-cnf},~\ref{thm:sw-sum-spanl},
and~\ref{thm:sw-quant-avg-hard}).

\subsubsection{Polynomial-Time Weight Functions}
How far can we push our tractability results? Next, we consider more general
ways of mapping $\doc$-tuples into numbers. The most general class of weight
functions we consider is the set of \emph{polynomial-time weight functions}
  (\PW). A function $\w$ from \PW is given in the input as a polynomial-time
Turing Machine $M$ that maps $(\doc,\tup)$-pairs to values in $\Q$ and defines
$\w(\doc,\tup) = M(\doc,\tup)$. Not surprisingly there are multiple drawbacks of
having arbitrary polynomial time weight functions. The first is that all
considered aggregates become intractable, even if we only consider unambiguous
\vset-automata (Theorems~\ref{thm:pwintractableKarp},
and~\ref{thm:pwintractableCook}). However, all aggregates can at least be
computed in exponential time (Theorem~\ref{thm:pwfexptime}).

\subsubsection{Regular Weight Functions}
As the class of polynomial-time weight functions quickly leads to
intractability, we focus on a restricted class \REG that we introduce next and is less restrictive than \SW
but not as general as \PW such that we can understand the structure of the
representation towards efficient algorithms.\footnote{We prove in Section~\ref{sec:weight-function-expressiveness} that $\SW \subseteq \REG \subseteq \PW$; also see Figure~\ref{fig:weight-inclusion-structure}.} Our final classes of weight
functions are based on $\K$-Annotators. More precisely, we consider
weighted \vset-automata and unambiguous weighted \vset-automata over the tropical semiring
\tropicalSR and the numerical semiring \numericalSR.\footnote{One can also
  consider the tropical semiring with max/plus, in which case the complexity
  results are analogous to the ones we have for the tropical semiring with
  min/plus, with \minp and \maxp interchanged.} Formally, let $\REG \eqdef
\regtrop \cup \regnum$ be the class of all annotators over the tropical or
numerical semiring. A \emph{regular} (\REG) weight function $\w$ is represented
by a weighted \vset-automaton $W$ and defines $\w(\doc,\tup)
= \repspnr{W}(\doc,\projectTup{\tup}{\vars(W)})$. Furthermore, as for constant
width weight functions, we assume that the variables used by $W$ are always
present in $\tup$, that is, $\vars(W) \subseteq \vars(\tup)$.

The set of \emph{unambiguous regular} (\UREG) weight functions is the subset of
\REG that is represented by unambiguous weighted \vset-automata, that
is $\UREG \eqdef \uregtrop \cup \uregnum$.

\begin{example}Figure~\ref{fig:threeDigitAutomaton} gives an unambiguous weighted
  \vset-automaton over the tropical semiring that extracts the values of
  three-digit natural numbers from text. It can easily be extended to extract
  natural numbers of up to a constant number of digits by adding nondeterminism.
  Likewise, it is possible to extend it to extract weights as in
  Example~\ref{ex:MainExampleAggregates}. If a single variable captures a list
  of numbers, similar to $\doc_{\mspan{32}{37}} = 10{-}15$, one may use
  ambiguity to extract the minimal number represented in this range.
\qed\end{example}

\begin{figure}
\centering
    \begin{tikzpicture}[>=latex]
\node[state, initial] (q0) at (0,0){$q_0$};
      \node[state] (q1) at ($(q0)+(2.2,0)$) {$q_1$};
      \node[state] (q2) at ($(q1)+(2.2,0)$) {$q_2$};
      \node[state] (q3) at ($(q2)+(2.2,0)$) {$q_3$};
      \node[state] (q4) at ($(q3)+(2.2,0)$) {$q_4$};
      \node[state, accepting] (q5) at ($(q4)+(2.2,0)$) {$q_5$};

      \draw (q0) edge[->, loop above] node[above] {$\alphabet;0$}(q0);
      \draw (q0) edge[->] node[above] {$\vop{x};0$}(q1);
      
      \draw (q1) edge[->, bend left=30] node[above] {$1;100$}(q2);
      \node (q12) at ($(q1)+(1,0.1)$) {$\vdots$};
      \draw (q1) edge[->, bend right=30] node[below] {$8;800$}(q2);
      \draw (q1) edge[->, bend right=75] node[below] {$9;900$}(q2);

      \draw (q2) edge[->, bend left=75] node[above] {$0;0$}(q3);
      \draw (q2) edge[->, bend left=30] node[above] {$1;10$}(q3);
      \node (q23) at ($(q2)+(1,0.1)$) {$\vdots$};
      \draw (q2) edge[->, bend right=30] node[below] {$8;80$}(q3);
      \draw (q2) edge[->, bend right=75] node[below] {$9;90$}(q3);

      \draw (q3) edge[->, bend left=75] node[above] {$0;0$}(q4);
      \draw (q3) edge[->, bend left=30] node[above] {$1;1$}(q4);
      \node (q34) at ($(q3)+(1,0.1)$) {$\vdots$};
      \draw (q3) edge[->, bend right=30] node[below] {$8;8$}(q4);
      \draw (q3) edge[->, bend right=75] node[below] {$9;9$}(q4);
      
      \draw (q4) edge[->] node[above] {$\vcl{x};0$}(q5);
      \draw (q5) edge[->, loop above] node[above] {$\alphabet;0$}(q5);
      
    \end{tikzpicture}
    \caption{An unambiguous weighted \vset-automaton over the tropical
      semiring with initial state $q_0$ (with weight $0$) and accepting state
      $q_5$ (with weight $0$), extracting three-digit natural numbers captured
      in variable $x$. Recall that, over the tropical semiring, the weight of a
      run is the sum of all its edge weights.\label{fig:threeDigitAutomaton} }
\end{figure}

Our results for regular and unambiguous regular weight functions are similar to \SW when it comes to \minp, \maxp, \sump, and \avgp. The
main difference is that, depending on the semiring, we require more unambiguity.
For instance, for the tropical semiring, one needs unambiguity of the regular
weight function for \maxp. For \sump and \avgp one needs unambiguity for
\emph{both} the spanner and the regular weight function to achieve tractability.
Contrary, over the numerical semiring, one needs unambiguity of the regular
weight function for \minp and \maxp, whereas for \sump and \avgp unambiguity of
the spanner is sufficient for tractability. For \quantp, the situation is
different from \SW in the sense that regular weight functions render the problem
intractable. We refer to Table~\ref{tab:overview} for an overview.

\subsection{Approximation}
In the cases where exact computation of the aggregate problem is intractable, we
consider the question of approximation. It turns out that there exist FPRAS's in
two settings that we believe to be interesting. Firstly, in the case of \sump
and \avgp and constant-width weight functions, the restriction of unambiguity in
the spanner can be dropped if the weight function uses only nonnegative weights.
Secondly, although \quantp is \sharpp-hard under Turing reductions for general
\fvsa, it is possible to \emph{positionally} approximate the \quantl element in
an \fpras-like fashion, even with the very general polynomial-time weight
functions. We discuss this problem in more detail in Section~\ref{approx}.

 \section{Preliminary Results}
\label{technical}
In this section, we give basic results for document spanners and weight
functions that we use throughout this article.

\subsection{Known Results on \texorpdfstring{$\srd$-Annotators}{K-Annotators}}
We begin by recalling some known results on $\srd$-annotators.

\begin{proposition}[{Doleschal et al.~\cite[Proposition
  6.1]{DoleschalKMP-lmcs22}}]\label{prop:epsilonRemoval} 
  For every weighted \vset-automaton $A$ there is an equivalent weighted
  \vset-automaton $A^\prime$ that has no $\varepsilon$-transitions. This
  automaton $A^\prime$ can be constructed from $A$ in polynomial time.
  Furthermore, $A$ is functional if and only if $A^\prime$ is functional.
\end{proposition}

The following Theorem follows directly from
Doleschal et al.~\cite[Theorem 6.4]{DoleschalKMP-lmcs22} and
Doleschal~\cite[Theorem~5.5.4, Lemma~5.5.5 and Lemma~5.5.9]{DoleschalThesis}.
\begin{theorem}\label{theo:closed-algebra}
  Let $A_1,A_2 \in \regk$ be $\K$-weighted functional \vset-automata and $X \subseteq
  \vars(A_1)$. Then, $A_\pi, A_\cup,A_{\join} \in \regk$ can be constructed in
  polynomial time, such that
  \begin{align*}
    \repspnrw{A_\cup} &= \repspnrw{A_1} \cup \repspnrw{A_2}\\
    \repspnrw{A_\pi} &= \pi_X \repspnrw{A_1}\\
    \repspnrw{A_{\join}} &= \repspnrw{A_1} \join \repspnrw{A_2}.
  \end{align*}
  Furthermore, $A_{\join} \in \uregk$, if $A_1, A_2 \in \uregk$, and
  $A_\cup\in\uregk$ if $A_1,A_2 \in \uregk$ and $\reflang(A_1) \cap
  \reflang(A_2) = \emptyset$.
\end{theorem}

\subsection{Relative Expressiveness of Weight Functions}\label{sec:weight-function-expressiveness}
We first show that every constant-width weight function is also an
unambiguous regular weight function.

\begin{proposition}\label{prop:swInUreg}$\SW \subseteq \uregnum \cap \uregtrop$.
\end{proposition}
\begin{proof}
  Let $\w \in \SW$ be a constant-width weight function, represented by a
  $\rationals$-weighted string relation $R$ over $X$, that is, tuples in $R$ map variables to strings. We begin by showing that $\w \in
  \uregnum$. Let $X = \{x_1,\ldots,x_n\}$. We construct a $\numerical$-annotator
  $W$ representing $\w$.  We define an unambiguous \vset-automaton
  $A_\tup$, for every tuple $\tup \in R$, such that $\tupu \in
  \repspnrk{A_\tup}{\B}(\doc)$ if and only if $\doc_\tupu = \tup$. Let $\tup \in
  R$. For every $x \in X$, let $w_x$ be the word $\tup(x)$ and let
  \[
    A^x_\tup \eqdef \alphabet^* \cdot \vop{x} w_x \vcl{x} \cdot \alphabet^*\;,
  \]
  that is, $A^x_\tup$ matches the string $\tup(x)$ in variable $x$ and outputs the corresponding $\{x\}$-tuple with the span. Since our definition of unambiguity requires one run per \emph{ref-word} in the language, it is easy to see that such an unambiguous $A^x_\tup$ exists.
  Furthermore, 
  \[
    A_\tup \eqdef A_\tup^{x_1} \join \cdots \join A_\tup^{x_n} \;,
  \]
  which is unambiguous due to Theorem~\ref{theo:closed-algebra}.  
  
  We define $W_\tup$ as the unambiguous $\numerical$-weighted
  \vset-automaton such that
  \[
    \repspnrk{W_\tup}{\numerical}(\doc,\tupu) =
    \begin{cases}
      R(\tup) & \text{if } \doc_\tupu = \tup \\
      \srzero & \text{otherwise.}
    \end{cases}
  \]
  This can be achieved by interpreting $A_\tup$ as a $\numerical$-weighted
  \vset-automaton, where all edges have weight $\srone$, the final weight
  function assigns weight $\srone$ to all accepting states, and the initial
  weight function assigns weight $R(\tup)$ to the initial state of $A_\tup$. We
  finally define $W$ as the union of all $W_\tup$. That is,
  \[
    W = \bigcup\limits_{\tup \in R} W_\tup\;.
  \]
  We observe that, by Theorem~\ref{theo:closed-algebra}, $W$ must be unambiguous,
  as all $W_\tup$ are unambiguous and the ref-word languages of the automata
  $W_\tup$ are pairwise disjoint.

  Recall that $\repspnrk{W}{\numerical}(\doc,\tup) = \srzero = 0$ if there is no
  run of $W$ on $\refWordFrom{\doc}{\tup}$, i.e. $\doc_\tup \notin R$.
  Therefore, $\repspnrk{W}{\numerical}(\doc,\tup) = R(\doc_\tup)$ as desired.

  The proof for $\SW \subseteq \uregtrop$ follows the same lines. However, the
  zero element of the tropical semiring is $\infty$, which implies that the
  automaton $W$ must have exactly one run $\rn$ for every tuple $\tup$, even if
  $\w(\doc,\tup) = 0$. To this end, let $W_\tup$ be as defined before, but
  interpreted over the tropical semiring. We construct an unambiguous 
  $\tropical$-weighted \vset-automaton $W_{\overline{R}}$, such that
  $\repspnrk{W_{\overline{R}}}{\tropical}(\doc,\tup) = 0$ if $\doc_{\tup} \notin R$
  and $W_{\overline{R}}$ has no run for $\tup$ otherwise. We observe that
  $R$ is a recognizable string relation.\footnote{A $k$-ary string relation is
    recognizable if it is a finite union of Cartesian products $L_1 \times
    \cdots \times L_k$, where each $L_i$ is a regular language. Note that $R$ is
    recognizable as it is the union over all tuples $\tup \in R$, where each
    tuple is represented by the Cartesian product $\{\tup(x_1)\} \times \cdots
    \times \{\tup(x_n)\}$ with $\vars(\tup) = \{x_1,\ldots,x_n\}$.} Therefore,
  due to Doleschal et al.~\cite[Theorem 6.11]{DoleschalKMP-lmcs22}, there is
  a document spanner $A_{R}$, with $\tup \in \toSpanner{A_R}(\doc)$ if and only
  if $\doc_\tup \in R$. Furthermore, let $A_{\overline{R}}$ be the complement of
  $A_R$, that is, $\tup \in \toSpanner{A_{\overline{R}}}(\doc)$ if and only if
  $\doc_\tup \notin R$. Note that $A_{\overline{R}} \in \vsa$ as regular document
  spanners are closed under difference (cf. Fagin et al.~\cite[Theorem
  5.1]{FaginKRV15}).
  By Doleschal et al.~\cite[Lemma 4.5]{DoleschalKMNN19-journal},
  we can assume \mbox{w.l.o.g.}
  that $A_{\overline{R}} \in \ufvsa$. Let $W_{\overline{R}}$ be $A_{\overline{R}}$,
  interpreted as $\tropical$-weighted \vset-automaton, that is, each transition,
  initial and final state gets weight $\srone = 0$. Note that, due to $A_{\overline
    R} \in\ufvsa$, $W_{\overline{R}}$ is unambiguous. It follows that
  $\repspnrk{W_{\overline{R}}}{\tropical}(\doc,\tup) = 0$ if $\doc_{\tup} \notin R$
  and $W_{\overline{R}}$ has no run for $\tup$ otherwise. Let
  \[
    W = W_{\overline{R}} \cup \bigcup\limits_{\tup \in R} W_\tup\;.
  \]
  Again, we observe that, by Theorem~\ref{theo:closed-algebra}, $W$ must be
  unambiguous as all involved automata are unambiguous and their ref-word
  languages are pairwise disjoint.
  Furthermore,
  \[
    \repspnrk{W}{\tropical}(\doc,\tup) =
    \begin{cases}
      R(\doc_\tup) & \text{if } \doc_\tup \in R \\
      0 & \text{otherwise.}
    \end{cases}
  \]
  Therefore, $\repspnrk{W}{\tropical}(\doc,\tup) = R(\doc_\tup)$ as desired.
\end{proof}

We now observe that every regular weight function is a polynomial-time weight function. Indeed, 
given a document $\doc$ and a \doc-tuple $\tup$, the weight $\w(\doc,\tup)$ for a regular weight function $w$ can
be computed in polynomial time (cf. Doleschal~\cite[Theorem
5.6.1]{DoleschalThesis}).
\begin{observation}\label{obs:reg-poly}$\REG \subseteq \PW$.
\end{observation}
To summarize, we provide the inclusion structure of the classes of weight functions we consider in Figure~\ref{fig:weight-inclusion-structure}. All inclusions that do not have a number hold by definition.

\begin{figure}[t]
    \centering
    \resizebox{\linewidth}{!}{
    \begin{tikzpicture}
    \node (cN) at (0,0) {$\SW_{\nat}$};
    \node (cQ) at (2.5,0) {$\SW_{\Q^+}$};
    \node (c) at (5,0) {$\SW$};
    \node (uregQ) at (7,1) {$\uregnum$};
    \node (uregT) at (7,-1) {$\uregtrop$};
    \node (ureg) at (9,0) {$\UREG$};
    \node (regQ) at (10,1) {$\regnum$};
    \node (regT) at (10,-1) {$\regtrop$};
    \node (reg) at (12,0) {$\REG$};
    \node (poly) at (14,0) {$\PW$};
    
\path 
    (cN) edge[draw=none] node {$\subseteq$} (cQ)
    (cQ) edge[draw=none] node {$\subseteq$} (c)
    (c) edge[draw=none] node (c-uregQ) [sloped] {$\subseteq$} (uregQ)
    (c) edge[draw=none] node (c-uregT) [sloped] {$\subseteq$} (uregT)
    (uregQ) edge[draw=none] node[sloped] {$\subseteq$} (ureg)
    (uregT) edge[draw=none] node[sloped] {$\subseteq$} (ureg)
    (ureg)  edge[draw=none] node[near start,sloped] {$\subseteq$} (reg)
    (uregQ) edge[draw=none] node {$\subseteq$} (regQ)
    (uregT) edge[draw=none] node {$\subseteq$} (regT)
    (regQ) edge[draw=none] node[sloped] {$\subseteq$} (reg)
    (regT) edge[draw=none] node[sloped] {$\subseteq$} (reg)
    (reg) edge[draw=none] node (reg-poly) {$\subseteq$} (poly)
    ;
    
    \node at ($(c-uregQ)+(-.35,.35)$) {\tiny \eqref{prop:swInUreg}};
    \node at ($(c-uregT)+(-.35,-.35)$) {\tiny \eqref{prop:swInUreg}};
    \node at ($(reg-poly)+(0,.35)$) {\tiny \eqref{obs:reg-poly}};
    \end{tikzpicture}}
    \caption{Inclusion structure of our considered weight functions}
    \label{fig:weight-inclusion-structure}
\end{figure}

\subsection{Preliminary Results on Document Spanners} 
We will also need some preliminary results concerning the number of possible
spans over a document $\doc$.

\begin{lemma}\label{lemma:countSpans}Given a document $\doc$, the number of spans over $\doc$ is polynomial in the
  size of $\doc$. In particular, $|\spans(\doc)| = \frac{(|\doc|+1) \cdot
    (|\doc|+2)}{2}$, for every $\doc \in \docs$.
\end{lemma}
\begin{proof}
  For a span $\spanij,$ let $\ell = j-i$ be the length of the span. It is easy
  to see that for every document $\doc$, there is exactly one span of length
  $|\doc|$, two spans of length $|\doc|-1$, three spans of length $|\doc|-2$,
  etc. Thus, there are $1 + 2 + \dots + (|\doc|+1) = \frac{(|\doc|+1) \cdot
    (|\doc|+2)}{2}$ spans over a document $\doc$, concluding the proof.
\end{proof}

It follows directly that the maximal number of tuples, extracted by a document
spanner is exponential in the size of the spanner. 

\begin{corollary}\label{cor:tuplesPerSpanner}Let $A \in \fvsa$ be a \vset-automaton and $\doc \in \docs$ be a document.
  Then $\spcountagg \leq |\spans(\doc)|^{|\vars(A)|} = \Big(\frac{(|\doc|+1)\cdot
    (|\doc|+2)}{2}\Big)^{|\vars(A)|}\;$.
\end{corollary}

As we see next, given a number of variables, a document $\doc$, and a number $k$
of tuples, we can construct an unambiguous  \vset-automaton $A$ and a
document $\doc'$ such that $A$ extracts exactly $k$ tuples on $\doc'$.

\begin{lemma}\label{lem:spannerWithKtuples}Let $X \eqdef \{x_1,\ldots,x_v\} \in \svars$ be a set of variables, $\doc \in
  \docs$ be a document, and $0 \leq k \leq |\spans(\doc)|^{|X|}$. Then there is
  a \vset-automaton $A \in \ufvsa$ with $\vars(A) = X$ and a document
  $\doc^\prime \in \docs$ such that $|\toSpanner{A}(\doc^\prime)| = k$.
  Furthermore, $A$ and $\doc^\prime$ can be constructed in time polynomial in
  $|X|$ and $\doc$.
\end{lemma}
\begin{proof}
  We observe that the statement holds for $k = 0$. Therefore we assume,
  \mbox{w.l.o.g.}, that $1 \leq k \leq |\spans(\doc)|^{v}$. 

  We begin by proving the statement for $|X| = 1$. Let $1 \leq k \leq
  |\spans(\doc)|$. Recalling the proof of Lemma~\ref{lemma:countSpans}, we
  observe that $k$ can be written as a sum $k = k_1 + \cdots + k_n$ of $n \leq
  |\doc|+1$ different natural numbers with $0 \leq k_1 < \cdots < k_n \leq
  |\doc|+1$. We construct an automaton $A_k \in \ufvsa$, which consists of $n$
  branches, corresponding to $k_1,\ldots,k_n$. On document $\doc$, the branch
  corresponding to $k_i$ selects all spans of length $\ell_i \eqdef |\doc|+1 -
  k_i$. Each of these branches can be constructed
  as an unambiguous \vset-automaton $A_{k_i} \eqdef \alphabet^* \cdot
  \vop{x} \alphabet^{\ell_i} \vcl{x} \cdot \alphabet^*$. We observe that there are exactly
  $k_i$ spans over $\doc$ with length $\ell_i$, and therefore
  $|\toSpanner{A_{k_i}}(\doc)| = k_i$. The automaton $A_k$ is defined as
  \[
    A_k \eqdef A_{k_1} \cup \dots \cup A_{k_n}\;.
  \]
  It is straightforward to verify that all automata $A_{k_i}$ are unambiguous. Thus, since the ref-word languages of all $A_{k_i}$ are pairwise disjoint, it holds that $A_k\in \ufvsa$ (cf.\ 
  Theorem~\ref{theo:closed-algebra}). Furthermore, we observe that
  \[
    |\toSpanner{A_k}(\doc)| = |\toSpanner{A_{k_1}}(\doc)| + \cdots +
    |\toSpanner{A_{k_n}}(\doc)| = k_1 + \cdots + k_n = k\;.
  \]

  It remains to show the statement for $v \eqdef |X| > 1$. Let $\# \notin
  \alphabet$ be a new alphabet symbol. We build upon the encoding for $|X| = 1$.
  That is, for every $1 \leq k \leq |\spans(\doc)|$, let $A_k^x$ be the automaton
  $A_k$, using variable $x$, as defined previously. We observe that every $1 \leq
  k \leq |\spans(\doc)|^v$ has an encoding $k = k_1 \cdots k_{v}$ in base
  $|\spans(\doc)|$ of length $v$. The document $\doc^\prime$ consists of $v$
  copies of $\doc \cdot \#$, more formally,
  \[
    \doc^\prime \eqdef (\doc \cdot \#)^v\;.
  \]

  For every $1 \leq i \leq v$, we construct an automaton $A^\prime_{k_i}$,
  which selects exactly $k_i \cdot |\spans(\doc)|^{v-i}$ tuples over document
  $\doc'$. More formally, 
  \[
    A^\prime_{k_i} \eqdef \doc \cdot \vop{x_1} \# \vcl{x_1} \cdot \doc \cdot
    \vop{x_2} \# \vcl{x_2} \cdots \doc \cdot \vop{x_{i-1}} \# \vcl{x_{i-1}} \cdot A_{k_i}^{x_i} \cdot \# \cdot
    A_{|\spans(\doc)|}^{x_{i+1}} \cdot \# \cdots \# \cdot
    A_{|\spans(\doc)|}^{x_v} \cdot \# \;.
  \]
  The automaton $A^\prime_k$ is then defined as the union of all
  $A^\prime_{k_i}$, that is,
  \[
    A^\prime_k \eqdef A^\prime_{k_1} \cup \dots \cup A^\prime_{k_v}\;.
  \]
  We observe that $A^\prime_{k_i} \in \ufvsa$ and due to the ref-word languages
  of all $A^\prime_{k_i}$ being pairwise disjoint, $A^\prime_k \in \ufvsa$ (cf.
  Theorem~\ref{theo:closed-algebra}). Furthermore, we observe that
  \[
    |\toSpanner{A^\prime_k}(\doc^\prime)| =
    |\toSpanner{A^\prime_{k_1}}(\doc^\prime)| + \cdots +
    |\toSpanner{A^\prime_{v}}(\doc)^\prime| = k_1 + \cdots + k_n = k\;.
  \]
  This concludes the proof. 
\end{proof}

 \section{Constant-Width Weight Functions}
\label{cwidth}
We begin this section by showing that \minp and \maxp are tractable for
constant-width weight functions. The reason for their tractability is that, for
a constant number of variables $\xvars \subseteq \vars(A)$, the spans
associated to $\xvars$ in output tuples can be computed in polynomial time.
Building upon Corollary~\ref{cor:tuplesPerSpanner}, we show that \minp and
\maxp are in \fp for constant-width weight functions and 
\vset-automata. We immediately have:
\begin{theorem}\label{thm:minmaxfp}$\minp[\fvsa,\SW]$ and $\maxp[\fvsa,\SW]$ are in \fp.
\end{theorem}
\begin{proof}
  Let $A \in \fvsa$, $\doc\in\docs$, $\xvars \subseteq \vars(A)$ with $|\xvars|
  \leq c$, and $\w \in \SW$ be given as a $\rationals$-weighted string relation $R$ over
  $\xvars$. We first show that the set $\{\pi_\xvars\tup \mid \tup \in
  \toSpanner{A}(\doc)\}$ can be computed in time polynomial in the sizes of $A$
  and $\doc$.
  
  We observe that, per definition of projection for document spanners (Section~\ref{sec:projection}),
  $\{\pi_\xvars\tup \mid \tup \in \toSpanner{A}(\doc)\} =
  \big(\pi_\xvars(\toSpanner{A})\big)(\doc)$. Since $A$ is functional (which we assume for \vset-automata throughout this article),  a
  \vset-automaton for $\pi_\xvars(\toSpanner{A})$ can be computed in polynomial
  time (cf. Freydenberger et al.~\cite[Lemma 3.8]{FreydenbergerKP18}). Due to
  $|\xvars| \leq c$, it follows from Corollary~\ref{cor:tuplesPerSpanner} that
  there are at most polynomially many tuples in
  $\big(\pi_\xvars(\toSpanner{A})\big)(\doc)$. Thus, the set $\{\pi_\xvars\tup
  \mid \tup \in \toSpanner{A}(\doc)\}$ can be materialized in polynomial time.
  
  In order to compute $\minp$ and $\maxp$, a polynomial time algorithm can
  iterate over all tuples $\tup$ in $\{\pi_\xvars\tup \mid \tup \in
  \toSpanner{A}(\doc)\}$, evaluate $R(\doc,\tup)$ and maintain the minimum and
  the maximum of these numbers.
\end{proof}

In order to calculate aggregates like $\spsum, \spavg,$ or $\spquant$, it is not
sufficient to know which weights are assigned, but also the multiplicity of each
weight is necessary. Recall that counting the number of output tuples is
tractable if the \vset-automaton is unambiguous
(Theorem~\ref{thm:countResults}) and $\spanl$-complete in general.
We now show that we can achieve tractability of the mentioned
aggregate problems if the \vset-automaton is unambiguous. The
reason is that we can compute in polynomial time the multiset $\suptmp \eqdef
\multiset{\pi_\xvars\tup \mid \tup \in \toSpanner{A}(\doc)}$, where we represent
the multiplicity of each tuple $\tupu$ (i.e., the number of tuples $\tup \in
\toSpanner{A}(\doc)$ such that $\pi_\xvars\tup = \tupu$) in binary.

\begin{lemma}\label{lemma:wfpconstufvsamultiset}Given a \vset-automaton $A$ and a document \doc, the multiset $\suptmp$ can be
  computed in $\fp$ if $A \in \ufvsa$.
\end{lemma}
\begin{proof}
  The procedure is given as Algorithm~\ref{alg:weightoccurences}. It is
  straightforward to verify that the algorithm is correct. Due to 
  Corollary~\ref{cor:tuplesPerSpanner}, the set
  $\pi_\xvars(\toSpanner{A})(\doc)$ is at most of polynomial size. Furthermore,
  the automaton $A_{\refWordFrom{\doc}{\tup}} \eqdef \refWordFrom{\doc}{\tup}
  \in \ufvsa$ can be constructed in polynomial time and due to
  Theorem~\ref{theo:closed-algebra} an unambiguous \vset-automaton
  for $A_\tup$ can be computed in polynomial time as well. By
  Theorem~\ref{thm:countResults}, each iteration of the for-loop also only
  requires polynomial time. Thus, the whole algorithm terminates after
  polynomially many steps.
\end{proof}

\begin{algorithm}[t]
\DontPrintSemicolon \KwIn{An unambiguous \vset-automaton $A \in
    \ufvsa$, a document $\doc \in \docs$.}
\KwOut{The multiset $\suptmp$.} $\suptm \gets \multiset{}$\; $\supts \gets
  \suptsp$\; \For{$\tup \in \supts$}{ $A_\tup \gets A \join
    A_{\refWordFrom{\doc}{\tup}}$
    \Comment*[r]{\textrm{$A_{\refWordFrom{\doc}{\tup}}$ is the \ufvsa that only accepts
        $\refWordFrom{\doc}{\tup}$.}} $\suptm(\pi_\xvars\tup) \gets
    \spcount(\toSpanner{A_\tup},\doc)$\label{alg:weightoccurencesCt}\; }
  \textbf{output} $\suptm$\;
\caption{Calculate the multiset $\suptmp$.}
\label{alg:weightoccurences}
\end{algorithm}

It follows that all remaining aggregate functions can be efficiently computed if
the spanner is given as an unambiguous \vset-automaton.

\begin{theorem}\label{thm:aggregatefp}For every $0 \leq q \leq 1$, $\sump[\ufvsa,\SW]$, $\avgp[\ufvsa,\SW]$, and $\quantp[\ufvsa, \SW]$ are in \fp.
\end{theorem}
\begin{proof}
  Let $A \in \ufvsa$ be a \vset-automaton, $\doc \in \docs$ be a document, $\w
  \in \SW$ be a weight function, represented by a $\rationals$-weighted string relation $R$ over
  $\xvars$. Due to Lemma~\ref{lemma:wfpconstufvsamultiset} the multiset
  $\suptmp$ can be computed in polynomial time. Thus one can compute the
  multiset $W \eqdef \multiset{R(\doc_\tup) \mid \tup \in \suptmp}$ in
  polynomial time. It is straightforward to compute the aggregates in
  polynomial time from $W$.
\end{proof}

We conclude this section by showing that $\spsum$, $\spavg$, and $\spquant$ are
not tractable, if the spanner is given as a \vset-automaton.

\begin{theorem}\label{thm:sw-sum-cnf}$\sump[\fvsa,\SW]$ is \sharpp-hard, even if $\w$ is represented by the
  $\rationals$-Relation $R$ over $\{x\}$ with
  \[
    R(\doc) \eqdef
    \begin{cases}
      1 & \text{if } \doc = 1\\
      -1 & \text{if } \doc = -1\\
      0 & \text{otherwise.}
    \end{cases} 
  \]
\end{theorem}
\begin{proof}
  We will give a reduction from the $\#\text{CNF}$ problem, which is \sharpp-complete under
  parsimonious reductions. To this end, let $\phi$ be a Boolean formula in CNF
  over variables $x_1,\ldots,x_n$ and let $\w \in \SW$ be the weight function
  which is represented by the $\rationals$-Relation $R$, which is as defined in
  the theorem statement.
  
  We construct a \vset-automaton $A \in \fvsa$ and a document $\doc \eqdef a^n
  \cdot - \cdot 1$, such that $\spsum(\toSpanner{A}, \doc, \w) = c$, where $c$ is the
  number of variable assignments which satisfy $\phi$.

  We begin by defining two \vset-automata $A_1,A_{-1}$, with $\vars(A_1) =
  \vars(A_{-1}) = \{x_1,\ldots,x_n,x\}$. Slightly overloading notation, we
  define both automata by regex formulas.

  The automaton $A_{1}$ selects exactly $2^n$ tuples on document $\doc$, all of
  which get assigned weight $1$ by $\w$. More formally (using $\lor$ to denote regular expression disjunction),
\[
    A_1 \eqdef (\vop{x_1}a\vcl{x_1} \; \lor \; \vop{x_1}\varepsilon\vcl{x_1} a)\cdots (\vop{x_n} a \vcl{x_n} \; \lor
    \; \vop{x_n} \varepsilon \vcl{x_n} a) \; - \; \vop{x} 1 \vcl{x}\;.
  \]
  Therefore, $\spsum(\toSpanner{A_1},\doc, \w) =
  \spcount(\toSpanner{A_1},\doc) = 2^n$.

  We use a similar encoding as Doleschal et al.~\cite[Theorem
  5.4]{DoleschalKMP-lmcs22} to encode variable assignments into tuples. That
  is, each variable $x_i$ of $\phi$ is associated with a corresponding capture
  variable $x_i$ of $A_{-1}$. With each assignment $\tau$ we associate the tuple
  $\tup_\tau$, such that
  \[
    \tup_\tau(x_i) \eqdef
    \begin{cases}
      \mspan{i}{i} & \text{ if } \tau(x_i) = 0, \text{ and}\\
      \mspan{i}{i+1} &\text{ if } \tau(x_i) = 1\;.\\
    \end{cases}
  \]
  We construct the automaton $A_{-1}$ as a regex formula $\alpha$, such that
  there is a one-to-one correspondence between the non-satisfying assignments
  for $\phi$ and tuples in $\toSpanner{\alpha}(\doc)$. More formally, for each
  clause $C_j$ of $\phi$ and each variable $x_i$, we construct a regex-formula
  \[
    \alpha_{i,j} \eqdef
    \begin{cases}
      x_i\{\varepsilon\}\cdot a& \text{ if $x_i $ appears in $C_j$,} \\
      x_i\{a\} & \text{ if $\neg x_i$ appears in $ C_j$,} \\
      (x_i\{\varepsilon\}\cdot a) \lor x_i\{a\} & \text{ otherwise.} \\
    \end{cases}
  \]
  Consequently, we define 
$\alpha_j \eqdef \alpha_{1,j}\cdots\alpha_{n,j}\cdot \vop{x}-1\vcl{x}$.

  For example, if we use variables $x_1,x_2,x_3,x_4$ and $C_j=x_1\lor
  x_3\lor\neg x_4$ is a clause, then
\[
    \alpha_j=\vop{x_1}\varepsilon\vcl{x_1} \; a \; (\vop{x_2}\varepsilon \vcl{x_2}\; a \lor
    \vop{x_2}a\vcl{x_2}) \; \vop{x_3}\varepsilon\vcl{x_3}\; a\; \vop{x_4}a\vcl{x_4} \; \vop{x} -1\vcl{x}.
  \]

  We observe that $\tup \in \toSpanner{\alpha_j}(\doc)$ if and only if the
  variable assignment $\tau$ of $\phi$ with $\tup = \tup_\tau$ does not satisfy
  clause $C_j$.

  We finally define $\alpha\eqdef\alpha_1 \lor \cdots \lor \alpha_{m}$, that is,
  the disjunction of all $\alpha_{i}$ and $A_{-1}$ as the \vset-automaton
  corresponding to $\alpha$.\footnote{It is easy to verify that the automaton
    $A_{-1} \in \fvsa$ can be constructed in polynomial time from $\alpha$.}
  Therefore, $\spcount(\toSpanner{A_{-1}},\doc) = s$, where $s = 2^n-c$ is the
  number of variable assignments which do not satisfy $\phi$. Furthermore, per
  definition of $A_{-1}$ and $\w$, it follows that
  \[
    \spsum(\toSpanner{A_{-1}}, \doc, \w) = -1 \cdot s = -s\;.
  \]
  We finally define the \vset-automaton $A$ as the union of $A_1$ and $A_{-1}$.
  We observe that every tuple $\tup \in \toSpanner{A}(\doc)$ is either selected
  by $A_1$ (if $\doc_{\tup(x)} = 1$) or by $A_{-1}$ (if $\doc_{\tup(x)} = -1$),
  but never by both automata. Recall that $c$ is the number of assignments which
  satisfy $\phi$ and $s = 2^n - c$ is the number non-satisfying assignments of
  $\phi$. Therefore, we have that 
  \[
    \spAsumagg = \spsum(A_1,\doc,\w) + \spsum(A_{-1},\doc,\w) = 2^n + (-s) = 2^n
    - (2^n - c)= c\;.
  \]
  This concludes the proof.
\end{proof}

If the weights are restricted to natural numbers, \sump becomes \spanl-complete.
Note that we restrict weight functions to natural numbers, because \spanl is a
class of functions that return natural numbers. Allowing positive rational
numbers does not fundamentally change the complexity of the problems though. We
will see in Section~\ref{approx} that this enables us to approximate $\sump$
aggregates.

\begin{theorem}\label{thm:sw-sum-spanl}$\sump[\fvsa,\SW_{\nat}]$ is \spanl-complete, even if $\w$ is represented by the
  $\rationals$-Relation $R$ over $\{x\}$ with
  \[
    R(\doc) \eqdef
    \begin{cases}
      1 & \text{if } \doc = 1\\
      0 & \text{otherwise.}
    \end{cases} 
  \]
\end{theorem}
\begin{proof}
  Recall that a function $f$ is in \spanl, if there is an \NL Turing machine $M$
  such that $f(x) = |M(x)|$. Let $A \in \fvsa$ be a \vset-automaton, $\doc \in
  \docs$ be a document, and $\w \in \SW_{\Q_+}$ be a weight function. We define
  $M$ as the Turing machine, which guesses a \doc-tuple $\tup$ and checks
  whether $\tup \in \toSpanner{A}(\doc)$. If yes, $M$ computes the weight
  $\w(\doc,\tup)$, which can be done in \NL, since $\w$ is given by a
  $\rationals$-Relation. The Turing machine $M$ then branches into
  $\w(\doc,\tup)$ accepting branches. If $\tup \notin \toSpanner{A}(\doc)$, $M$
  rejects. Thus, $|M(A,\doc)| = \spsumagg$, and therefore
  $\sump[\fvsa,\SW_{\nat}]$ is in \spanl.

  For the lower bound, we give a reduction from $\countp[\fvsa]$, which is
  \spanl-complete (cf. Theorem~\ref{thm:countResults}). Let $A \in \fvsa$, $\doc \in
  \docs$. We assume, \mbox{w.l.o.g.}, that $1 \notin \alphabet$ and $x \notin
  \vars(A)$. We construct a document $\doc^\prime \eqdef \doc \cdot 1$ and a
  \vset-automaton $A^\prime \eqdef A \cdot \vop{x}1\vcl{x}.$ We observe that
  $\spsum(\toSpanner{A^\prime}, \doc^\prime, \w) = \spAcountagg$, concluding the
  proof.
\end{proof}

We conclude this section by showing that \avgp and \quantp are \sharpp-hard
under Turing reductions.
\begin{theorem}\label{thm:sw-quant-avg-hard}Let $0 < q < 1$ be a fixed number. The problems $\avgp[\fvsa,\SW_{\Q^+}]$ and
  $\quantp[\fvsa,\SW]$ are \sharpp-hard under Turing reductions, even if $\w$ is
  represented by the $\rationals$-Relation $R$ over $\{x\}$ with
  \[
    R(\doc) \eqdef
    \begin{cases}
      1 & \text{if } \doc = 1\\
      0 & \text{otherwise.}
    \end{cases} 
  \]
\end{theorem}
\begin{proof}
  Recall that $\countp[\fvsa]$ is \sharpp-hard under Turing reductions (Corollary~\ref{cor:countSharpp}). We begin
  by giving a Turing reduction from $\countp[\fvsa]$ to $\avgp[\fvsa,\SW]$. Let
  $A,\doc,$ and $\doc^\prime$ be as defined in the proof of
  Theorem~\ref{thm:sw-sum-spanl}. The \vset-automaton $A^\prime$ builds upon $A$
  but selects a single additional tuple $\tup$ with $\tup(x) =
  \spanFromTo{|\doc|+2}{|\doc|+2}$ for all variables. As we will see later, this
  tuple is used to calculate $\spAcountagg$ from
  $\spavg(\toSpanner{A^\prime},\doc^\prime)$. Let $\vars(A) =
  \{x_1,\ldots,x_n\}$. We define
\[
    A^\prime \;\eqdef\; (A \cdot \vop{x}1\vcl{x}) \;\lor\; (\doc\cdot 1\cdot \vop{x_1} \vop{x_2}
    \cdots \vop{x_n}  \vop{x} \varepsilon \vcl{x} \vcl{x_n}\cdots \vcl{x_2} \vcl{x_1})\;.
  \]
  Observe that, for all $\tup \in A^\prime(\doc^\prime)$ it holds that
  $\doc_{\tup(x)} = 1$ if and only if $\pi_{\vars(A)}\tup \in
  \toSpanner{A}(\doc)$. Thus, per definition of $A^\prime$ and $\w$,
  $\spsum(\toSpanner{A^\prime}, \doc^\prime, \w) = \spAcountagg$ and
  $\spcount(\toSpanner{A^\prime}, \doc^\prime) = \spAcountagg + 1$. Therefore,
  it holds that
  \[
    \spavg(\toSpanner{A^\prime},\doc^\prime, \w) =
    \frac{\spAcountagg}{\spAcountagg + 1}\;.
  \]
  Solving the equation for
  $\spAcountagg,$ we have that
  \[
    \spAcountagg = \frac{\spavg(\toSpanner{A^\prime},\doc^\prime,
      \w)}{1-\spavg(\toSpanner{A^\prime},\doc^\prime, \w)} \;.
  \]
  This concludes the proof that $\avgp[\fvsa,\SW_{\Q^+}]$ is \sharpp-hard under Turing
  reductions.

  It remains to show that $\quantp[\fvsa,\SW]$ is also \sharpp-hard under Turing
  reductions. Let $A \in \fvsa$ be a \vset-automaton and $\doc \in
  \docs$ be a document. We will show the lower bound for $q = \frac{1}{2}$ first
  and study the general case of $0 < q < 1$ afterwards. Let $x \notin \vars(A)$
  be a new variable. Let $0 \leq r \leq |\spans(\doc)|^{|\vars(A)|}$. By
  Lemma~\ref{lem:spannerWithKtuples} there is a \vset-automaton $A^\prime$ and a
  document $\doc^\prime$ with $\spcount(\toSpanner{A^\prime},\doc^\prime) =
  |\toSpanner{A^\prime}(\doc^\prime)| = r$. Let $0, 1 \notin \alphabet$ be a new
  alphabet symbol. Let $\doc_r = 0 \cdot \doc \cdot 1 \cdot \doc^\prime$ and
\[
    A_r = \big(\vop{x}0\vcl{x} \cdot A \cdot 1 \cdot \doc^\prime\big) \lor
    \big(0 \cdot \doc \cdot \vop{x}1\vcl{x} \cdot A^\prime\big)\;.
  \]
  Thus, $\spcount(\toSpanner{A_r},\doc_r) = \spcount(\toSpanner{A},\doc) +
  \spcount(\toSpanner{A^\prime},\doc^\prime).$ Recalling the definition of $\w$
  it holds, for every tuple $\tup \in \toSpanner{A_r}$, that $\w(\doc_r,\tup) = 1$
  if $\tup$ was selected by $A^\prime$ and $\w(\doc_r,\tup) = 0$ otherwise,
  i.e., $\tup$ was selected by $A$. Therefore,
  $\qquantl{\frac{1}{2}}(\toSpanner{A_r},\doc_r,\w) = 0$ if and only if
  $\spcount(\toSpanner{A},\doc) \geq \spcount(\toSpanner{A^\prime},\doc^\prime)
  = r$. Let $r_{\text{max}}$ be the biggest $r$ such that we have
  $\qquantl{\frac{1}{2}}(\toSpanner{A_r},\doc_r,\w) = 0$. Using binary search,
  we can calculate $r_{\text{max}}$ with a polynomial number of calls to an
  $\qquantp{\frac{1}{2}}$ oracle. Furthermore, due to
  $\spcount(\toSpanner{A},\doc) \in \nat$ and $R_{\text{max}}$ being maximal, it
  must hold that $\spcount(\toSpanner{A},\doc) = r_{\text{max}}$, concluding
  this part of the proof.
 
  The general case of $0 < q < 1$ follows by slightly adopting the above
  reduction. Let $q = \frac{a}{b}$ with $a,b \in \nat$ be given by its numerator
  and denominator. Observe that $b > a$ as $0 < \frac{a}{b} < 1$. Let
  $A^\prime,\doc^\prime$ be as above and let $c \eqdef
  \spcount(\toSpanner{A},\doc)$. The document $\doc_r$ consists of $a$ copies of
  $\doc$, separated by $0's$ and $(b-a)$ copies of $\doc^\prime$ separated by
  $1's$. Formally, $\doc_r = 0\cdot \doc_1 \cdot 0 \cdot \doc_2 \cdot 0 \cdots
  \doc_a \cdot 0\cdot 1 \cdot \doc^\prime_1 \cdot 1 \cdot \doc^\prime_2 \cdot 1
  \cdots \doc^\prime_{b-a}\cdot 1$, where each $\doc_i$ (resp. $\doc^\prime_i$)
  is a copy of $\doc$ (resp. $\doc^\prime$). Furthermore, let
\[
    A_r = \big(\alphabet_0^* \cdot \vop{x}0\vcl{x} \cdot A \cdot 0 \cdot \alphabet_0^*
    \cdot \alphabet_1^* \big) \lor \big(\alphabet_1^* \cdot \alphabet_1^* \cdot
    \vop{x}1\vcl{x} \cdot A^\prime \cdot 1 \cdot \alphabet_1^*\big)\;,
  \]
  where $\alphabet_{0} \eqdef \alphabet \cup \{0\}$ (resp. $\alphabet_1 \eqdef
  \alphabet \cup \{1\}$). Observe that $\w$ assigns $0$ to exactly $c \cdot a$
  tuples in $\toSpanner{A_r}(\doc_r)$ and $\spcount(\toSpanner{A_r},\doc_r) =
  c\cdot a + r \cdot (b-a)$. Thus, $\qquantl{\frac{a}{b}}(A_r,\doc_r,\w) = 0$
  if and only if $\frac{c \cdot a}{c\cdot a + r \cdot (b-a)} \geq
  \frac{a}{b}$. We now show that $c \geq r$ if and only if
  $\qquantl{\frac{a}{b}}(\toSpanner{A_r},\doc_r,\w) = 0$. Assume that $c \geq
  r$. Then,
  \[
    \frac{c \cdot a}{c \cdot a + r \cdot (b-a)} \geq \frac{c \cdot a}{c
      \cdot a + c \cdot (b-a)} = \frac{c \cdot a}{c \cdot b} = \frac{a}{b}\;.
  \]
  Therefore, $\qquantl{\frac{a}{b}}(\toSpanner{A_r},\doc_r,\w) = 0$. On the
  other hand, if $c < r$,
  \[
    \frac{c \cdot a}{c \cdot a + r \cdot (b-a)} < \frac{c \cdot a}{c
      \cdot a + c \cdot (b-a)} = \frac{c \cdot a}{c \cdot b} = \frac{a}{b}\;.
  \]
  Thus, $\qquantl{\frac{a}{b}}(\toSpanner{A_r},\doc_r,\w) = 1$.

  Recall that $c = \spcount(\toSpanner{A},\doc)$. As for $q = \frac{1}{2}$, let
  $r_{\text{max}}$ be the biggest $r$ such that
  $\qquantl{\frac{a}{b}}(\toSpanner{A_r},\doc_r,\w) = 0$. Using binary search,
  we can calculate $r_{\text{max}}$ with a polynomial number of calls to an
  $\qquantp{\frac{a}{b}}$ oracle. Again it holds that
  $\spcount(\toSpanner{A},\doc) = r_{\text{max}}$, concluding the proof.
\end{proof}

 \section{Polynomial-Time Weight Functions}
\label{pvar}

Before we study regular weight functions, we make a few observations on the very
general polynomial-time computable weight functions. For weight functions $\w
\in \PW$, we assume that $\w$ is represented as a Turing Machine $A$ that
returns a value $A(\doc,\tup)$ in polynomially many steps for some fixed
polynomial of choice (e.g., $n^2$).\footnote{Our complexity results are
  independent of the choice of this polynomial.} Furthermore, to avoid
complexity due to the need to verify whether $A$ is indeed a valid input (i.e., 
timely termination), we will assume that $\w(\doc,\tup) = 0$, if $A$ does not
produce a value within the allocated time.

We first observe that polynomial-time weight functions make all our aggregation
problems intractable, which is not surprising. In fact, all the lower bounds already hold for regular weight functions. 
\begin{theorem}\label{thm:pwintractableKarp}The problems $\minp[\ufvsa,\PW]$ and $\maxp[\ufvsa,\PW]$
  are \optp-hard. Furthermore, $\sump[\ufvsa,\PW]$ and $\avgp[\ufvsa,\PW]$ are
  \sharpp-hard.
\end{theorem}
\begin{proof}
  We will see later that these problems are already hard for weight functions in \REG, which are a subclass of \PW (Theorems~\ref{thm:minmaxRegintractable} and~\ref{thm:sumTropicalHard}).
\end{proof}

\begin{theorem}\label{thm:pwintractableCook}Let $0 < q < 1$. Then $\quantp[\ufvsa,\PW]$ is \sharpp-hard under Turing reductions.
\end{theorem}
\begin{proof}
  We will see later that the problem is already hard for $\UREG$ weight functions (Theorem~\ref{thm:quantUregHard}).
\end{proof}

We note
that all studied problems can be solved in exponential time, by first
constructing the relation $\toSpanner{A}(\doc)$, which might be of exponential
size, computing the weights associated to all tuples, and finally computing the
desired aggregate. 

\begin{theorem}\label{thm:pwfexptime}Let $0 < q < 1$. Then $\aggp[\fvsa,\PW]$ is in \fexptime for every $\aggp \in
  \{\minp,\maxp,\sump,\avgp,\quantp\}$.
\end{theorem}
\begin{proof}
  Let $A \in \fvsa$, $\doc \in \docs$, and $\w \in \PW$.
  The algorithm first computes the multiset
  \[
    W_{A,\doc,\w} \eqdef \multiset{\w(\doc,\tup) \mid \tup \in
      \toSpanner{A}(\doc)}\;,
  \]
  which might be exponentially large. It is easy to see that $W_{A,\doc,\w}$ can
  be computed in exponential time. Furthermore, it follows directly that
  $\aggp[\fvsa,\PW]$ is in \fexptime for every $\aggp \in
  \{\minp,\maxp,\sump,\avgp,\quantp\}$.
\end{proof}

Throughout this section, we do not study excessively whether we can give a more
precise upper bound than the general \fexptime upper bound. However, we
sometimes give such bounds. For instance, we are able to provide \optp and
\fptosp upper bounds if the weight functions return natural numbers (or integers
in the case of the \fptosp upper bounds).

\begin{theorem}\label{thm:maxInOptp}$\minp[\fvsa,\PW]$ and $\maxp[\fvsa,\PW]$ are in \optp if the weight function
  only assigns natural numbers.
\end{theorem}
\begin{proof}
  We only give the upper bound for \maxp. The proof for \minp is analogous. To
  this end, let $A \in \fvsa$, $\doc \in 
  \docs$, and $\w \in \PW$ be a weight function which only assigns natural
  numbers. The Turing Machine $N$ guesses a \doc-tuple
  \tup and accepts with output $0$ if $\tup \notin A(\doc)$. Otherwise, $N$
  computes the weight $\w(\doc,\tup)$ and accepts with output $\w(\doc, \tup)$.
  It is easy to see that the maximum output value of $N$ is exactly
  $\spmax(\toSpanner{A}, \doc, \w)$.
\end{proof}

In the following theorem we show that \sump, \avgp, and \quantp can be computed
in \fptosp if all weights are integers. The key idea is that, due to the
restriction to integer weights, we can compute the aggregates by multiple calls
to a \sharpp oracle. For instance for \sump, we define two weight functions,
$\w^+$ and $\w^-$, such that $\w^+$ computes the sum of all positive and $w^-$
the sum of all negative weights. Each of these sums can be computed by a single
call to a \sharpp oracle.

\begin{theorem}\label{thm:fpospUpperBounds}For every $0 \leq q \leq 1$, the problems $\sump[\fvsa,\PW]$, $\avgp[\fvsa,\PW]$, and
  $\quantp[\fvsa,\PW]$ are in $\fptosp$ if the weight function only assigns
  integers.
\end{theorem}
\begin{proof}
  We first prove that $\sump[\fvsa,\PW]$ is in \sharpp if the weight function
  only assigns natural numbers. We will use this as an oracle for the general
  upper bound. Let $A$ be a \vset-automaton, $\doc \in \docs$ be a document and
  $\w \in \PW$ be a weight function that only assigns natural numbers. A
  counting Turing Machine $M$ for solving the problem in $\sharpp$ would have
  $\w(\doc,\tup)$ accepting runs for every tuple in $A(\doc)$. More precisely,
  $M$ guesses a \doc-tuple $\tup$ over $\vars(A)$ and checks whether $\tup \in
  \toSpanner{A}(\doc)$. If $\tup \in \toSpanner{A}(\doc)$ and $\w(\doc,\tup) >
  0$, then $M$ branches into $\w(\doc,\tup)$ accepting branches, which it can do
  because $\w$ is given in the input as a polynomial-time deterministic Turing
  Machine. Otherwise, $M$ rejects. Per construction, $M$ has exactly
  $\w(\doc,\tup)$ accepting branches for every tuple $\tup \in
  \toSpanner{A}(\doc)$ with $\w(\doc,\tup) > 0$. Thus, the number of accepting
  runs is exactly $\sum_{\tup \in \toSpanner{A}(\doc)} \w(\doc,\tup) =
  \spAsumagg$.

  We now continue by showing that $\sump[\fvsa,\PW]$ is in $\fptosp$ if the
  weight function only assigns integers. Let $A$ be a \vset-automaton, $\doc \in
  \docs$ be a document, and $\w \in \PW$ be a weight function, which only
  assigns integers.

  We define two weight functions $\w^{+}, \w^{-} \in \PW$, such that
  \[
    \spsum(A,
    \doc, \w) = \spsum(A, \doc, \w^{+}) - \spsum(A, \doc, \w^{-})\;.
  \]
  Formally, we define the following two weight functions:
  \begin{align*}
    \w^{+}(\doc, \tup) & \eqdef
    \begin{cases}
      \w(\doc, \tup)\hphantom{-} & \text{if $\w(\doc,\tup) \geq 0$}\text{, and}\\
      0 & \text{otherwise;}
    \end{cases}\\
    \w^{-}(\doc, \tup) & \eqdef
    \begin{cases}
      -\w(\doc, \tup) & \text{if $\w(\doc,\tup) < 0$, and}\\
      0 & \text{otherwise.}
    \end{cases}
  \end{align*}

  Therefore, $\spsum(\toSpanner{A}, \doc, \w) = \spsum(\toSpanner{A}, \doc,
  \w^{+}) - \spsum(\toSpanner{A}, \doc, \w^{-})$ and the answer to
  $\sump[\spannerClass,\PW]$ can be obtained by taking the difference of the
  answers of two calls to the $\sump[\spannerClass,\PW]$ $\sharpp$ oracle. The
  upper bound for $\avgp[\fvsa, \PW]$ is immediate from the upper bound of
  $\sump[\fvsa, \PW]$ and Theorem~\ref{thm:countResults}. For the upper bound of
  $\quantp[\fvsa,\PW]$ we define the weight function
  \[
    \w_{\leq k}(\doc,\tup) =
    \begin{cases}
      1 & \text{if } w(\doc,\tup) \leq k \text{, and} \\
      0 & \text{otherwise.}
    \end{cases}
  \]
  Recall that
  \[
    \spquantagg \eqdef \min\bigg\{r \in \wCodom(\spanner,\doc,\w) \ \bigg| \ 
    \frac{|\set{\tup\in\spanner(\doc)\mid \w(\doc,\tup)\leq r}|}
    {|\spanner(\doc)|} \geq q\bigg\}\;.
  \]
  And therefore 
  \[
    \spquantagg = \min\bigg\{r \in \wCodom(\spanner,\doc,\w) \ \bigg| \ 
    \frac{\spsum(\toSpanner{A},\doc,\w_{\leq k})}
    {\spcount(\toSpanner{A},\doc)} \geq q\bigg\}\;.
  \]
  Thus, the upper bound of $\quantp[\fvsa,\PW]$ can be obtained by
  performing binary search, using the upper bound of
  $\sump[\fvsa, \PW]$ and Theorem~\ref{thm:countResults}.
\end{proof}

 \section{Regular Weight Functions}
\label{reg}
We now turn to \REG and \UREG weight functions. As we have shown in
Proposition~\ref{prop:swInUreg}, every \SW weight function can be translated
into an equivalent $\UREG$ weight function. Furthermore, the weight functions
which were used for the lower bounds can be represented by unambiguous
weighted \vset-automata of constant size. Therefore, all lower bounds
for \SW also hold for \UREG.

\subsection{Compact DAG Representation}
As we show next, aggregation problems for regular weight functions can
often be reduced to problems about paths on weighted \emph{directed acyclic
  graphs (DAGs)}, where the weights come from the semiring of the weight
function. To this end, let $\sr$ be a semiring. A \emph{$\srd$-weighted DAG} is
a DAG $D = (N,E)$, where $N$ is a set of nodes, $E \subseteq N
\times \srd \times N$ is a finite set of weighted edges, and $\dags$ (resp.,
$\dagt$) is a unique node in $N$ without incoming (resp., outgoing) edges. We
define $\len(e) = \ell$, where $e = (v,\ell,v') \in E$. Furthermore, we define
paths $p$ in the obvious manner as sequences of edges and the length $\len(p)$
of $p$ as the product ($\srtimes$) of the lengths of its edges. More formally, a
path
\[
  p \eqdef n_1\ell_1n_2 \cdots \ell_{n-1}n_j
\]
is a sequence of nodes $n_i \in N$ with $1 \leq i \leq j$ and
$(n_i,\ell_i,n_{i+1},) \in E$, for all $1 \leq i < j$, and the length
\[
  \len(p) \eqdef \ell_1 \srtimes \cdots \srtimes \ell_{j-1}\;.
\]
We denote the set of all paths in $D$ from $\dags$ to $\dagt$ by $\dagpaths$.

Given a document $\doc$, a \vset-automaton $A$ and a regular weight
function $\w \in \regk$, we will construct a DAG $D$ which plays the role of
a compact representation of the materialized intermediate result. The DAG $D$ is
obtained by a product construction between $A$, $W$, and $\doc$, such that every
path from $\dags$ to $\dagt$ corresponds to an accepting run of $W$ that
represents a tuple in $\toSpanner{A}(\doc)$. If $A$ and $W$ are unambiguous this
correspondence is actually a bijection.

\begin{lemma}\label{lem:dagComp}Let $\K \in \{\numerical,\tropical\}$ be either the numerical or the tropical
  semiring. Let $\doc$ be a document, $A\in\fvsa$, and $W$ be the 
  weighted \vset-automaton representing $\w \in \regk$. We can compute, in
  polynomial time, a $\K$-weighted DAG $D$, such that there is a surjective
  mapping $m$ from paths $p \in \dagpaths$ in $D$ to tuples $\tup \in
  \toSpanner{A}(\doc)$. Furthermore,
  \begin{enumerate}[label=(\arabic*)]
  \item the mapping $m$ is a bijection, if $A$ and $W$ are unambiguous,
    and\label{dag:bijection}
  \item $\w(\doc,\tup) = \boplus\limits_{p\in\dagpaths, m(p)=\tup}\len(p)$, for
    every $\tup \in \toSpanner{A}(\doc)$, if $A\in\ufvsa$ or $\K =
    \tropical$.\label{dag:weight}
  \end{enumerate}
\end{lemma}
\begin{proof}
  Let $\doc \in \docs$, $A \in \fvsa$, and $W$ be the  weighted
  \vset-automaton representing $\w \in \regk$. By
  Proposition~\ref{prop:epsilonRemoval}, we can assume, \mbox{w.l.o.g.}, that
  all \vset-automata used in this proof do not contain
  $\varepsilon$-transitions.

  We begin by giving the construction of $D$. Let $W_A$ be the 
  weighted \vset-automaton obtained by interpreting $A$ as a $\K$-weighted
  \vset-automaton. More formally, every transition in $A$ is interpreted as a
  weighted transition with weight $\srone$ and every transition which is not in
  $A$ is interpreted as a transition with weight $\srzero$. Furthermore, let
  $W_\doc \eqdef \doc$ be the  weighted \vset-automaton with
  $\vars(W_\doc) = \emptyset$ that assigns the weight $\srone$ to the empty
  tuple on input $\doc$ and $\srzero$ to every tuple on input $\doc' \neq \doc$.
  By Theorem~\ref{theo:closed-algebra} the join of  weighted \vset-automata
  can be computed in polynomial time. Let
  \[
    W_D \eqdef W \join W_A \join W_\doc\;\;.
  \]

  Per definition of join for $\K$-relations, it holds that 
  \[
    \repspnrw{W_D}(\doc,\tup) = \repspnrw{W}(\doc,\projectTup{\tup}{\vars(W)})
    \srtimes \repspnrw{W_A}(\doc,\projectTup{\tup}{\vars(W_A)}) \srtimes
    \repspnrw{W_\doc}(\doc,\projectTup{\tup}{\vars(W_\doc)})\;\;.
  \]

  Let $A \in \ufvsa$ be unambiguous or $\K = \tropical$. In both cases, it holds
  that 
  \[
    \repspnrw{W_A}(\doc,\tup) =
    \begin{cases}
      \srone & \text{if } \tup \in \toSpanner{A}(\doc) \text{, and} \\
      \srzero & \text{otherwise.}
    \end{cases}
  \]
  Furthermore,
  \[
    \repspnrw{W_\doc}(\doc',\tup) =
    \begin{cases}
      \srone & \text{if } \vars(\tup) = \emptyset \text{ and } \doc' = \doc
      \text{, and} \\
      \srzero & \text{otherwise.}
    \end{cases}
  \]

  Therefore, if $A \in \ufvsa$ or $\K = \tropical$, it holds, for every tuple
  $\tup \in \toSpanner{A}(\doc)$. that
  \[
    \repspnrw{W_D}(\doc,\tup) =
    \repspnrw{W}(\doc,\projectTup{\tup}{\vars(W)}) \quad(\dagger)
  \]
  We will use this equality in the proof of condition \ref{dag:weight}.

  The DAG $D = (N_D,E_D)$ is obtained from $W_D =
  (\alphabet,V,Q,I,F,\delta)$ as follows. The set of nodes $N_D \eqdef \big(Q
  \times (\alphabet \cup \varop{V} \cup \emptyset)\big) \uplus \{\dags, \dagt\}$
  contains the nodes $\dags, \dagt$, plus a state $(q,\sigma)$ for each $q \in
  Q$ and $\sigma \in (\alphabet\cup\varop{V} \cup \emptyset)$, where $\sigma
  \neq \emptyset$ encodes the label of the last transition and $q$ the state.
  The set of edges is defined as follows:
  \begin{align*}
    E_D \eqdef & \{(\dags, \ell, (x,\emptyset))
               \mid I(x) = \ell \neq \infty\} \\
             & \uplus \{((x_1,\sigma_1),\ell,(x_2,\sigma_2)
               \mid \delta(x_1,\sigma_2,x_2) = \ell \neq \srzero
               \text{, where } \sigma_1 \in (\alphabet\cup\varop{V}\cup\emptyset)\} \\
             & \uplus \{((x,\sigma), \ell, \dagt)
               \mid F(x) = \ell \neq \infty
               \text{, where } \sigma \in (\alphabet\cup\varop{V}\cup\emptyset)\}\;.
  \end{align*}

  In the following we assume that $D$ is trimmed, that is, for every node $n \in
  N_D$ there is at least one path from $\dags$ to $\dagt$, which visits
  $n$.\footnote{Note that this condition can be enforced in linear time by two
    graph traversals (e.g. using breadth first search), one starting from
    $\dags$ to identify all states which can be reached from $\dags$ and one
    starting from $\dagt$ to identify all states which can reach $\dagt$. We
    remove all states which are not marked by both graph traversals.}

  We observe that the construction of $D$ only requires polynomial time. Note
  that there is a one-to-one correspondence between paths $p \in \dagpaths$ and
  accepting runs of $W_D$ on $\doc$. That is,
  \[
    p = \dags \cdot \ell_0\cdot (q_0,\emptyset) \cdot \ell_1 \cdot
    (q_1,\sigma_1) \cdots (q_n,\sigma_n) \cdot \ell_{n+1} \cdot \dagt 
  \]
  is a path from $\dags$ to $\dagt$ in $D$ if and only if
  \[
    \rn = q_0 \overset{\sigma_1}{\rightarrow} q_1
    \overset{\sigma_2}{\rightarrow} \cdots
    \overset{\sigma_{n}}{\rightarrow} q_n\;,
  \]
  with $I(q_0) = \ell_0$ and $F(q_n) = \ell_{n+1}$ is an accepting run of $W_D$
  on $\doc$. Furthermore, we observe that the weight of $p$ is exactly the
  weight assigned to the run $\rn$ by $W_D$, that is, $\len(p) = \weight\rn$.

  For the sake of contradiction, assume that $D$ is cyclic. Per assumption, all
  nodes $n \in N$ are on a path from $\dags$ to $\dagt$, thus, $D$ must have a
  path $p$ from $\dags$ to $\dagt$, which contains a cycle. Let $\rn$ be the run
  of $W_D$ corresponding to $p$. The automaton $W_\doc$ is acyclic. Observe that
  $W_D$ is functional as $W$, $W_A,$ and $W_\doc$ are functional. Thus,
  $\fromRun{\rn}$ is valid and therefore the cycle can not contain an edge
  labeled by a variable operation. Per assumption, all involved \vset-automata
  do not contain $\varepsilon$-transitions. Therefore, the cycle must only
  consist of edges, labeled by alphabet symbols. Let $\rn'$ be the run, obtained
  from $\rn$ by removing all cycles. Due to commutativity of $\srtimes,$ it
  follows that $\weight{\rn'} = \weight{\rn} \srtimes x$ for some $x \neq
  \srzero$. We observe that $\clr(\fromRun{\rn'}) \neq \doc$. Therefore, there
  is a run $\rn'$ of $W_D$ on $\clr(\fromRun{\rn'}) \neq \doc$ with weight
  $\weight{\rn'} \neq \srzero$, which is the desired contradiction to the
  observation that for all runs $\rn$ of $W_D$ it holds that $\weight{\rn} \neq
  \srzero$ if and only if $\clr(\fromRun{\rn}) = \doc$.

  We now define the mapping $m$. Let $p \in \dagpaths$ and let $\rn$ be the
  corresponding run of $W_D$. We define the mapping $m(p) \eqdef \toTuple{\rn}$.
  It follows directly that $m$ is surjective. If $A \in \ufvsa$ or $\K =
  \tropical$ and for $\tup \in \toSpanner{A}(\doc)$, we have that
  \begin{align*}
    \w(\doc,\tup)
    &=  \;\;\repspnrk{W}{\srd}(\doc,\projectTup{\tup}{\vars(W)}) 
    \eqnumber{\dagger} \;\;\repspnrk{W_D}{\srd}(\doc,\tup) \\
    & =  \boplus_{\rn\in \Rn{W_D}{\doc} \text{ and }\tup = \toTuple{\rn}} \weight{\rn} 
    = \boplus_{p\in \dagpaths, m(p) = \tup} \len(p)\;.
  \end{align*}
  The first and the third equalities follow from the definitions of $\REG$
  weight functions and $\srd$-annotators. The last equality follows from the
  definition of $D$. This concludes the proof of condition \ref{dag:weight}.
  
  It remains to show that condition \ref{dag:bijection} holds. Assume that $A
  \in \ufvsa$ and $W$ are unambiguous. Then, by
  Theorem~\ref{theo:closed-algebra}, $W_D$ is unambiguous.\footnote{Recall that
    $W_\doc$ is unambiguous.} Assume that there are two paths $p_1 \neq p_2$
  such that $p_1,p_2 \in \dagpaths$ with $m(p_1) = m(p_2)$. Let $\rn_1\neq
  \rn_2$ be the corresponding runs of $W_D$. Due to $m(p) = \toTuple{\rn}$, it
  must hold that $\rn_1$ and $\rn_2$ are two runs of $W_D$, encoding the same
  tuple $\tup$. Due to the unambiguity condition \ref{cond:unambig} in Section~\ref{sec:vset-automata}, both runs
  must encode a different ref-word, that is, $\fromRun{\rn_1} \neq
  \fromRun{\rn_2}$ however this implies that either $\fromRun{\rn_1}$ or
  $\fromRun{\rn_2}$ must violate the variable order condition, contradicting the
  unambiguity condition~\ref{cond:voc} in Section~\ref{sec:vset-automata}. Thus, $m$ must be a bijection.
\end{proof}

\subsection{\minp and \maxp Aggregation}
We will now study the computational complexity of \minp and \maxp aggregation.
We begin by giving the tractable cases which are based on
Lemma~\ref{lem:dagComp}. The weighted DAG from Lemma~\ref{lem:dagComp} allows us
to reduce \minp to the shortest path problem in DAGs. If the weight function is
unambiguous, \maxp can be reduced to the longest path problem in DAGs. Notice
that, although the longest path problem is intractable in general, it is
tractable for DAGs.
\begin{theorem}\label{thm:minmaxRegTractable}The problems $\minp[\fvsa,\regtrop]$, $\minp[\ufvsa,\uregnum]$, $\maxp[\fvsa, \uregtrop]$,
  and $\maxp[\ufvsa,\uregnum]$ are in \fp.
\end{theorem}
\begin{proof}
  Let $\doc$ be a document, $A\in\fvsa$, and $W$ be the  weighted
  \vset-automaton representing $\w \in \regtrop$ or $\w \in \uregnum$. Let $D$
  and $m$ be the DAG and the surjective mapping as guaranteed by
  Lemma~\ref{lem:dagComp}. In the following, we will reduce all four cases to
  finding the path with minimal (resp., maximal) length in $D$. Note that given
  a weighted DAG $D$, one can compute the path with minimal (resp., maximal)
  length in polynomial time, via dynamic programming, e.g. using the
  Bellman-Ford algorithm.\footnote{One has to be careful in the case of the
    numeric semiring as the lengths along the path are multiplied. Therefore one
    has to maintain the minimal as well as the maximal length between two nodes,
    as edges with negative length change the sign, resulting in minimal path's
    to be maximal and vice versa.} 

  We begin by giving the proofs for the numerical semiring. If $A \in \ufvsa$
  and $W \in \uregnum$, it follows directly from property~\ref{dag:bijection} of
  Lemma~\ref{lem:dagComp} that $m$ is a bijection. Therefore, for every tuple
  $\tup \in \toSpanner{A}(\doc)$, there is exactly one path $p \in \dagpaths$
  with $m(p)=\tup$. Thus, $\w(\doc,\tup) = \len(p)$, where $p \in \dagpaths$
  with $m(p) = \tup$. It follows directly that $\spAminagg$ and $\spAmaxagg$ can
  be computed from $D$ by searching for the path $p$ with minimal (respectively
  maximal) length.

  It remains to give the proofs for the tropical semiring. We begin by giving
  the proof for $\minp[\fvsa,\regtrop]$. Due to property~\ref{dag:weight} of
  Lemma~\ref{lem:dagComp},
  \[
    \spAminagg = \min\limits_{\tup \in
      \toSpanner{A}(d)}\;\min\limits_{{p\in\dagpaths, m(p)=\tup}} \len(p) =
    \min\limits_{p \in \dagpaths} \len(p)
  \]
  and therefore $\minp[\fvsa,\regtrop]$ again reduces to computing the path of
  minimal length in $D$.
  
  For \maxp, the situation is different, because the maximal weight of an output
  tuple is
  \[
    \spAmaxagg = \max\limits_{\tup \in
      \toSpanner{A}(d)}\;\min\limits_{{p\in\dagpaths, m(p)=\tup}} \len(p)\;.
  \]
  However, if $W$ is unambiguous, it must hold that $\len(p) = \len(p^\prime)$
  for all runs $p, p^\prime \in \dagpaths$ with $m(p) = m(p^\prime)$. Otherwise
  $W$ would be required to have at least two runs which accept the same tuple
  but assign different weights. Thus, $W$ would not be unambiguous. We can
  therefore conclude that,
  \[\spmaxagg = \max\limits_{\tup \in
      \toSpanner{A}(d)}\min\limits_{\{p \mid m(p) = \tup\}} \len(p) =
    \max\limits_{p \in \dagpaths} \len(p)\;.
  \]
  Again, we can reduce $\maxp[\fvsa, \uregtrop]$ to the max length problem on
  $D$.
\end{proof}

As we show now, the results of Theorem~\ref{thm:minmaxRegTractable} are close to
the tractability frontier: For instance, if we relax the unambiguity condition
in the weight function, the problem \maxp does not correspond to finding the
longest paths in DAGs and becomes intractable.

\begin{theorem}\label{thm:minmaxRegintractable}$\minp[\ufvsa,\regnum]$, $\maxp[\ufvsa,\regtrop]$, and $\maxp[\ufvsa,\regnum]$
  are \optp-hard.
\end{theorem}
\begin{proof}
  We begin by giving the proofs for $\maxp[\ufvsa,\regtrop]$. We give a metric
  reduction\footnote{Recall that a metric reduction from $f$ to $g$ is a pair of
    polynomial-time computable functions $T_1, T_2$, where $T_1: \alphabet^* \to
    \alphabet^*$ and $T_2: \alphabet^* \times \nat \to \nat$, such that $f(x) =
    T_2(x,g(T_1(x)))$ for all $x \in \alphabet^*$. } from the \optp-complete
  problem Maximum Satisfying Assignment (MSA) \cite{Krentel88}, which is defined
  as follows. Let $\phi(x_1,\ldots,x_n)$ be a propositional formula in CNF and
  let $v = v_1\cdots v_n \in \B^n$ be a variable assignment of $\phi$.
  Furthermore, let $n_v \in \nat$ be the natural number encoded by $v$ in
  binary. MSA asks, given the CNF formula $\phi(x_1,\ldots,x_n)$, for the
  maximum $n_v\in \nat$ such that $v$ satisfies $\phi$, or $0$ if $\phi$ is not
  satisfiable. In the following, we denote by $\text{MSA}(\phi)$ the output of
  MSA on input $\phi$.
  
  Let $\phi(x_1,\ldots,x_n)$ be a Boolean formula in CNF. We use a similar
  construction as in the proofs of Theorem~\ref{thm:sw-sum-cnf} and Doleschal et
  al.~\cite[Theorem 7.6]{DoleschalKMP-lmcs22}, to encode the CNF
  formula $\phi$. Let $\doc = a^n$ be the document. We define
  \[
    A \eqdef ((\vop{x_1} \varepsilon\vcl{x_1} a) \lor
    (\vop{x_1}a\vcl{x_1})) \cdots ((\vop{x_n}\varepsilon\vcl{x_n} a) \lor
    (\vop{x_n}a\vcl{x_n}))\;.
  \]
  Notice that $A$ can be defined with a polynomial-time constructible $\ufvsa$.
  Observe that there is a one-to-one correspondence between tuples $\tup$ in
  $\toSpanner{A}(\doc)$ and variable assignments $\alpha_\tup$ for $\phi$: we
  can set $\alpha_{\tup}(x_i) = 1$ if and only if $\tup(x_i) =
  \spanFromTo{i}{i+1}$. We construct a weight function $\w \in \regtrop$ such
  that
  \[
    \w(\doc,\tup) =
    \begin{cases}
      n_{\alpha_\tup} & \text{if } \alpha_\tup \models \phi\\
      0 & \text{otherwise.}
    \end{cases}
  \]
  Recall that $n_{\alpha_\tup}$ is the natural number which is encoded by the
  variable assignment $\alpha_\tup$. It follows directly that $\text{MSA}(\phi)
  = \spAmaxagg$. Defining $T_2(x,y) \mapsto y$ gives the desired reduction.
 
  It remains to construct a weighted \vset-automaton $W$ which encodes $\w$. We
  define the  weighted \vset-automaton $W$ as the union of two
  automata. Let $V$ be the set of variables of $\phi$. The first automaton $W_A$
  is a copy of $A$, assigning weight $0$ to all edges, which are present in $A$.
  Furthermore, let $\delta$ assign weight $2^{i-1}$ to the $a$ labeled edge between
  opening and closing variable $x_i$ (that is, $\vop{x_i}$ and $\vcl{x_i}$). Let
  $I(q) = 0$ if $q$ is the start state of $A$ and $\infty$, otherwise.
  Analogously, let $F(q) = 0$ if $q$ is an accepting state of $A$ and $\infty$
  otherwise. It follows directly that $\repspnrw{W_A}(a^n,\tup) =
  n_{\alpha_\tup}$.

  The second automaton, $W'$ consists of $m$ disjoint branches, where each
  branch corresponds to a clause $C_i$ of $\phi$; we call these \emph{clause
    branches}. Each branch has exactly one run $\rn$ with weight $\srone$ for
  each tuple $\tup$ associated to an assignment $\alpha_\tup$ which does not
  satisfy the clause $C_i$.
  
  We now give a formal construction of $W'$. The set of states $Q \eqdef \{
  q_{i,j}^{a} \mid 1 \leq i \leq m, 1 \leq j \leq n, 1 \leq a \leq 5\}$ contains
  $5n$ states for each clause branch. Intuitively, $W'$ has a gadget, consisting
  of $5$ states, for each variable and each clause branch.
  Figure~\ref{fig:exGadgetAggregates} depicts the three types of gadgets we use
  here. Note that the weights of the drawn edges are all $0$. We use the left
  gadget if $x$ does not occur in the relevant clause and the middle (resp.,
  right) gadget if the literal $\lnot x$ (resp., $x$) occurs. Furthermore,
  within the same branch of $W'$, the last state of each gadget is the same
  state as the start state of the next variable, i.e., $q_{i,j}^{5} =
  q_{i,j+1}^{1}$ for all $1 \leq i \leq k, 1 \leq j < n$.
  
  \begin{figure}
    \resizebox{\linewidth}{!}{
    \begin{tikzpicture}[>=latex, ->]
      
      \tikzstyle{every node}=[font=\Large]
      
      \def\w{2}
      \def\h{.75}
      \def\wtwo{2}
      
      \node[state] (q1) at (0,0){};
      \node[state] (q2) at ($(q1) + (\w,0)$){};
      \node[state] (q3) at ($(q2) + (\w,\h)$){};
      \node[state] (q4) at ($(q2) + (\w,-\h)$){};
      \node[state] (q5) at ($(q4) + (\w,\h)$){};
      
      \node[state] (q11) at ($(q5) + (\wtwo,0)$){};
      \node[state] (q12) at ($(q11) + (\w,0)$){};
      \node[state] (q13) at ($(q12) + (\w,\h)$){};
      \node[state] (q14) at ($(q12) + (\w,-\h)$){};
      \node[state] (q15) at ($(q14) + (\w,\h)$){};
      
      \node[state] (q21) at ($(q15) + (\wtwo,0)$){};
      \node[state] (q22) at ($(q21) + (\w,0)$){};
      \node[state] (q23) at ($(q22) + (\w,\h)$){};
      \node[state] (q24) at ($(q22) + (\w,-\h)$){};
      \node[state] (q25) at ($(q24) + (\w,\h)$){};
      
      \draw (q1) edge node[above,sloped] {$\vop{x}$}(q2);
      \draw (q2) edge node[above,sloped] {$\vcl{x}$}(q3);
      \draw (q2) edge node[below,sloped] {$\sigma$}(q4);
      \draw (q3) edge node[above,sloped] {$\sigma$}(q5);
      \draw (q4) edge node[below,sloped] {$\vcl{x}$}(q5);
      
      \draw (q11) edge node[above,sloped] {$\vop{x}$}(q12);
      \draw (q12) edge node[above,sloped] {$\vcl{x}$}(q13);
      \draw (q13) edge node[above,sloped] {$\sigma$}(q15);
      \draw (q14) edge node[below,sloped] {$\vcl{x}$}(q15);
      
      \draw (q21) edge node[above,sloped] {$\vop{x}$}(q22);
      \draw (q22) edge node[above,sloped] {$\vcl{x}$}(q23);
      \draw (q22) edge node[below,sloped] {$\sigma$}(q24);
      \draw (q24) edge node[below,sloped] {$\vcl{x}$}(q25);
    \end{tikzpicture}
    }
    \caption{Example gadgets for variable $x$.}\label{fig:exGadgetAggregates}
  \end{figure}
  \begin{figure}
    \resizebox{\linewidth}{!}{
    \begin{tikzpicture}[>=latex]
      
      \tikzstyle{every node}=[font=\Large]
      
      \def\w{2}
      \def\h{.75}
      \def\htwo{1.5}
      \node[state] (q1111) at (0,0) {$q_{1,1}^{1}$};
      \node[state] (q1112) at ($(q1111) + (\w,0)$) {$q_{1,1}^{2}$};
      \node[state] (q1113) at ($(q1112) + (\w,\h)$) {$q_{1,1}^{3}$};
      \node[state] (q1114) at ($(q1112) + (\w,-\h)$) {$q_{1,1}^{4}$};
      \node[state] (q1115) at ($(q1114) + (\w,\h)$) {$q_{1,1}^{5}$};
      \draw (q1111) edge[->] node[above,sloped] {$\vop{x_1}$}(q1112);
      \draw (q1112) edge[->] node[above,sloped] {$\vcl{x_1}$}(q1113);
      \draw (q1112) edge[->] node[below,sloped] {$\sigma$}(q1114);
      \draw (q1114) edge[->] node[below,sloped] {$\vcl{x_1}$}(q1115);
      \draw [decorate,decoration={brace,amplitude=10pt}]
        ($(q1111) + (0,\htwo)$) -- node[above=10pt] {$x_1$} ($(q1115) + (0,\htwo)$);
      
      \node[state] (q1212) at ($(q1115) + (\w,0)$) {$q_{1,2}^{2}$};
      \node[state] (q1213) at ($(q1212) + (\w,\h)$) {$q_{1,2}^{3}$};
      \node[state] (q1214) at ($(q1212) + (\w,-\h)$) {$q_{1,2}^{4}$};
      \node[state] (q1215) at ($(q1214) + (\w,\h)$) {$q_{1,2}^{5}$};
      \draw (q1115) edge[->] node[above,sloped] {$\vop{x_2}$}(q1212);
      \draw (q1212) edge[->] node[above,sloped] {$\vcl{x_2}$}(q1213);
      \draw (q1212) edge[->] node[below,sloped] {$\sigma$}(q1214);
      \draw (q1214) edge[->] node[below,sloped] {$\vcl{x_2}$}(q1215);
      \draw [decorate,decoration={brace,amplitude=10pt}]
        ($(q1215) + (0,-\htwo)$) -- node[below=10pt] {$x_2$} ($(q1115) + (0,-\htwo)$);
      
      \node[state] (q1312) at ($(q1215) + (\w,0)$) {$q_{1,3}^{2}$};
      \node[state] (q1313) at ($(q1312) + (\w,\h)$) {$q_{1,3}^{3}$};
      \node[state] (q1314) at ($(q1312) + (\w,-\h)$) {$q_{1,3}^{4}$};
      \node[state] (q1315) at ($(q1314) + (\w,\h)$) {$q_{1,3}^{5}$};
      \draw (q1215) edge[->] node[above,sloped] {$\vop{x_3}$}(q1312);
      \draw (q1312) edge[->] node[above,sloped] {$\vcl{x_3}$}(q1313);
      \draw (q1312) edge[->] node[below,sloped] {$\sigma$}(q1314);
      \draw (q1313) edge[->] node[above,sloped] {$\sigma$}(q1315);
      \draw (q1314) edge[->] node[below,sloped] {$\vcl{x_3}$}(q1315);
      \draw [decorate,decoration={brace,amplitude=10pt}]
        ($(q1215) + (0,+\htwo)$) -- node[above=10pt] {$x_3 \lor \lnot x_3$} ($(q1315) + (0,+\htwo)$);
      
      \node[state] (q1412) at ($(q1315) + (\w,0)$) {$q_{1,4}^{2}$};
      \node[state] (q1413) at ($(q1412) + (\w,\h)$) {$q_{1,4}^{3}$};
      \node[state] (q1414) at ($(q1412) + (\w,-\h)$){$q_{1,4}^{4}$};
      \node[state] (q1415) at ($(q1414) + (\w,\h)$){$q_{1,4}^{5}$};
      \draw (q1315) edge[->] node[above,sloped] {$\vop{x_4}$}(q1412);
      \draw (q1412) edge[->] node[above,sloped] {$\vcl{x_4}$}(q1413);
      \draw (q1413) edge[->] node[above,sloped] {$\sigma$}(q1415);
      \draw (q1414) edge[->] node[below,sloped] {$\vcl{x_4}$}(q1415);
      \draw [decorate,decoration={brace,amplitude=10pt}]
        ($(q1415) + (0,-\htwo)$) -- node[below=10pt] {$\lnot x_4$} ($(q1315) + (0,-\htwo)$);
      \end{tikzpicture}
    }
    \caption{The clause branch of $W$ corresponding to $C_1$ and $x_1 = x_2 =
      1, x_4 = 0$.}\label{fig:exBranchAggregates}
  \end{figure}
  
  We illustrate the crucial part of the construction on an example. Let $\phi =
  (\lnot x_1 \lor \lnot x_2 \lor x_4) \land (x_2 \lor x_3 \lor x_4).$ The
  corresponding  weighted \vset-automaton $W'$ therefore has two
  disjoint branches, one for each clause of $\phi$.
  Figure~\ref{fig:exBranchAggregates} depicts the clause branch $C_1$ that
  corresponds to all assignments which do not satisfy $C_i$, that is, all
  assignments with $x_1 = x_2 = 1$ and $x_4 = 0$.

  Formally, the initial weight function is $I(q_{i,j}^{a}) = \srone$ if $j = 1
  = a$ and $I(q_{i,j}^{a}) = \srzero$ otherwise. The final weight function
  $F(q_{i,j}^{a}) = \srone$ if $j = n$ and $a = 5$ and $F(q_{i,j}^{a}) =
  \srzero$, otherwise. The transition function $\delta$ is defined as follows:
  \[
    \delta(q_{i,j}^{a}, o, q_{i,j}^{a'}) =
    \begin{cases}
      \srone  & a = 1, a' = 2, o = \vop{x_j} \\
      \srone  & a = 2, a' = 3, o = \vcl{x_j} \\
      \srone  & a = 2, a' = 4, \text{ and there is a variable assignment
        $\tau$ with} \\
        &\;\;\text{$\tau(x_j) = 1$ and $\tau \not\models C_i$} \\ 
      \srone  & a = 3, a' = 5, o = a, \text{ and there is a variable assignment
        $\tau$ with} \\
        &\;\;\text{$\tau(x_j) = 0$ and $\tau \not\models C_i$} \\
      \srone  & a = 4, a' = 5, o = \vcl{x_j} 
    \end{cases}
  \]
  All other transitions have weight $\srzero$.

  We claim that $W'$ represents $\w'$, where $\w'(\doc,\tup) = \srone$ if
  $\alpha_\tup \not\models \phi$ and $\w'(\doc,\tup) = \srzero$ otherwise. To
  this end, let $\tup \in \toSpanner{A}(\doc)$ be a tuple and let $\tau =
  \alpha_\tup$ be the variable assignment encoded by $\tup$. It is easy to see
  that there is an accepting run $\rn$ of $W'$ for $\refWord$ with weight
  $\weight{\rn} = \srone$, starting in $q_{i,0}^{a}$, if and only if $\tau$ does not
  satisfy clause $C_i$.

  As mentioned before, the  weighted \vset-automaton $W$ is the union
  of $W'$ and $W_A$. Recall that, over the tropical semiring, $\srzero = \infty,
  \srone = 0$, and the weight of a tuple $\tup$ is the minimal weight over all
  accepting runs which encode $\tup$. Thus, the weight function represented by
  $W$ is exactly $\w$, as claimed. This concludes the proof that
  $\maxp[\ufvsa,\regtrop]$ is \optp-hard.

  It remains to show that $\minp[\ufvsa,\regnum]$ and $\maxp[\ufvsa, \regnum]$
  are \optp-hard. We first show \optp-hardness for $\maxp[\ufvsa,\regnum]$.

  We give a metric reduction from the \optp-complete problem of weighted
  satisfiability (WSAT) \cite{Krentel88}, which is defined as follows. Let
  $\phi(x_1,\ldots,x_n)$ be a propositional formula in CNF with binary weights.
  WSAT asks, given the CNF formula $\phi(x_1,\ldots,x_n)$ with $m$ clauses and
  weights $w_1,\ldots,w_m$, for the maximal weight of an assignment, where the
  weight of an assignment is the sum of the weights of the satisfied clauses.

  Denote by $\text{WSAT}(\phi)$ the output of WSAT on input $\phi$. Let
  $\phi(x_1,\ldots,x_n)$ be a Boolean formula in CNF. Let $\doc, A, W$ be as
  defined before. However, the weights in $W$ are defined differently. That is,
  $W$ is the union of $W_A$ and $W'$, where $W_A$ is a copy of $A$, where all
  transitions have weight $\srone$. Furthermore, let $x$ be the sum of all
  clause weights and $F(q) = x$, if $q$ is an accepting state of $A$. The
  automaton $W'$ is defined exactly as before, however, accepting with final
  weight $F(q) = -w_i$ if $q$ is the final weight of the branch of clause $C_i$
  and $w_i$ is the weight of $C_i$. Observe that $\w(\doc,\tup) =
  \repspnrk{W}{\numerical}(\doc,\tup)$ is exactly the weighted sum of all
  clauses, which are satisfied by the valuation $\alpha_\tup$ encoded by $\tup$.
  It follows that $\spmax(\spanner, \doc, \w) = \text{WSAT}(\phi)$. Defining
  $T_2(x,y) \mapsto y$ concludes the proof for $\maxp[\ufvsa,\regnum]$.

  The proof for $\minp[\ufvsa,\regnum]$ is analogous, replacing the weight $x$
  with $-x$ and $-w_i$ with weight $w_i$. Therefore, $\spmin(\spanner, \doc, \w)
  = - \text{WSAT}(\phi)$. Defining $T_2(x,y) \mapsto -y$ concludes the proof.
\end{proof}

\subsection{\sump and \avgp Aggregation}
Since \sump and \avgp are already intractable for \fvsa spanners and \SW weight
functions (Theorems~\ref{thm:sw-sum-cnf},~\ref{thm:sw-sum-spanl},
and~\ref{thm:sw-quant-avg-hard}), they are intractable for \fvsa spanners and
\REG/\UREG weight functions as well. In a similar vein as in
Section~\ref{cwidth}, the problems become tractable if we have unambiguity.
However, in the case of the tropical semiring, we require unambiguity of
\emph{both} the spanner and the representation of the weight function. We begin
by showing that $\sump[\ufvsa,\regnum]$ and $\sump[\ufvsa,\uregtrop]$ are in
\fp. In both cases the problem can be reduced to computing the sum of the lengths of source-to-target paths in a DAG, by using Lemma~\ref{lem:dagComp}.
\begin{theorem}\label{thm:sumREGnumerical}$\sump[\ufvsa,\regnum]$ is in \fp.
\end{theorem}
\begin{proof}
  Let $\doc\in\docs, A \in \ufvsa$, and $W$ be a  weighted
  \vset-automaton representing $\w \in \regnum$. Let $D = (N,E)$ and
  $m$ be as guaranteed by Lemma~\ref{lem:dagComp}. It follows that
  \[
    \spAsumagg = \sum\limits_{\tup \in
      \toSpanner{A}(\doc)}\;\sum\limits_{p\in\dagpaths, m(p)=\tup} \len(p) =
    \sum\limits_{p \in \dagpaths} \len(p)\;. 
  \]
  All paths $p \in \dagpaths$ consist of $|\doc|+1+2\cdot |\vars(A)|$ edges. We
  assume, \mbox{w.l.o.g.}, that $N = \{1,\ldots, n\}$, with $\dags = 1$ and
  $\dagt = n$. Therefore, $\spAsumagg$ can be computed by interpreting the edge
  relation $E$ as a $\rationals^{n \times n}$ matrix $M$ and computing the
  weight
  \[
    I \times M^{|\doc|+1+2\cdot |\vars(A)|} \times F^T\;,
  \]
  where $I = (1,0,\ldots,0)$ (resp., $F = (0,\ldots, 0,1)$) is the vector which
  assigns $0$ to all nodes but $1$ (resp., $n$), which is assigned the weight
  $1$. Recall that the numerical semiring has an efficient encoding. Therefore,
  $\spsumagg$ can indeed be computed in polynomial time.
\end{proof}

\begin{theorem}\label{thm:sumUregFP}$\sump[\ufvsa,\uregtrop]$ is in \fp.
\end{theorem}
\begin{proof}
  Let $D,m$ be the DAG and the bijection guaranteed by
  Lemma~\ref{lem:dagComp}. We have that
  \begin{align*}
    \spAsumagg
    & \eqnumber 1\; \sum\limits_{\tup \in \toSpanner{A}(\doc)} w(\doc,\tup) \\
    & \eqnumber 2\; \sum\limits_{\tup \in \toSpanner{A}(\doc)} \min\limits_{p\in\dagpaths,m(p)=\tup}\len(p) \\
    & \eqnumber 3\; \sum\limits_{p \in \dagpaths} \len(p)\;.
  \end{align*}
  The first equation follows from the definition of $\sump$. The second
  equation follows from property~\ref{dag:weight} of Lemma~\ref{lem:dagComp}. The
  third equation must hold due to $m$ being a bijection between tuples $\tup \in
  \toSpanner{A}$ and paths $p \in \dagpaths$.

  It remains to show that the sum of the lengths of source-to-target paths in a
  DAG $D = (N,E)$ can be computed in polynomial time. We begin by observing that
  given two nodes $x,y \in D$ the number of paths from $x$ to $y$ in $D$ can be
  computed in polynomial time via dynamic programming. Furthermore, given an
  edge $e = (x,y) \in E$ one can compute the number of paths from $\dags$ to
  $\dagt$ which use $e$ by multiplying the number of path's from $\dags$ to $x$
  with the number of paths from $y$ to $\dagt$. Therefore, the function $c:E \to
  \nat$ which, given an edge $e \in E$ assigns the number of paths using $e$
  can be computed in polynomial time. Recall that over the tropical semiring,
  $\srtimes = +$ and therefore $\len(p) = \sum\limits_{e \in p}\len(e)$. It
  therefore follows that
  \begin{align*}
    \spAsumagg
    & = \sum\limits_{p \in \dagpaths} \len(p) \\
    & = \sum\limits_{p \in \dagpaths} \sum\limits_{e \in p} \len(e) \\
    & = \sum\limits_{e \in E} (\len(e) \times c(e))\;.
  \end{align*}
  Therefore, \sump can be computed by representing the weights $\len(e)$ as
  a vector $I$ and the counts $c(e)$ as a vector $F$. Thus, $\spAsumagg = I
  \times F^T$, which can be computed in polynomial time, as $\regtrop$ has an
  efficient encoding.
\end{proof}

We observe that \fp upper bounds for \avgp follows directly from the
corresponding upper bound for \sump and the \fp upper bound for \countp
(Theorem~\ref{thm:countResults}).

\begin{corollary}\label{cor:avgregfp}$\avgp[\ufvsa,\regnum]$ and $\avgp[\ufvsa,\uregtrop]$ are in \fp.
\end{corollary}

If we relax the restriction that weight functions are given as unambiguous
automata, \sump and \avgp become \sharpp-hard again.

\begin{theorem}\label{thm:sumTropicalHard}$\sump[\ufvsa,\regtrop]$ and $\avgp[\ufvsa,\regtrop]$ are
  \sharpp-hard. 
\end{theorem}
\begin{proof}
  We begin by giving a parsimonious reduction from the \sharpp-complete problem of
  $\#$CNF. To this end, let $c = 1$ in the case of $\sump$ and $c = 2^n$ in the
  case of \avgp.

  Let $\phi(x_1,\ldots,x_n)$ be a propositional formula in conjunctive normal
  form. Let $A,\doc$ be as constructed in the proof of
  Theorem~\ref{thm:minmaxRegintractable} and let $\w$ be the weight function
  such that $\w(\doc,\tup) = c$ if the corresponding assignment $\alpha_\tup$
  satisfies $\phi$ and $\w(\doc,\tup) = 0$ otherwise. Therefore, with $c \eqdef 1$
  it follows directly that $\#\text{CNF}(\phi) = \spAsumagg$, which shows that
  the problem is $\sharpp$-hard. For \avgp let $c \eqdef 2^n$. It follows that
  $\#\text{CNF}(\phi) = x = \frac{x \cdot 2^n}{2^n} = \frac{x \cdot c}{2^n} =
  \spAavgagg$, implying that $\avgp[\ufvsa,\regtrop]$ is also \sharpp-hard.
  
  It remains to show that there is a weighed automaton $W$ representing $\w \in
  \regtrop$. As in the proof of Theorem~\ref{thm:minmaxRegintractable}, $W$ is
  the union of two weighted \vset-automata $W_A$ and $W'$, where $W_A$ is a copy
  of $A$, assigning weight $0$ to all initial states and transitions
  of $A$ and weight $c$ to all final states. Furthermore, $W'$ is as defined,
  that is
  \[
    \repspnrk{W'}{\tropical}(a^n,\tup) =
    \begin{cases}
      0 & \text{if } \alpha_\tup \not\models \phi\\
      \infty & \text{otherwise.}
    \end{cases}
  \]
  It follows directly that $W$ encodes the weight function $\w$, concluding the
  proof.
\end{proof}

Finally, we show that \sump and \avgp for \regtrop weight functions are in
$\fptosp$.
\begin{theorem}\label{thm:regUpperBounds}$\sump[\fvsa,\REG]$ and $\avgp[\fvsa,\REG]$ are in
  $\fptosp$.
\end{theorem}
\begin{proof}
  We will begin by showing that $\sump[\fvsa,\REG]$ is in \fptosp if all weights
  assigned by $\w$ are natural numbers. We will use this as an oracle for the
  general upper bound. Let $A$ be a \vset-automaton, $\doc \in \docs$ be a
  document and $\w \in \REG$ be a weight function, which only assigns natural
  numbers and is represented by a weighted \vset-automaton $W$. A
  counting Turing Machine $M$ for solving the problem in $\sharpp$ would have
  $\w(\doc,\tup)$ accepting runs for every tuple in $A(\doc)$. More precisely,
  $M$ guesses a \doc-tuple $\tup$ over $\vars(A)$ and can check whether $\tup
  \in \toSpanner{A}(\doc)$ and $\w(\doc,\tup) > 0$. If so, $M$ branches into
  $\w(\doc,\tup)$ accepting branches. Otherwise, $M$ rejects. Per construction,
  $M$ has exactly $\w(\doc,\tup)$ accepting branches for every tuple $\tup \in
  \toSpanner{A}(\doc)$ with $\w(\doc,\tup) > 0$. Thus, the number of accepting
  runs is exactly $\sum_{\tup \in \toSpanner{A}(\doc)} \w(\doc,\tup) =
  \spAsumagg$.
  
  Now, let $\w \in \REG$ be a weight function, represented by the 
  weighted \vset-automaton $W$. We can assume, \mbox{w.l.o.g.}, that all
  rationals in $W$ have the denominator $d_{\text{lcm}}$.\footnote{This can be
    achieved by computing the least common multiple of all denominators
    $d_{\text{lcm}}$ in $W$ and expanding all fractions $\frac{a}{b}$ by
    $\frac{b}{d_{\text{lcm}}}$.} We recall that $\w(\doc,\tup) =
  \repspnr{W}(\doc,\projectTup{\tup}{\vars(W)})$. Thus, $w(\doc,\tup)$ is the
  product of $|\doc|+1+2*|\vars(A)|$ rationals, where each factor has the
  denominator $d_{\text{lcm}}$. Therefore,
  $\repspnr{W}(\doc,\projectTup{\tup}{\vars(W)})$ must have the denominator
  $d_{\text{lcm}}^{|\doc|+1+2|\vars(A)|}$\footnote{For the tropical semiring the
    denominator is actually $d_{\text{lcm}}$, as the multiplicative operation is
    $+$, which does not increase the denominator if both summands have the same
    denominator.}, which has an encoding length linear in $W$ and $\doc$. Thus,
  $\sump[\fvsa,\REG]$ can be computed by two calls to a $\sump[\fvsa,\REG]$-oracle. The first call only considers positive numerators, whereas the second
  call only considers negative numerators. Then, $\sump[\fvsa,\REG]$ is the
  difference of the results of both oracle calls, divided by
  $d_{\text{lcm}}^{\;|\doc|+1+2*|\vars(A)|}$.
  
  The upper bound for $\avgp[\fvsa, \regtrop]$ is immediate from the upper
  bound of $\sump[\fvsa, \regtrop]$ and Theorem~\ref{thm:countResults}.
\end{proof}

\subsection{Quantile Aggregation}
The situation for \quantp is different from the other aggregation problems,
since it remains hard, even when both the spanner and weight function are
unambiguous. The reason is that the problem reduces to counting the number of
paths in a weighted DAG that are shorter than a given target weight, which is
\sharpp-complete due to Mihal{\'{a}}k et al.~\cite{MihalakSW16}.
\begin{theorem}\label{thm:quantUregHard}$\quantp[\ufvsa,\UREG]$ is \sharpp-hard under Turing reductions, for every
  $0 < q < 1$.
\end{theorem}

At the core of the quantile problem is the problem of counting up to a threshold
$k \neq \infty$:
\[
  \spcount_{<k}(\spanner, \doc, \w) \eqdef |\{\tup \in P(\doc) \mid
  \w(\doc,\tup) \leq k\}|.
\]
The problems $\spcount_{>k}(\spanner, \doc, \w)$ and $\spcount_{=k}(\spanner,
\doc, \w)$ are defined analogously. The decision problem
$\countp_{<k}[\spannerClass,\cW]$ is defined analogously to
$\sump[\spannerClass,\cW]$. We begin by showing that
$\countp_{<k}[\ufvsa,\uregnum]$ and $\countp_{<k}[\ufvsa,\uregtrop]$
\sharpp-hard under Turing reductions. To this end, we reduce from $\#$Partition
and $\#$-Product-Partition.

Given a set $N = \{n_1,\ldots,n_n\}$ of natural numbers. Two sets $N_1,N_2$ are
a partition of $N$ if $N_1 \cup N_2 = N$ and $N_1 \cap N_2 = \emptyset$.
Furthermore, a partition is perfect, if the sums of the natural numbers in both
sets are equal. Given such a set $N = \{n_1,\ldots,n_n\}$, the $\#$Partition
problem asks for the number of perfect partitions.

Analogously, a partition $N_1,N_2$ is called a perfect product partition, if the
products of the natural numbers in both sets are equal. Furthermore, the
Product-Partition Problem asks whether there is a perfect product partition and
the problem $\#$Product-Partition asks for the number of perfect product
partitions.
\begin{proposition}$\#$Partition and $\#$Product-Partition are \sharpp-complete under Turing
  reductions. 
\end{proposition}
\begin{proof}
  Mihal{\'{a}}k et al.~\cite[Theorem 1]{MihalakSW16} show that $\#$Partition is
  \sharpp-complete.
  
  The \sharpp-completeness of $\#$Product-Partition follows by a reduction of
  Ng et al.~\cite[Theorem 1]{NgBCK10}, who give a reduction from Exact Cover by
  3-sets (X3C) to Product-Partition. We note that this reduction is weakly
  parsimonious, as defined by Hunt et al.~\cite[Definition 2.5]{HuntMRS98}. That
  is, for every solution of an X3C instance, there are exactly $2$ solutions for
  the constructed Product-Partition instance. Furthermore, Hunt et
  al.~\cite[implicit in Theorem 3.8]{HuntMRS98} show that $\#$X3C is
  $\sharpp$-hard. Therefore, the reduction of Ng et al.~\cite[Theorem
  1]{NgBCK10} can be used to give a Turing reduction from $\#$X3C to
  $\#$Product-Partition, which implies that $\#$Product-Partition is also
  $\sharpp$-hard under Turing reductions. It is easy to see that
  $\#$Product-Partition is in \sharpp.
\end{proof}

\begin{lemma}\label{lem:countlk}Let $k \in \rationals$. Then $\countp_{<k}[\ufvsa,\uregtrop]$ is
  \sharpp-hard under Turing reductions.
\end{lemma}
\begin{proof}
  We use the same idea as Mihal{\'{a}}k et al.~\cite[Theorem 1]{MihalakSW16} to
  encode $\#$Partition. Let $N = \{n_1,\ldots,n_n\}$ be an instance of
  $\#$Partition. Let $\doc = a^n$. We construct $A$ and $W$ such that every
  tuple $\tup \in \toSpanner{A}(\doc)$ corresponds to a partition of $N$.
  Furthermore, $\w(\doc,\tup) = k$ if and only if the partition encoded by
  $\tup$ is perfect.
  
  More formally, $A \eqdef (\alphabet, V, Q, q_0, Q_F, \delta)$, where $\alphabet
  \eqdef \{a\}, V \eqdef \{x_1,\ldots, x_n\}, Q \eqdef \{ q_i^j \mid 1 \leq i
  \leq n, 1 \leq j \leq 5\}$, where $q_i^5 = q_{i+1}^1$ for all $1 \leq i < n$,
  $q_0 \eqdef q_1^1, Q_F \eqdef \{ q_n^5 \}$, and for $1 \leq i \leq n$,
  $\delta$ is defined as follows:
  \[
    \delta(q_i^j, \sigma) \df
    \begin{cases}
      \{q_i^2\} & \text{ if } 1 \leq i \leq n,  \sigma = \vop{x_i}, \text{ and }
      j = 1\\
      \{q_i^3\} & \text{ if } 1 \leq i \leq n, \sigma = \vcl{x_i}, \text{ and }
      j = 2\\
      \{q_i^4\} & \text{ if } 1 \leq i \leq n, \sigma = a, \text{ and }
      j = 2\\
      \{q_i^5\} & \text{ if } 1 \leq i \leq n, \sigma = a, \text{ and }
      j = 3\\
      \{q_i^5\} & \text{ if } 1 \leq i \leq n, \sigma = \vcl{x_i}, \text{ and }
      j = 4 \;.
    \end{cases}
  \]
  Recall, that $q_i^5 = q_{i+1}^1$ for all $1 \leq i < n$. 

  Furthermore, we define the  weighted \vset-automaton $W$ encoding
  $\w$ the same way as $A$. That is, all transitions labeled by a variable
  operation $x \in \varop V$ are assigned weight $\srone$, $\delta(q_i^3,a,q_i^5) =
  n_i$ and $\delta(q_i^2,a,q_i^4) = -n_i$, the initial- and final weight
  functions:
  \begin{align*}
    I(q) &\df
    \begin{cases}
      \srone & \text{if } q = q_0\\
      \srzero & \text{otherwise}\;;
    \end{cases}\\
    F(q) &\df
    \begin{cases}
      k & \text{if } q \in Q_F\\
      \srzero & \text{otherwise}\;.
    \end{cases}
  \end{align*}
  We observe that every tuple $\tup \in \toSpanner{A}(\doc)$ encodes a partition
  of $N$, that is, $n_i \in N_1$ if $\tup(x_i) = \spanFromTo{i}{i}$
  and $n_i \in N_2$ if $\tup(x_i) = \spanFromTo{i}{i+1}$.
  Furthermore, for every tuple $\tup \in \toSpanner{A}(\doc)$, the weight
  $\w(\doc,\tup)$ is exactly $k$ plus the difference of the sum of all elements
  in $N_1$ and the sum of all elements in $N_2$. We make some observations about
  $A,\doc,$ and $\w$.
  \begin{enumerate}[label=(\arabic*)]
  \item The number of perfect partitions is exactly $\spcount_{=
      k}(\toSpanner{A}, \doc, \w)\;;$
  \item $\spcount_{<k}(\toSpanner{A},\doc,\w) = \spcount_{>k}(\toSpanner{A},\doc,\w)\;;$
  \item $\spcount(\toSpanner{A},\doc) = 2 \cdot\spcount_{<k}(\toSpanner{A},\doc,\w) + \spcount_{=
      k}(\toSpanner{A},\doc,\w)\;;$
  \item $\spcount(\toSpanner{A},\doc) = 2^{n+1}\;;$
  \item $\spcount_{= k}(\toSpanner{A},\doc,\w) = 2^{n+1} - 2\cdot\spcount_{<k}(\toSpanner{A},\doc,\w)\;.$
  \end{enumerate}
  Due to Observations $(1)$ and $(5)$ it follows that the number of
  perfect partitions can be computed by a single call to a
  $\spcount_{<k}(\toSpanner{A},\doc,\w)$-oracle. 

  It remains to argue that the observations $(1) - (5)$ hold. Observation $(1)$
  follows directly from the previous observation that the weight of each tuple
  is $k$ plus the difference of the sum of all elements in $N_1$ and the sum of
  all elements in $N_2$. Observation $(2)$ follows from the fact that the
  partition problem is symmetric, that is for every partition $N_1,N_2$ of $N$
  there is also a partition $N_2,N_1$. Observation $(3)$ follows from $(2)$, and
  $(4)$ from the fact that there are $2^n$ subsets of $N$ and therefore $2 \cdot
  2^n$ possible partitions. The last observation $(5)$ follows from $(3)$ and
  $(4)$. This concludes the proof.
\end{proof}

Along the same lines we show that $\countp_{<1}[\ufvsa,\uregnum]$ is
\sharpp-hard under Turing reductions. Note that we do not show hardness for
$\countp_{<k}[\ufvsa,\uregnum]$, but only for the case $k = 1$.\footnote{Recall
  that, in the proof for the tropical semiring, we add $k$ to all accepting runs
  by having $F(q) = k$, if $q\in Q_F$. This is not possible over the numerical
  semiring, as the multiplicative operation is the numerical multiplication
  $\cdot$ and not the numerical addition $+$.}

\begin{lemma}\label{lem:countlOneNum}$\countp_{<1}[\ufvsa,\uregnum]$ is \sharpp-hard under Turing reductions.
\end{lemma}
 \begin{proof}
   Let $N$ be an instance of $\#$Product-Partition. We construct $A,\doc,\w$ and
   $W$, as constructed in the proof of Lemma~\ref{lem:countlk}. However in $W$,
   $\delta(q_i^3,a,q_i^5) = n_i$ and $\delta(q_i^2,a,q_i^4) = \frac{1}{n_i}$.
   Observe, that $\w(\doc,\tup)$ is exactly the product of all elements in $N_1$
   divided by the product of all elements in $N_2$, where $n_i \in N_1$ if and
   only if $\tup(x_i) = \spanFromTo{i}{i}$ and $n_i \in N_2$ if and only if
   $\tup(x_i) = \spanFromTo{i}{i+1}$. Therefore, the number of perfect product
   partitions is exactly the number of tuples $\tup \in \toSpanner{A}(\doc)$
   with $\w(\doc,\tup) = 1$. Using the same argument as in the proof of
   Lemma~\ref{lem:countlk}, it follows that
   \[
     \#\text{Product-Partition} = 2^{n+1} - 2 \cdot
     \spcount_{<1}(\toSpanner{A},\doc,\w)\;, 
   \]
   and thus, $\#$Product-Partition can be computed by a single
   $\countp_{<1}[\ufvsa,\uregnum]$-oracle call.
\end{proof}

The following corollary follows directly from Lemmas~\ref{lem:countlk}
and~\ref{lem:countlOneNum}.

\begin{corollary}\label{cor:countlOne}$\countp_{<1}[\ufvsa,\UREG]$ is \sharpp-hard under Turing reductions.
\end{corollary}

We are finally ready to give the proof of Theorem~\ref{thm:quantUregHard}.
\begin{proof}[Proof of Theorem~\ref{thm:quantUregHard}]
  We show that $\spcount_{<1}(\toSpanner{A}, \doc, \w)$ can be computed in
  polynomial time, using a $\quantp[\ufvsa,\UREG]$-oracle therefore, concluding
  that the problem $\quantp[\ufvsa,\UREG]$ is also \sharpp-hard under Turing
  reductions.

  Let $A\in \ufvsa$, $\doc \in \docs$, and $\w\in\UREG$ represented by an
  unambiguous  weighted \vset-automaton $W$. Furthermore, let $0 < q <
  1$, such that $q = \frac{a}{b}$. Due to Theorem~\ref{thm:countResults}, $c
  \eqdef \spAcountagg$ can be computed in polynomial time. Let $0 \leq r \leq c
  \cdot (b-1)$. By Lemma~\ref{lem:spannerWithKtuples}\footnote{For instance with
    $v = \vars(A)\cdot b$.}, there are \vset-automata $A_r,A_r^\prime \in
  \ufvsa$ and documents $\doc_r,\doc_r^\prime$, such that
  $\spcount(\toSpanner{A_r},\doc_r) = r$ and
  $\spcount(\toSpanner{A_r^\prime},\doc_r^\prime) = c \cdot (b-1) - r$. Let
  $W_r$ (resp., $W_r^\prime$) be $A_r$ (resp., $A_r^\prime$), interpreted as
  unambiguous  weighted \vset-automaton, where all transitions of
  $A_r$ (resp., $A_R^\prime$) have weight $\srone$, the initial weight function
  assigns weight $\srone$ to the initial state of $A_r$ (resp., $A_r^\prime$),
  and the final weight function assigns weight $0$ (resp., $1$)\footnote{Note
    that we use $0$ and $1$ instead of $\srzero$ and $\srone$ on purpose. The
    reason is that we want to assign the same weights for both semirings.} to
  all accepting states of $A_r$ (resp., $A_r^\prime$). Slightly overloading
  notation, we define
  \[
    A^\prime \eqdef (A \cdot \doc_r \cdot \doc_r^\prime) \lor (\doc \cdot A_r \cdot
    \doc_r^\prime) \lor (\doc \cdot \doc_r \cdot A_r^\prime)\;
  \]
  and
  \[
    W^\prime \eqdef (W \cdot \doc_r \cdot \doc_r^\prime) \lor (\doc \cdot W_r \cdot
    \doc_r^\prime) \lor (\doc \cdot \doc_r \cdot W_r^\prime)\;
  \]
  It is straightforward to verify that both, $A^\prime$ and $W^\prime$ are
  unambiguous. Let $\doc^\prime = \doc \cdot \doc_r \cdot \doc_r^\prime$ and let
  $\w^\prime$ (resp, $\w_r, \w_r^\prime$) be the weight function, represented by
  $W^\prime$ (resp., $W_r,W_r^\prime$). It follows from the definition that
  \begin{align*}
    \spcount(\toSpanner{A^\prime},\doc^\prime)
    & = \spcount(\toSpanner{A},\doc) + \spcount(\toSpanner{A_r},\doc_r) + 
      \spcount(\toSpanner{A_r^\prime},\doc_r^\prime) \\
    & = c + r + (c\cdot(b-1)-r) = c\cdot b\;.
  \end{align*}
  Furthermore, recalling that $\w(\doc,\tup) = 0$ for all tuples $\tup \in
  \toSpanner{A_r}(\doc_r)$ and $\w(\doc,\tup) = 1$ for all tuples $\tup \in
  \toSpanner{A_r^\prime}(\doc_r^\prime)$, we have that
  \begin{align*}
    & \spcount_{<1}(\toSpanner{A^\prime}, \doc^\prime,\w^\prime)\\ 
    & = \spcount_{<1}(\toSpanner{A},\doc, \w) 
      + \spcount_{<1}(\toSpanner{A_r},\doc_r, \w_r)
      + \spcount_{<1}(\toSpanner{A_r^\prime},\doc_r^\prime, \w_r^\prime) \\
    & = \spcount_{<1}(\toSpanner{A},\doc, \w) + r + 0\;.
  \end{align*}
  
  Using binary search, we compute $r_{min}$ as the smallest $r$ with
  $\qquantl{q}(\toSpanner{A^\prime},\doc^\prime,w^\prime) < 1$. Thus,
  \[
    \frac{\spcount_{<1}(\toSpanner{A^\prime},\doc^\prime,w^\prime)}
    {\spcount(\toSpanner{A^\prime},\doc^\prime)}
    = \frac{\spcount_{<1}(\toSpanner{A},\doc, \w) + r_{min}}{c\cdot b} \geq q\;.
  \]
  For the sake of contradiction, assume that $\frac{\spcount_{<1}(\toSpanner{A},\doc, \w) +
    r_{min}}{c\cdot b} > q = \frac{c\cdot a}{c \cdot b}$. It follows that,
  $\spcount_{<1}(\toSpanner{A},\doc, \w) + r_{min} > c \cdot a$ and therefore, as all
  involved numbers are natural numbers, $\spcount_{<1}(\toSpanner{A},\doc, \w) + r_{min} - 1
  \geq c \cdot a$. Thus, $\frac{\spcount_{<1}(\toSpanner{A},\doc, \w) + (r_{min} - 1)
  }{c\cdot b} \geq q$, leading to the desired contradiction, as $r_{min}$ was
  assumed to be minimal.
  
  We have that $\frac{\spcount_{<1}(\toSpanner{A},\doc, \w) + r_{min}}{c\cdot b} =
  q = \frac{c \cdot a}{c \cdot b}$. It follows that
  \[
    \spcount_{<1}(\toSpanner{A},\doc, \w) = c \cdot a - r_{min}\;,
  \]
  which concludes the proof.
\end{proof}

 \section{Aggregate Approximation}
\label{approx}
Now that we have a detailed understanding on the complexity of computing exact
aggregates, we want to see in which cases the result can be approximated. We
only consider the situation where the exact problems are intractable and want to
understand when the considered aggregation problems can be approximated by fully
polynomial-time randomized approximation schemes (\fpras), and when the existence of
such an \fpras would contradict commonly believed conjectures, like $\rp \neq \np$
and the conjecture that the polynomial hierarchy does not collapse.

Based on the results for the computation of exact aggregates, we can already
give some insights into the possibility of approximation. That is,
Zuckerman~\cite{Zuckerman96} shows that $\#$SAT can not be approximated by an
\fpras unless $\np = \rp$. Furthermore, as shown by Dyer et
al.~\cite{DyerGGJ04}, this characterization extends to all problems which are
\sharpp-complete under parsimonious reductions. Therefore, due to
Theorems~\ref{thm:sw-sum-cnf}, and~\ref{thm:sumTropicalHard}, we have the
following corollary.

\begin{corollary}\label{cor:sumavgNoFpras}Unless $\np = \rp$, the problems
  $\sump[\fvsa,\SW]$, $\sump[\ufvsa,\regtrop]$, and $\avgp[\ufvsa,\regtrop]$ do not have an \fpras.
\end{corollary}

Arenas et at.~\cite[Corollary 3.3]{ArenasCJR19} showed that every function in
\spanl admits an \fpras. Therefore, due to Theorem~\ref{thm:sw-sum-spanl},
we have the following corollary.

\begin{corollary}\label{cor:sumSpanl}$\sump[\fvsa,\SW_{\nat}]$ has an \fpras.
\end{corollary}

In the remainder of this section, we will revisit the other intractable cases of
spanner aggregation and study whether or not approximation is possible.

\subsection{Approximation is Hard at First Sight}
We begin with some inapproximability results. For instance, as we show now, the
existence of an \fpras for the problems \minp, \maxp with \regnum weight
functions would imply a collapse of the polynomial hierarchy, even when spanners
are unambiguous. Furthermore, for \maxp and \regtrop weight functions the same
result holds.

\begin{theorem}\label{thm:minapx}$\minp[\ufvsa, \regnum]$ and $\maxp[\ufvsa, \regnum]$ do not have an \fpras, unless the polynomial hierarchy collapses to the
  second level.
\end{theorem}
\begin{proof}
  Assume there is an \fpras for $\minp[\ufvsa, \regnum]$. We will show that such
  an \fpras implies that the \np-complete problem SAT is in \bpp, which implies
  that the polynomial hierarchy collapses to the second level.\footnote{$\np
    \subseteq \bpp$ implies that $\ph \subseteq \bpp$ (cf.\
    Zachos~\cite{Zachos88}) and as $\bpp \subseteq (\pitwop\cap\sigmatwop)$
    (cf.\ Lautemann~\cite{Lautemann83}) the polynomial hierarchy collapses on
    the second level. Furthermore, as $\bpp$ is closed under complement, $\conp
    \subseteq \bpp$ implies that $\np \subseteq \bpp$ resulting in the same
    collapse of the polynomial hierarchy.}

  Let $\phi(x_1,\ldots,x_n)$ be a Boolean formula, given in CNF, and let $A,
  \doc,$ and $W'$ be as defined in the proof for $\maxp[\ufvsa,\regtrop]$ of
  Theorem~\ref{thm:minmaxRegintractable}, where $W'$ is interpreted as a
  weighted \vset-automaton over the numerical semiring. Observe that, due to
  $\srone = 1$ and $\srzero = 0$, it follows that
  $\repspnrk{W'}{\numerical}(\doc,\tup) \geq 1$ if the valuation $\alpha_\tup$
  encoded by $\tup$ does not satisfy at least one clause of $\phi$ and $0$
  otherwise. Let $\w$ be the weight function encoded by $W'$.

  For the sake of contradiction, assume that there is an \fpras for
  $\minp[\ufvsa, \regnum]$ and let $\delta = 0.4$. Assume that $\phi$ is
  satisfiable, thus $\spAminagg = 0$. Then the \fpras must return $0$ with
  probability at least $\frac{3}{4}.$ On the other hand, if $\phi$ is not
  satisfiable, the \fpras must return a value $x \geq (1-\delta) \cdot 1 = 0.6$
  with probability at least $\frac{3}{4}$. Consider the algorithm which calls
  the \fpras and accepts if the approximation is 0, and rejects otherwise. This
  algorithm is a \bpp algorithm for SAT, resulting in the desired contradiction.

  The proof for $\maxp[\ufvsa,\regnum]$ is analogous. The only difference is
  that the final weight function of $W'$ is multiplied by $-1$, that is, $W'$
  assigns weight $-x$ to each tuple, encoding a valuation $\alpha$ which
  does not satisfy $x$ clauses of $\phi$.  
\end{proof}

\begin{theorem}\label{thm:maxapx}$\maxp[\ufvsa, \regtrop]$ cannot be approximated by an \fpras, unless the
  polynomial hierarchy collapses to the second level.
\end{theorem}
\begin{proof}
  Let $\phi(x_1,\ldots,x_n)$ be a Boolean formula, given in CNF. We assume,
  \mbox{w.l.o.g.}, that the valuation which assigns false to all variables does
  not satisfy $\phi$. Let $A, \doc,$ and $\w$ be as defined in the proof for
  $\maxp[\ufvsa,\regtrop]$ in the proof of
  Theorem~\ref{thm:minmaxRegintractable}. Thus, $\spAmaxagg \geq 1$ if $\phi$ is
  satisfiable and $\spAmaxagg = 0$ if $\phi$ is not satisfiable.

  For the sake of contradiction, assume that there is an \fpras for
  $\maxp[\ufvsa, \regtrop]$ and let $\delta = 0.4$. Assume that $\phi$ is
  satisfiable, thus $\spAmaxagg \geq 1$. Then the \fpras must return a value $x
  \geq (1-\delta) \cdot 1 = 0.6$ with probability at least $\frac{3}{4}.$ On
  the other hand, if $\phi$ is not satisfiable, the \fpras must return $0$ with
  probability at least $\frac{3}{4}.$ Therefore, we can obtain a \bpp algorithm
  for SAT as follows. The algorithm first calls the \fpras, accepts if the
  approximation is bigger than 0, and rejects otherwise.
\end{proof}

Concerning \sump and \avgp the only case which is not resolved by
Corollary~\ref{cor:sumavgNoFpras} is the case of $\avgp[\fvsa,\SW]$. We show now
that, under reasonable complexity assumptions, this problem can also not be
approximated by an \fpras.

\begin{theorem}\label{thm:negSumAvg}$\avgp[\fvsa,\SW]$ cannot be approximated by an \fpras, unless the polynomial
  hierarchy collapses to the second level.
\end{theorem}
\begin{proof}
  We will show that such an \fpras implies that the \np-complete problem SAT is
  in \bpp, which implies that the polynomial hierarchy collapses to the second
  level.

  To this end, let $A,\doc$ and $\w$ be as constructed in the proof of
  Theorem~\ref{thm:sw-sum-cnf}. Recall that given a propositional formula
  $\phi$ in CNF, we have that $\spAsumagg = c$, where $c$ is the number of
  satisfying assignments of $\phi$.

  Assume there is an \fpras for $\avgp[\fvsa,\SW]$ and let $\delta = 0.5$. Assume
  that $\phi$ is not satisfiable. Then the \fpras on input $A, \doc, \w$ must
  return $0$ with probability at least $\frac{3}{4}.$ On the other hand, if
  $\phi$ is satisfiable, thus $c > 0$, the \fpras must return a value $x \geq
  (1-\delta) * \spavg(\toSpanner{A},\doc,\w) = \frac{1}{2} \cdot
  \frac{c}{\spAcountagg} > 0$, with probability at least $\frac{3}{4}$.
  Therefore, the algorithm which first approximates $\spavg(\toSpanner{A}, \doc,
  \w)$ with $\delta = 0.5$, rejects if the approximation is $0$ and accepts
  otherwise is a \bpp algorithm for SAT, implying that $\np \subseteq \bpp$,
  which implies that the polynomial hierarchy collapses to the second level.
\end{proof}

We now turn to the quantile problem. It turns out that this problem is difficult
to approximate even if the weight functions only return 0 or 1.

\begin{theorem}\label{thm:noFprasQuant}Let $0 < q < 1$. Then, $\quantp[\fvsa,\SW]$ cannot be approximated by an
  \fpras, unless the polynomial hierarchy collapses to the second level.
\end{theorem}
\begin{proof}
  We will show that an \fpras for $\quantp[\fvsa, \SW]$ implies a \bpp algorithm
  for SAT. To this end, let $\phi$ be a propositional formula $\phi$ in CNF.
  Assume that $q = \frac{1}{2}$ and let $A$ and $\doc$ be as constructed in the
  proof of Theorem~\ref{thm:sw-sum-cnf}. However, let $\w$ be the weight
  function which is represented by the $\rationals$-Relation $R$ over $\{x\}$
  with
  \[
    R(\doc) \eqdef
    \begin{cases}
      1 & \text{if } \doc = 1\\
      0 & \text{otherwise.}
    \end{cases} 
  \]
  Recall from the construction of $A$ and $\doc$ that $A$ is the union of two
  automata $A_1,A_{-1}$, such that $\spcount(\toSpanner{A_1},\doc) = 2^n$ and
  $\spcount(\toSpanner{A_{-1}},\doc) = s$, where $s$ is the number of
  non-satisfying assignments for $\phi$, furthermore, $\tup \in
  \toSpanner{A_1}(\doc)$ if and only if $\doc_{\tup(x)} = 1$ and $\tup \in
  \toSpanner{A_{-1}}(\doc)$ if and only if $\doc_{\tup(x)} = -1$. We observe
  that $R(-1) = 0$ and therefore, for every $\tup \in \toSpanner{A}(\doc)$ we
  have that
  \[
    \w(\doc,\tup) =
    \begin{cases}
      1 &\text{if } \tup \in \toSpanner{A_1}(\doc)\\
      0 &\text{if } \tup \in \toSpanner{A_{-1}}(\doc)\;.
    \end{cases}
  \]
  Thus, $\spAqquantagg{\frac{1}{2}} = 0$ if and only if $\phi$ is not
  satisfiable.

  Assuming there is an \fpras for $\quantp[\fvsa,\SW]$, one can decide SAT with a
  probability of $\frac{3}{4}$ by approximating $\spAquantagg$ with $\delta =
  0.5$, rejecting if the approximation is $0$ and accepting otherwise. This,
  however, implies that $\np \subseteq \bpp$, which implies a collapse of the
  polynomial hierarchy on the second level.

  The general case for $0 < q < 1$ follows by slightly adopting the previous
  construction. That is, assume that $q = \frac{a}{b}$. Due to
  $0 \leq q \leq 1$, it must hold that $1 \leq a < b$. We construct a
  \vset-automaton $A^\prime$ and a document $\doc^\prime$ as follows. Let
  $\sigma \notin \alphabet$ be a new alphabet symbol. The document $\doc^\prime$
  consists of $b$ copies of $\doc$, separated by $\sigma$ and $A^\prime$
  consists of $a$ copies of $A_{-1}$ and $b-a$ copies of $A_1$.  More formally,
  \[
    \doc^\prime \eqdef (\doc \cdot \sigma)^{b}\;.
  \]
  Furthermore, slightly abusing notation, we define
  \[
    A^\prime \eqdef (A_{-1} \cdot \sigma)^a\cdot (A_1
    \cdot \sigma)^{b-a}\;.
  \]
  We observe that on input document $\doc^\prime$, the automaton $A^\prime$
  accepts exactly $2^n\cdot (b-a)$ tuples $\tup$ with $\w(\doc^\prime,\tup) = 1$
  and $s \cdot a$ tuples with weight $0$. Therefore, $\spqquantagg{\frac{a}{b}}
  = 0$ if and only if
  \[
    \frac{s \cdot a}{2^n \cdot (b-a) + s \cdot a} \geq
    \frac{a}{b}\;.
  \]
  Solving this equation for $s$, it holds that $\spqquantagg{\frac{a}{b}} = 0$
  if and only if $s = 2^n$ and therefore $\spqquantagg{\frac{a}{b}} = 0$ if and
  only if $\phi$ is not satisfiable.

  The rest of the proof is analogous to the case that $q = \frac{1}{2}$.
\end{proof}

When the spanners are unambiguous, the simplest intractable case for \quantp is
the one with \UREG weight functions (see Table~\ref{tab:overview}). Again, we
can show that approximation is hard.
\begin{theorem}\label{thm:noUregFprasQuant}Let $0 < q < 1$. Then, $\quantp[\ufvsa,\uregtrop]$ cannot be approximated by
  an \fpras, unless the polynomial hierarchy collapses on the second level.
\end{theorem}
\begin{proof}
  We show that an \fpras for $\quantp[\ufvsa,\uregtrop]$ implies a \bpp
  algorithm for the \np-complete Partition problem. To this end, let $S = \{s_1,
  \ldots, s_n\}$ be a set of natural numbers. Furthermore, let $A, \doc, \w$ be
  constructed from $S$ as in the proof of Lemma~\ref{lem:countlk} with $k =
  0$.

  Per construction of $A,\doc$ and $\w$, every tuple $\tup \in
  \toSpanner{A}(\doc)$ corresponds to a partition of S, such that the partition
  is perfect if and only if $\w(\doc,\tup) = 0$. Furthermore, due to the
  partition problem being symmetrical, for every tuple $\tup \in
  \toSpanner{A}(\doc)$ with $\w(\doc,\tup) = k$ there is a tuple $\tup^\prime
  \in \toSpanner{A}(\doc)$ with $\w(\doc,\tup) = -k$. Thus,
  $\spAqquantagg{\frac{1}{2}} = 1$ if and only if there is a tuple $\tup \in
  \toSpanner{A}(\doc)$ with $\w(\doc,\tup) = 0$.

  Let $q = \frac{1}{2}$. Assuming there is an \fpras for
  $\quantp[\ufvsa,\uregtrop]$, one can decide Partition with a probability of
  $\frac{3}{4}$ by approximating $\spquant(\toSpanner{A}, \doc, \w)$ with
  $\delta = 0.5$, accepting if the approximation is $0$ rejecting otherwise.
  This implies that the algorithm accepts if and only if there is a perfect
  partition and therefore, $\np \subseteq \bpp$, which implies a collapse of the
  polynomial hierarchy on the second level.

  For the general case, assume that $q = \frac{a}{b}$. We observe that due to $0
  < q < 1$, it must hold that $a < b$. By Observation $(4)$ in the proof of
  Lemma~\ref{lem:countlk}, $\spcount(\toSpanner{A},\doc) = 2^{n+1}$. As in the
  proof of Theorem~\ref{thm:quantUregHard}, we construct a \vset-automaton
  $A^\prime$, a document $\doc^\prime$ and a weight function $\w^\prime$,
  represented by the weighted automaton $W^\prime \in \uregtrop$ , such that
  $\spquant(A^\prime,\doc^\prime,\w^\prime) = 0$ if and only if $S$ has a
  perfect partition. By Lemma~\ref{lem:spannerWithKtuples}, there are
  \vset-automata $A_{-1},A_1 \in \ufvsa$ and documents $\doc_{-1},\doc_1 \in
  \docs$ such that $\spcount(\toSpanner{A_{-1}},\doc_{-1}) = (a-1) \cdot 2^n$
  and $\spcount(\toSpanner{A_{1}},\doc_{1}) = (b-a-1) \cdot 2^n$. Let $W_{-1}$
  (resp., $W_1$) be the same as $A_{-1}$ (resp., $A_1$) interpreted as weighted
  automaton over the tropical semiring, such that all transitions are assigned
  weight $0$ and the final weight function assigns weight $-1$ (resp., $1$) to
  all accepting states. Let $\w_{-1}$ (resp., $\w_1$) be the weight function,
  represented by $W_{-1}$ (resp., $W_1$) Thus, $\w_{-1}(\doc_{-1},\tup) = -1$ if and
  only if $\tup \in \toSpanner{A_{-1}}(\doc_{-1})$ and $\w_{1}(\doc_{1},\tup) = 1$ if and
  only if $\tup \in \toSpanner{A_{1}}(\doc_{1})$. Let $\sigma$ be a new alphabet
  symbol. We construct
  $A^\prime,\doc^\prime,$ and $W^\prime$ as follows.
  \begin{align*}
    \doc^\prime &= \doc_{-1} \cdot \sigma \cdot \doc \cdot \sigma \cdot \doc_{1}\\
    A^\prime &= (A_{-1} \cdot \sigma \cdot \doc \cdot \sigma \cdot \doc_1) \lor (\doc_{-1} \cdot \sigma \cdot A \cdot \sigma \cdot \doc_1) \lor (\doc_{-1} \cdot \sigma \cdot \doc \cdot \sigma \cdot A_1)\\
    W^\prime &= (W_{-1} \cdot \sigma \cdot \doc \cdot \sigma \cdot \doc_1) \lor (\doc_{-1} \cdot \sigma \cdot W \cdot \sigma \cdot \doc_1) \lor (\doc_{-1} \cdot \sigma \cdot \doc \cdot \sigma \cdot W_1)\;.
  \end{align*}
  Furthermore, let $\w^\prime$ be the weight function, represented by
  $W^\prime$. It follows that
  \begin{align*}
    \spcount_{< 0}(\toSpanner{A^\prime},\doc^\prime,\w^\prime)
    &= (a-1)\cdot 2^n + \spcount_{< 0}(\toSpanner{A},\doc,\w)\\
    \spcount_{\leq 0}(\toSpanner{A^\prime},\doc^\prime,\w^\prime)
    &= (a-1)\cdot 2^n + \spcount_{\leq 0}(\toSpanner{A},\doc,\w)\\
    \spcount(\toSpanner{A^\prime},\doc^\prime) 
    &= (a-1)\cdot 2^n + 2 \cdot 2^n + (b-a-1)\cdot 2^n = b \cdot 2^n\;.
  \end{align*}
  We make a case distinction on $S$. If $S$ has a perfect partition,
  $\spcount_{<0}(\toSpanner{A},\doc,\w) < 2^n$ and $\spcount_{\leq
    0}(\toSpanner{A},\doc,\w) \geq 2^n$. Thus,
  $\spquant(A^\prime,\doc^\prime,\w^\prime) = 0$. Otherwise, if $S$ has no 
  perfect partition, $\spcount_{<0}(\toSpanner{A},\doc,\w) = 2^n$ and therefore
  $\spquant(A^\prime,\doc^\prime,\w^\prime) < 0$. Therefore,
  $\spquant(A^\prime,\doc^\prime,\w^\prime) = 0$ if and only if $S$ has a
  perfect partition. This concludes the proof.
\end{proof}

We note that the case of approximating $\quantp[\ufvsa,\uregnum]$ does not
follow analogous to the proof for $\quantp[\ufvsa,\uregtrop]$. The main reason
is the fact that $\#$Partition can be encoded into a weight function automaton
$\w_\tropical \in \uregtrop$, such that perfect partitions correspond to tuples
with weight $0$, whereas $\#$Product-Partition is encoded into a weight
function $\w_\numerical \in \uregnum$, such that perfect product partitions
correspond to tuples with weight $1$. Furthermore, all weights assigned by
$\w_\tropical$ are integers, whereas $\w_\numerical$ assigns rational numbers.
Therefore it is not obvious whether or not $\quantp[\ufvsa,\uregnum]$ can
be approximated by an \fpras. This case is left open for future research.

\subsection{When an FPRAS is Possible}
We show that Theorem~\ref{thm:negSumAvg} is very much on the intractability
frontier: it shows that approximation is intractable if weight functions can
assign $1$ and $-1$. On the other hand, if the weight functions are restricted
to \emph{nonnegative} numbers, then approximating \sump and \avgp is possible
with an \fpras.

\begin{theorem}\label{thm:fprasSumAvg}$\sump[\fvsa,\SW_{\Q_+}]$ and $\avgp[\fvsa,\SW_{\Q_+}]$ can be
  approximated by an \fpras.
\end{theorem}
\begin{proof}
  From Corollary~\ref{cor:sumSpanl} and Theorem~\ref{thm:countResults} we conclude that there is an
  \fpras for each of the problems $\sump[\fvsa,\SW_{\nat}]$ and $\spcount[\fvsa]$. We
  will use these \fpras to give an \fpras for $\sump[\fvsa,\SW_{\Q^+}]$ and
  $\avgp[\fvsa,\SW_{\Q^+}]$.
  
  In the following, we will denote an \fpras approximation with error rate
  $\delta$ of the problem $\spAcountagg$ (resp., $\spAsumagg$ and $\spAavgagg$)
  by $\spAcountaggapx$ (resp., $\spAsumaggapx$ and $\spAavgaggapx$).

  We begin by showing that $\sump[\fvsa,\SW_{\Q^+}]$ admits an \fpras. Let
  $A\in\fvsa$ be a \vset-automaton, $\doc \in \docs$ be a document, and $\w \in
  \SW_{\Q^+}$ be a weight function. Recall that every weight $x \in \Q^+$ is
  encoded by its numerator and its denominator. Let $D$ be the set of
  denominators used by $\w$ and let $\text{lcm}$ be the least common multiple of
  all elements in $D$. We note that, as argued in the proof of
  Theorem~\ref{thm:regUpperBounds}, $\text{lcm}$ can be computed in polynomial
  time. Let $\w_\nat(\doc,\tup) = \w(\doc,\tup) \cdot \text{lcm}$. Per
  definition of $\text{lcm}$, $\w_\nat \in \SW_\nat$ only assigns natural
  numbers. Furthermore, $\w(\doc,\tup) = \frac{\w_\nat(\doc,\tup)}{\text{lcm}}$.
  It follows that $\spAsumaggapx \eqdef
  \frac{\spsum(\toSpanner{A},\doc,\w_\nat,\delta)}{\text{lcm}}$ is an
  $\delta$-approximation of $\spsumagg$ with success probability $\frac{3}{4}$,
  concluding this part of the proof.
  
  It remains to show that $\avgp[\fvsa,\SW_{\Q_+}]$ admits an \fpras. We show
  that the algorithm which, with success rate $(\frac{3}{4})^{0.5}$, calculates
  a $\frac{\delta}{3}$-approximations for $\spcount$ and $\spsum$, and then
  returns the quotient of the results, is an \fpras for the problem
  $\avgp[\fvsa,\SW_{\Q_+}]$. We note that the probability that both
  approximations are successful is $(\frac{3}{4})^{0.5} \cdot
  (\frac{3}{4})^{0.5} = \frac{3}{4}$.
  
  It remains to show that the quotient of both results, $\spAavgaggapx \eqdef
  \frac{\spsum(\toSpanner{A},\doc,\w,\frac{\delta}{3})}{\spcount(\toSpanner{A},
    \doc, \frac{\delta}{3})}$, is indeed a $\delta$-approximation of
  $\spAavgagg$. Formally, we have to show that
  \[
    (1-\delta) \cdot \spavgagg \leq \spAavgaggapx \leq (1+\delta)\cdot
    \spAavgagg\;.
  \]
  We begin with the first inequality:
  \begin{align*}
    \spAavgaggapx
    &= \frac{\spsum(\toSpanner{A},\doc,\w,\frac{\delta}{3})}{\spcount(\toSpanner{A},\doc,\frac{\delta}{3})} 
    \geq \frac{(1-\frac{\delta}{3})\cdot \spAsumagg}{(1+\frac{\delta}{3})\cdot\spAcountagg} \\
    &= \frac{1-\frac{\delta}{3}}{1+\frac{\delta}{3}} \cdot \frac{\spAsumagg}{\spAcountagg} 
    \geq (1-\delta) \cdot \spAavgagg\;.
  \end{align*}
  It is straightforward to verify that
  $\frac{1-\frac{\delta}{3}}{1+\frac{\delta}{3}} \geq (1-\delta)$ holds for every
  $0 \leq \delta \leq 1$.
  The second inequality follows analogously:
  \begin{align*}
    \spAavgaggapx
    &= \frac{\spsum(\toSpanner{A},\doc,\w,\frac{\delta}{3})}{\spcount(\toSpanner{A},\doc,\frac{\delta}{3})}  
    \leq \frac{(1+\frac{\delta}{3})\cdot \spAsumagg}{(1-\frac{\delta}{3})\cdot\spAcountagg} \\
    &= \frac{1+\frac{\delta}{3}}{1-\frac{\delta}{3}} \cdot \frac{\spAsumagg}{\spAcountagg}
    \leq (1+\delta) \cdot \spAavgagg\;.
  \end{align*}
  Again, it is straightforward to verify that
  $\frac{1+\frac{\delta}{3}}{1-\frac{\delta}{3}} \leq (1+\delta)$ holds for every
  $0 \leq \delta \leq 1$.
\end{proof}

Our second positive result is about approximating quantiles \emph{in a
  positional manner}. Let \doc be a document, \spanner be a
document spanner, \w be a weight function and $0 \leq q \leq 1$ with $q \in
\rationals$. Then, for $\delta > 0$, we say that $k \in \rationals$ is a
positional $\delta$-approximation of $\spquantagg$ if there is a $q^\prime \in
\rationals,$ with $q-\delta \leq q^\prime \leq q + \delta$ and $k =
\spqquantagg{q^\prime}$.\footnote{The idea of positional quantile approximations
  was originally introduced by Manku et al.~\cite{MankuRL98} in the context of
  quantile computations with limited memory.}

\begin{lemma}[Hoeffding's Inequality]Let $X_1,\ldots,X_n$ be independent random variables with $0 \leq X_i \leq 1$
  for $1 \leq i \leq n$. Let $X = \alphabet_{i=1}^n X_i$ and let $\mathrm{EX}$
  denote the expectation of $X$. Then, for any $\lambda > 0$,
  $\propab{X-\mathrm{EX} \geq \lambda} \leq e^{\frac{-2\lambda^2}{n}}$.
\end{lemma}

\begin{algorithm}[t]
  \DontPrintSemicolon
	\KwIn{$A \in \fvsa, \doc \in \docs, \w \in \PW, 0 \leq q \leq 1, 0 \leq \delta
    \leq 1$}
  \KwOut{A positional $\delta$-approximation of $\spAqquantagg{q}$ with success
    rate $\frac{3}{4}$.}
  $W \gets \multiset{\cdot}$\;
  \For{$1 \leq i \leq 4\cdot \lceil \frac{\ln(16)}{2\delta^2} \rceil$}{
    $\tup \gets$ Sample($A,\doc,\frac{\delta}{3}$)\;
    Add $\w(\doc,\tup)$ to $W$\;
  }
  \If{$|W| < \lceil \frac{\ln(16)}{2\delta^2} \rceil$}{
    \textbf{Fail}\Comment*[r]{\textrm{Sample size too
        small}}\label{alg:posQuantilapxFail}
  }
  \textbf{Return} \qquantl{q}(W)\;
  \caption{PositionalQuantileApprox($A,\doc,\w,q,\delta$)}\label{alg:posQuantilapx}
\end{algorithm}

\begin{theorem}\label{thm:epsQuantApproximation}Let $0 \leq q \leq 1$. There is a probabilistic
  algorithm that calculates a positional $\delta$-approximation of
  $\quantp[\fvsa,\PW]$ with success probability at least $\frac{3}{4}$.
  Furthermore, the run time of the algorithm is polynomial in the input and
  $\frac{1}{\delta}.$
\end{theorem}
\begin{proof}

  Let $A \in \fvsa$ be a functional \vset-automaton and $\doc \in \docs$ be a
  document. Arenas et al.~\cite[Corollary 4.1]{ArenasCJR19} showed that given a
  functional \vset-automaton, one can sample tuples $\tup \in
  \toSpanner{A}(\doc)$ uniformly at random with success probability $\geq
  \frac{1}{2}$.\footnote{We note that the sampling algorithm by Arenas et
    al.~\cite[Corollary 4.1]{ArenasCJR19} detects and reports failures.} We will
  use this sampling algorithm to first create a sample of the assigned weights
  and then return the $\qquantl{q}$ of this sample. The algorithm is depicted in
  Algorithm~\ref{alg:posQuantilapx}.

  We note that this algorithm has two points of failure. On one hand, it can
  happen that less then $s \eqdef \lceil \frac{\ln(16)}{2\delta^2} \rceil$ calls
  to the sampling algorithm of Arenas et al.~\cite{ArenasCJR19} are successful.
  On the other hand, it can happen that the returned quantile is no positional
  $\delta$-approximation of the quantile. We show that both of these points of
  failure have a probability of less than $\frac{1}{8}$. Thus, the probability
  that the whole algorithm is successful is $\frac{7}{8} \cdot \frac{7}{8} >
  \frac{3}{4}$. We will first show that Line~\ref{alg:posQuantilapxFail} is
  reached with probability less than $\frac{1}{8}$.

  The success probability of each call to the sampling algorithm of Arenas et
  al.~\cite{ArenasCJR19} is at least $\frac{1}{2}$. Thus, the expected number of
  samples, generated by $4s$ consecutive calls to the algorithm is at least
  $2s$. Using Hoeffding's Inequality, the probability that $4s$
  consecutive calls to the sampling algorithm yield less than $s$ samples is
  less than $e^{-s}$ and therefore less than $\frac{1}{8}$ for every $s \geq
  3$.\footnote{Obviously, we can call the sampling algorithm 16 times for $s =
    1$ and $s = 2$ to ensure a failure rate of less than $\frac{1}{8}$.}

  It remains to show that a total of $s$ samples is enough to guarantee that the
  $\qquantl{q}$ of $W$ is a positional $\delta$-approximation of $\spAquantagg$
  with probability at least $\frac{7}{8}$.
  
  Let $w_{q-\delta} = \spAqquantagg{(q-\delta)}$ and $w_{q+\delta} =
  \spAqquantagg{(q+\delta)}$. Furthermore, let $W_{q-\delta} = \multiset{x \in W
    \mid x < w_{q-\delta}}$ and $W_{q+\delta} = \multiset{x \in W \mid x >
    w_{q+\delta}}$. We say that a sample is bad, if either $|W_{q-\delta}| \geq
  q \cdot s$ or $|W_{q+\delta}| \geq (1-q) \cdot s$. We will first show that the
  probability that $|W_{q-\delta}| \geq q \cdot s$ is at most
  $e^{-2\delta^2\cdot s}$. For each element $x \in W$ the probability that $x
  \in W_{q-\delta}$ is at most $(q-\delta)$. Thus, the expected size of
  $W_{q-\delta}$ is $(q-\delta) \cdot s$. Using Hoeffding's Inequality, with
  $\lambda = \delta\cdot s$ the probability that $|W_{q-\delta}| \geq q\cdot s$
  is at most $e^{-2\delta^2\cdot s}$. On the other hand, the for each element $x
  \in W$ the probability that $x \in W_{q+\delta}$ is at most $(1-(q+\delta)) =
  1-q-\delta$. Thus, the expected size of $W_{q+\delta}$ is $(1-q-\delta) \cdot
  s$. Again, using Hoeffding's Inequality, with $\lambda = \delta \cdot s$ the
  probability that $|W_{q+\delta}| \geq (1-q) \cdot s$ is at most
  $e^{-2\delta^2\cdot s}$. Therefore, the probability for a bad sample is at
  most $2\cdot e^{-2\delta^2\cdot s}$. Due to $s = \lceil
  \frac{\ln(32)}{2\delta^2}\rceil$, the probability of a bad sample is at most
  $\frac{1}{8}$, concluding the proof.
\end{proof}

 \section{Conclusions}
\label{conclusions}
We investigated the computational complexity of common aggregate functions over
regular document spanners given as regex formulas and \vset-automata. While each
of the studied aggregate functions is intractable in the general case, there are
polynomial-time algorithms under certain general assumptions. These include the
assumption that the numerical value of the tuples is determined by a constant
number of variables, or that the spanner is represented as an (unambiguous)
\vset-automaton. Moreover, we established quite general tractability results
when randomized approximations (FPRAS) are possible. The upper bounds that we
obtained for general (functional) \vset-automata immediately generalize to
aggregate functions over queries that involve relational-algebra operators and
string-equality conditions on top of spanners, whenever these inner queries can
be \emph{efficiently} compiled into a single
\vset-automaton~\cite{FreydenbergerKP18, PeterfreundFKK19}. Moreover, these
upper bounds immediately generalize to allow for \emph{grouping} (i.e., the
GROUP BY operator) by computing the tuples of the grouping variables and
applying the algorithms to each group separately.

We identified several interesting cases where the computation of
$\alpha(\spanner(\doc))$ can avoid the materialization of the exponentially
large set $\spanner(\doc)$, where, $\doc$ is the document, $\spanner$ is the
spanner, and $\alpha$ is the aggregate function. Notably, this is the case (1)
for \minp with general \vset-spanners and weight functions in \regtrop, \UREG,
and \SW, (2) for \maxp with general \vset-spanners and weight functions in \UREG
and \SW, (3) for \sump and \avgp with \ufvsa-spanners and weight functions in
\regnum, \UREG and \SW, and (4) for \quantp with \ufvsa-spanners and \SW weight
functions.

Yet, several basic questions are left for future investigation. A natural next
step would be to seek additional useful assumptions that cast the aggregate
queries tractable: Can monotonicity properties of the numerical functions lead
to efficient algorithms in cases that are otherwise intractable? What are the
regex formulas that can be efficiently translated into unambiguous
\vset-automata (and, hence, allow to leverage the algorithms for such
\vset-automata)? Another important direction is to generalize the results in a
more abstract framework, such as the \emph{Functional Aggregate Queries}
(FAQ)~\cite{KhamisNR16}, in order to provide a uniform explanation of our
findings and encompass general families of aggregate functions rather than
specific ones. Finally, the practical side of our work remains to be studied:
How do we make our algorithms efficient in practice? How effective is the
sampling approach in terms of the balancing between accuracy and execution cost?
Can we accurately compute estimators of aggregate functions over (joins of)
spanners within the setting of \emph{online aggregation}~\cite{HaasH99,
  LiWYZ16}?

Some of our tractability results reduce the aggregation problems to path problems in DAGs. Since these DAGs are not prohibitively large, we believe that this approach may already be a valid basis for a concrete implementation. Testing empirically whether this is actually the case is, however, a topic for future work.

\section*{Acknowledgment}
The authors are grateful to Noa Bratman for participating in the initial efforts on the research reported in this manuscript. Furthermore, we thank the anonymous reviewers of ICDT 2021 and LMCS for many helpful
remarks.
This work was supported by the German-Israeli Foundation for Scientific Research and Development (GIF), grant I-1502-407.6/2019.
The work of Johannes Doleschal and Wim Martens was also supported by the Deutsche Forschungsgemeinschaft (DFG), grant 369116833.  
The work of Benny Kimelfeld was also supported by the Israel Science Foundation (ISF), grants 1295/15 and 768/19, and the DFG project 412400621 (DIP program).

\bibliographystyle{alphaurl}
\bibliography{references}

\end{document}